\documentclass[12pt,reqno,hidelinks]{amsart}
\usepackage{mathrsfs}
\usepackage{etoolbox}
\usepackage{amssymb}
\usepackage{hyperref}
\usepackage{graphicx}
\usepackage{color}
\usepackage{amsthm}
\usepackage{amsmath}
\usepackage{mathrsfs}
\usepackage{mathtools}
\usepackage{verbatim}
\usepackage{amsthm}
\usepackage{framed}
\usepackage{wasysym}
\usepackage{upgreek}
\usepackage{color}
\usepackage[dvipsnames]{xcolor}
\usepackage{tensor}
\usepackage{accents}
\usepackage{dsfont}
\usepackage{hyperref}
\usepackage{enumerate}
\usepackage[normalem]{ulem}
\usepackage{longtable}
\usepackage{mathtools}
\usepackage{graphicx}
\usepackage[square, comma, sort&compress, numbers]{natbib}

 \theoremstyle{definition}

 \numberwithin{equation}{section}

\makeatletter

\numberwithin{equation}{section}
\newcommand{\vertiii}[1]{{\left\vert\kern-0.25ex\left\vert\kern-0.25ex\left\vert #1 \right\vert\kern-0.25ex\right\vert\kern-0.25ex\right\vert}}
\newcommand{\vertiiii}[1]{{\left\vert\kern-0.25ex\left\vert\kern-0.25ex\left\vert\kern-0.25ex\left\vert #1 \right\vert\kern-0.25ex\right\vert\kern-0.25ex\right\vert\kern-0.25ex\right\vert}}

\newcommand{\Rmnum}[1]{\expandafter\@slowromancap\romannumeral #1@}
\makeatother
\newtheorem{theorem}{Theorem}[section]
\newtheorem{lemma}[theorem]{Lemma}

\newtheorem{proposition}[theorem]{Proposition}

\theoremstyle{definition}

\newtheorem{assumption}[theorem]{Assumption}

\newtheorem{remark}[theorem]{Remark}

\newcommand{\AND}{{\quad\text{and}\quad}}
\newcommand{\p}[1]{
	\begin{pmatrix}
		#1
	\end{pmatrix}
}
\newcommand{\longeq}{\scalebox{3}[1]{=}}

\DeclareMathOperator{\cir}{Circ}
\DeclareMathOperator{\til}{Tild}
\DeclareMathOperator{\diag}{diag}

\begin{document}
\title[Future stability of FLRW for a large class of perfect fluids]{Future stability of the FLRW spacetime for a large class of perfect fluids}
\author{Chao Liu$^\dag$}
\address[Chao Liu]{Center for Mathematical Sciences, Huazhong University of Science and Technology,
1037 Luoyu Road, Wuhan, Hubei Province, China; Beijing International Center for Mathematical Research (BICMR), Peking University, No.5 Yiheyuan Road Haidian District, Beijing, China. }
\email{chao.liu.math@foxmail.com}

\author{Changhua Wei$^\dag$}
\address[Changhua Wei]{Department of Mathematics, Zhejiang Sci-Tech University,
	Hangzhou,  310018,  China. }
\email{changhuawei1986@gmail.com. }
\thanks{$\dag$ The authors contribute equally to this article and should be regarded as co-first authors.}

\date{}
\maketitle
\allowdisplaybreaks

\begin{abstract}
	
	We establish the future non-linear stability
	of Friedmann-Lema\^{\i}tre-Robertson-Walker (FLRW) solutions to the Einstein--Euler equations of the universe filled with a large class of perfect fluids (the equations of state are allowed to be certain  nonlinear or linear types both). Several previous results as specific examples can be covered in the results of this article.  We emphasize that the future stability of FLRW metric for \textit{polytropic fluids} with positive cosmological constant has been a difficult problem and can not be directly generalized from the previous known results. Our result in this article has not only covered this difficult case for the polytropic fluids, but also unified more types of fluids in a same scheme of proofs.
	
	 \vspace{2mm}
	
	
	\vspace{2mm}
	
\end{abstract}

\setcounter{tocdepth}{6}


\section{Introduction}\label{section:1}
Cosmological observations predict that our universe is currently undergoing an accelerated expansion which is potentially achieved by various models. Candidates such as positive cosmological constant, quintessence of dark energy have been widely studied, for example, in \cite{Bamba2012,Benaoum2012,Bento2003,Debnath2004,GORINI2006,Heydari-Fard2007,Kamenshchik2000,Kahya2015,Pedram2008,Yang2014,Rendall2004a,Rendall2005,Rendall2006}. During these candidates, a well-known family of Friedmann-Lema\^{\i}tre-Robertson-Walker (FLRW) solutions are often used by cosmologists to model a fluid-filled, spatially homogeneous and isotropic universe. In mathematics, the future non-linear stability of perturbations
of FLRW solutions to the Einstein--Euler equations with a positive cosmological constant and a linear
equation of state $p = K\rho$ has been well studied. However, in reality, the equations of state of the fluids can not be precisely linear. A natural question arises: \textit{if the equation of state deviates from the linear one, what happens to the longtime behavior of perturbations of FLRW metrics?} Or, more generally,  \textit{how the equations of state of the perfect fluids influence the future non-linear stability
of FLRW solutions}. This question is attractive to us because if small deviations from the linear equation of state of fluids destroy the non-linear future stability of FLRW metrics, then the FLRW metric with positive cosmological constant and the linear model of equation of state of fluids is not decent to predict the future of the universe due to the instability of this model with respect to the equation of state of the filled fluids.
This article aim to solve the proposed question. However, our results can not answer above questions completely and it is very difficult to investigate this question directly, since the equation of state affects the system in very complicated ways. In fact, we attempt to investigate it by asking first \textit{what types of fluids can guarantee the validity} of the future non-linear stability of FLRW solutions to the Einstein--Euler system with a positive cosmological constant. In this article, we construct a large class of fluids from the mathematical point of view and demonstrate several common and frequently used fluids (the equations of state are allowed to be certain  nonlinear or linear types both) are in this class and prove that
these fluids do make sure the target stability holds. On the other hand, another direct motivation for us is to investigate the evolution of the general Chaplygin fluids and polytropic fluids
in accelerated expanding spacetime, which is not studied so far to the authors' knowledge. In fact, one would see the fluids proposed in this article include certain cases of general Chaplygin fluids and polytropic fluids. We emphasize that the future stability of FLRW metric for \textit{polytropic fluids} with positive cosmological constant has been a difficult problem and can not be directly generalized from the previous known results. Our result in this article has not only covered this difficult case for the polytropic fluids, but also unified more types of fluids in a same scheme of proofs.

The dimensionless Einstein--Euler system is given by
\begin{align}
\widetilde{G}^{\mu\nu}+\Lambda \widetilde{g}^{\mu\nu} = &\widetilde{T}^{\mu\nu},\label{e:Ein1} \\
\widetilde{\nabla}_{\mu}\widetilde{T}^{\mu\nu} = & 0, \label{e:Eu1}
\end{align}
where
$\widetilde{G}^{\mu\nu}=\widetilde{R}^{\,\mu\nu}-\frac{1}{2}\widetilde{R}\widetilde{g}^{\mu\nu}$ is the Einstein tensor of the metric
\begin{equation*}\label{1.3}
\widetilde{g}=\widetilde{g}_{\mu\nu}dx^{\mu}dx^{\nu},
\end{equation*}
and
\begin{equation*}\label{1.4}
\widetilde{T}^{\mu\nu}=(\rho+p)\widetilde{u}^{\mu}\widetilde{u}^{\nu}+p\widetilde{g}^{\mu\nu},
\end{equation*}
is the stress energy tensor of the perfect fluid. Here,  $\widetilde{R}_{\mu\nu},\,\widetilde{R}$ are the Ricci and scalar curvature of the metric $\widetilde{g}$ respectively, $\widetilde{\nabla}_{\mu}$ is the covariant derivative of $\widetilde{g}$, and  $\rho$, $p=p(\rho)$ denote the energy density and pressure of the perfect fluid, respectively. We require that $p(0)=0$ and $p(\rho)$ is analytic on a compact set $I_\rho\subset[0,+\infty]$.
$\widetilde{u}^\mu$ is the fluid four-velocity, which we assume is normalized by
\begin{equation}\label{e:nol1}
	\widetilde{g}_{\mu\nu}\widetilde{u}^{\mu}\widetilde{u}^{\nu}=-1.
\end{equation}
$\widetilde{g}^{\mu\nu}$ is the inverse of $\widetilde{g}_{\mu\nu}$ and
\begin{equation*}
	\widetilde{R}^{\mu\nu}=\widetilde{g}^{\alpha\mu}\widetilde{g}^{\beta\nu} \widetilde{R}_{\alpha\beta}.
\end{equation*}

A well-known cosmological model is  the family of Friedmann-Lema\^{\i}tre-Robertson-Walker (FLRW) solution to \eqref{e:Ein1}--\eqref{e:Eu1} representing a homogeneous, fluid filled universe that is undergoing accelerated expansion. We use $x^{i}$ ($i=1, 2, 3$) to denote the standard periodic coordinates on the
$3$-torus $\mathbb{T}^3$ and $\tau = x^{0}$ a time coordinate on the interval $(0, 1]$, then the FLRW solutions on the manifold
\begin{equation*}
	\mathfrak{M}= (0,1]\times \mathbb{T}^3
\end{equation*}
are defined by
\begin{align}
\widetilde{\eta}=&\frac{1}{\tau^{2}}\Bigl(-\frac{1}{\omega^{2}(\tau)}d\tau^{2}+\delta_{ij}dx^i d x^j\Bigr),\label{e:uncfmt}\\
\tilde{u} =&-\tau\omega \partial_\tau
\end{align}
and the corresponding density of the fluid $\bar{\rho}$ verifies the estimate (see \eqref{e:bgrhoest1} later)
\begin{equation}\label{e:bgdsty}
	\tau^4\bar{\rho}(1)\leq \bar{\rho}(\tau) \leq \tau^3\bar{\rho}(1),
\end{equation}
where $\omega(\tau) \in C^2([0,1])$ and the initial proper density $\bar{\rho}(1)$ can be freely specified.

\begin{remark}
	We emphasize that, as is pointed out at Remark $1.2$ in \cite{Liu2018}, the expression \eqref{e:uncfmt} of FLRW solutions is not the standard one because of the choice of the time coordinate which compactifies the time interval from infinity interval $[0, \infty)$ in the standard presentation to $(0, 1]$ in the
	coordinates used here. In order to recover the standard presentation from this expression,
	let us define a new time coordinate $t$ in term of $\tau\in (0,1]$\footnote{note the following identity implies
	\begin{equation*}
		\frac{dt}{d\tau}=-\frac{1}{\tau\omega(\tau)}.
\end{equation*}
},
	\begin{equation}\label{e:ttau}
		t:=\mathfrak{t} (1/\tau):=-\int^\tau_1\frac{1}{y \omega(y)} d y>0.
	\end{equation}
	It is evident that $\mathfrak{t}$ is a strictly increasing function of $1/\tau$ due to the strictly positive integrand. Therefore, there exists an inverse function $\mathfrak{t}^{-1}$, which we denote by $a(t)$, such that	
	\begin{equation}\label{e:ttau2}
		a(t):=\mathfrak{t}^{-1}(t)=\frac{1}{\tau}.
	\end{equation}
	According to our choice of time coordinate $\tau$, the future lies in the direction of decreasing $\tau$ and timelike infinity is
	located at $\tau = 0$.
	Transforming the time coordinate $\tau$ to $t$ via \eqref{e:ttau}--\eqref{e:ttau2}, the FLRW metric \eqref{e:uncfmt} recovers to the standard one
	\begin{equation*}
		ds^{2}=-dt^{2}+a^{2}(t)\delta_{ij}dx^i d x^j,
	\end{equation*}
	which can be found in a variety of references, for instance \cite{Wald2010}.
\end{remark}

\begin{remark}
	The fluid-four velocity $\widetilde{u}^\mu$ is assumed to be future oriented, which is equivalent
	to the condition
	\begin{equation*}
		\widetilde{u}^0<0.
	\end{equation*}
\end{remark}


%
%
%

Before stating the main theorem of this article, we fix notations and conventions first. A number of new variables and preliminary concepts are introduced as well.

\subsection{Notations} \label{S:notation}

\subsubsection{Indices and coordinates}
Unless stated otherwise, our indexing convention will be as follows: we use lower
case Latin letters, e.g. $i; j; k$, for spatial indices that run from $1$ to $3$, and lower case Greek letters, e.g. $\alpha, \beta, \gamma$; for spacetime indices that run from $0$ to $3$. We will follow the Einstein summation convention, that is, repeated lower and upper indices are implicitly summed over their ranges. We use $x^{i}$ ($i=1, 2, 3$) to denote the standard periodic coordinates on the
$3$-torus $\mathbb{T}^3$ and $\tau = x^{0}$ a time coordinate on the interval $(0, 1]$.


\subsubsection{Descriptions of background and perturbed manifolds}
Throughout this article, we use $\widetilde{g}$ and $\widetilde{\eta}$ to denote the original metric and original background FLRW metric respectively; we also use $g$ and $\eta$ to denote the conformal metric and the conformal background metric.
$\widetilde{\Gamma}$, $\tilde{\gamma}$, $\Gamma$ and $\gamma$ denote the Christoffel symbols with respect to $\widetilde{g}$, $\widetilde{\eta}$, $g$ and $\eta$, respectively, similar conventions are used for all kinds of the curvature tensors $\widetilde{R}$, $\widetilde{\mathcal{R}}$, $R$, $\mathcal{R}$.





\subsubsection{Derivatives} \label{s:deriv}
Partial derivatives with respect to coordinates $(x^\mu)=(\tau,x^i)$ will be denoted by $\partial_\mu = \partial/\partial x^\mu$ which are the partial derivatives of the conformal spacetime.  $\widetilde{\nabla}$ and $\nabla$ are the covariant derivatives of the original physical spacetime and the conformal spacetime, respectively. We use
$Du=(\partial_j u)$ and $\partial u = (\partial_\mu u)$ to denote the spatial and spacetime gradients, respectively.
$f^\prime(g)$, or simply $f^\prime$ if there is no confusion, will be used to denote $f^\prime(g):=df(g)/dg$.

Greek letters will also be used to denote multi-indices, e.g.
$\alpha = (\alpha_1,\alpha_2,\ldots,\alpha_n)\in \mathbb{Z}_{\geq 0}^n$, and we will employ the standard notation $D^\alpha = \partial_{1}^{\alpha_1} \partial_{2}^{\alpha_2}\cdots
\partial_{n}^{\alpha_n}$ for spatial partial derivatives. It will be clear from context whether a Greek letter stands for a spacetime coordinate index or a multi-index. Furthermore, we will use $D^k u = \{ D^\alpha u \,|\, |\alpha|=k\}$ to denote the collection of partial derivatives of order $k$. 

Given a vector-valued map $f(u)$, where $u$ is a vector, we use $D f$ and $D_u f$ interchangeably to denote the derivative with respect to the vector $u$, and use the standard notation
\begin{equation*}
D f(u)\cdot \delta u := \left.\frac{d}{dt}\right|_{t=0} f(u+t\delta u)
\end{equation*}
for the action of the linear operator $D f$ on the vector $\delta u$. For vector-valued maps $f(u,v)$ of two (or more)
variables, we use the notation $D_1 f$ and $D_u f$ interchangeably for the partial
derivative with respect to the first variable, i.e.
\begin{equation*}
D_u f(u,v)\cdot \delta u := \left.\frac{d}{dt}\right|_{t=0} f(u+t\delta u,v),
\end{equation*}
and a similar notation for the partial derivative with respect to the other variable.

\subsubsection{Function spaces}
For a function $u(\tau,x)$, we define the following standard Sobolev norms
\begin{align*}
\|u(\tau,x)\|_{L^{2}(\mathbb T^{n})}:=& \left(\int_{\mathbb T^{n}}|u(\tau,x)|^{2}d^nx\right)^{\frac{1}{2}},
\\
\|u(\tau,x)\|_{H^{k}(\mathbb T^{n})}:=& \sum_{|\alpha|=0}^{k}\|D^\alpha u(\tau,x)\|_{L^{2}(\mathbb{T}^n)},
\intertext{and}
\|u(\tau,x)\|_{L^{\infty}(\mathbb T^{n})}:=&\text{ess} \sup_{x\in\mathbb T^{n}}|u(\tau,x)|.
\end{align*}

\subsubsection{Remainder terms\label{remainder}}
In order to simplify the handling of remainder terms whose exact forms are not important, we will use, unless otherwise stated,
upper case calligraphic letters, e.g.,
$\mathcal{S}(\tau, \xi)$, $\mathcal{T}(\tau, \xi)$ and $\mathcal{U}(\tau, \xi)$, to denote vector-valued maps that are elements
of the space $C^1([0,1], C^\infty(\mathbb{R}^M))$ for $\xi\in \mathbb{R}^M$ and upper case letters in typewriter font, e.g.,
$\texttt{S}(\tau, \xi)$, $\texttt{T}(\tau, \xi)$ and $\texttt{U}(\tau, \xi)$, to denote vector-valued maps that are elements
of the space $C^0([0,1], C^\infty(\mathbb{R}^M))$. We also remark that because the exact forms of these callgraphic or typewriter fond remainders are not important, the remainders of the same letter may change from line to line.

We will say that a function $f(x,y)$ \textit{vanishes to the $n^{\text{th}}$ order in $y$} if it satisfies $f(x,y)\sim \mathrm{O}(y^n)$ as $y\rightarrow 0$, that is, there exists a positive constant $C$ such that $|f(x,y)|\leq C|y|^n$ as $y \rightarrow 0$.

\subsubsection{Intermediate point}\label{s:intpt}
We use $\alpha_{K_\ell}$ to denote the intermediate point between $\bar{\alpha}$ and $\alpha$ measured by $K_\ell$ in a linear fashion, that is,
\begin{equation*}
	\alpha_{K_\ell}:=\bar{\alpha}+K_\ell(\alpha-\bar{\alpha})
\end{equation*}
for some constant $K_\ell\in (0,1)$ ($\ell=1,2,\cdots$).

\subsection{Constraints on the perfect fluids}\label{S:Makinolike}
Throughout this article, we concentrate on a type of perfect fluids with the equation of state $p=p(\rho)$ where $p(0)=0$ and $p(\rho)$ is analytic on a compact set $I_\rho\subset[0,+\infty]$,  and we also require that for this fluid,
there exists an invertible transformation $\alpha=\mu^{-1}(\rho)$ determined by a set of quantities $\{\mu(\alpha), \varrho, \varsigma, \beta(\tau)\}$ satisfying the following Assumptions \ref{a:Maksym}--\ref{a:postvty}, which rephrase this fluid in terms of a new density variable $\alpha$, and $\{\varrho, \varsigma, \beta(\tau)\}$ restrict the property of the transformation $\mu$.

\begin{assumption}\label{a:Maksym} \emph{(The symmetrization condition)}
	There exists an invertible transformation
	\begin{eqnarray*}
		C^2 \ni \mu: \qquad I_\alpha &\rightarrow & I_\rho,  \\
		\alpha(x^\mu)  &\mapsto &  \rho(x^\mu),
	\end{eqnarray*}
	that is, $\mu(\alpha)=\rho$, where $I_\alpha\subset [-\infty, +\infty]$ is a compact set, such that a quantity $\lambda(\alpha)$ constructed by $\mu$ and the pressure of the fluid $p$ is uniformly bounded and away from $0$. That is, there is a constant
	\begin{equation}\label{e:hdel}
	    \hat{\delta}\in \Bigl(0,  \min \Bigl\{ \frac{3}{4}\Bigl(1+\sqrt{\frac{3}{\Lambda}}\Bigr) ,  \frac{ \Lambda}{3+\Lambda}\Bigl(1+\sqrt{\frac{3}{\Lambda}}\Bigr), 1 \Bigr\} \Bigr),
	\end{equation}
	such that
	\begin{equation}\label{e:Maksym}
	\lambda(\alpha) :=\frac{s(\alpha)}{ \mu(\alpha)+\mu^* p (\alpha) }\frac{d\mu(\alpha)}{d\alpha} \in [\hat{\delta},1/\hat{\delta}]
	\end{equation}
	and $\lambda\in C^2(I_\alpha)$,
	where
	\begin{equation*}
	s :=\mu^* c_s  \AND c_s :=\sqrt{\frac{ d p }{d\rho} } \label{e:qdef}
	\end{equation*}
	describe the sound speed and $\mu^*$ is the pullback of $\mu$. We further assume $s\in C^2(I_\alpha)$.
\end{assumption}

\begin{assumption} \label{a:tLip}
	Suppose $\bar{\rho}=\bar{\rho}(\tau)$ is the density of the homogeneous, isotropic fluid (for example, in this article, we take it to be the background FLRW solution \eqref{e:uncfmt}--\eqref{e:bgdsty}), then we denote
	\begin{equation*}\label{e:abdef}
	\bar{\alpha}:=\mu^{-1}(\bar{\rho}).
	\end{equation*}
	Assume there exists a function $\varrho\in C\bigl([0,1], C^\infty( \mathbb{R})\bigr)$ satisfying $\varrho(\tau,0)=0$ and a rescaling function $\beta(\tau)\in C([0,1]) \cap C^1((0,1])$ of $\alpha$
	such that
	\begin{align}\label{e:rhopro}
	\mu(\alpha)-\mu(\bar{\alpha})=\tau^{\varsigma}  \varrho\bigl(\tau,  \beta^{-1}(\tau)(\alpha- \bar{\alpha})\bigr) ,\quad \varsigma \geq 2.
	\end{align}
\end{assumption}

\begin{assumption}\label{a:postvty}
	Let $q(\alpha):= s(\alpha)/\lambda(\alpha) $, then $\bar{q}:=q(\bar{\alpha})=\bar{s}/\bar{\lambda}$ where $\bar{s}:=s(\bar{\alpha})$,  $\bar{\lambda}:=\lambda(\bar{\alpha})$. Suppose
	\begin{equation}
		\bar{s} \lesssim \beta(\tau), \qquad  \lambda^\prime (\bar{\alpha}) \partial_\tau \beta(\tau)
		\lesssim 1 \AND  \frac{\bar{s}}{\tau}  \lambda^\prime (\bar{\alpha})
		\lesssim 1, \label{e:B0bd}
	\end{equation}
	and one of the following two conditions holds,
	\begin{enumerate}
		\item If there is a positive constant $\hat{\delta}$ given by \eqref{e:hdel}, such that $\beta(\tau)$ is bounded by
		\begin{align}\label{e:btest}
		 \frac{1}{C^*\hat{\delta}} \tau \leq \beta(\tau) \leq \frac{1}{C^*\hat{\delta}} \sqrt{\tau}
		\AND
		\chi(\tau):= \tau\partial_\tau \ln \beta(\tau)\geq 0,
		\end{align}
		where
		$	C^*:=\bigl(\sqrt{\frac{3}{\Lambda}}+1\bigr)^{-1}\bigl(\frac{9}{4\Lambda^2}+2\bigr)^{-\frac{1}{2}}$
		and $\chi(\tau)$ satisfies
		\begin{gather}
		1-3\bar{s}^2  \geq  \chi(\tau) +\hat{\delta}  \label{e:assbg2} 	
		\intertext{and}
		\frac{1}{3} \chi(\tau)  + \frac{1}{\hat{\delta}} \geq    q^\prime(\bar{\alpha})\geq  \frac{1}{3}  \chi(\tau)  +\hat{\delta}  \label{e:assbg1},
		\end{gather}
		for all $\tau\in[0,1]$.
		\item
		If $\beta\equiv $ constant$>0$, $s$ and one of the following cases happens
		\begin{enumerate}
			\item $q=\bar{q}$ and $ \hat{\delta}\leq 1-3\bar{s}^2\leq 1-3\hat{\delta}^2$ where $\hat{\delta}$ is given by \eqref{e:hdel};
			\item $q=\bar{q}$ and $ 1-3s^2 =0 $;
		\end{enumerate}
	for all $\tau\in[0,1]$.
	\end{enumerate}	
\end{assumption}

Furthermore, Assumption \ref{a:postvty}.$(1)$ and $(2)$ separate the perfect fluids into two classes according to their different proofs. We mention \textit{Fluids $(I)$} if Assumption \ref{a:postvty}.$(1)$ holds and the other one is \textit{Fluids $(II)$} which satisfy Assumption \ref{a:postvty}.$(2)$.

\begin{remark}
	With the help of that $\bar{s}^2\geq 0$, \eqref{e:btest} and \eqref{e:assbg2} imply
	\begin{equation}\label{e:chiest}
		0 \leq \chi(\tau) \leq 1-\hat{\delta}.
	\end{equation}
	In addition, \eqref{e:btest} and \eqref{e:assbg2} and Assumption \ref{a:postvty}.$(2)$, with the fact that $c_s(\rho)=s(\alpha)$, yield that
	\begin{equation}\label{e:srange}
		0 \leq \bar{s}^2 \leq \frac{1}{3} \AND 0 \leq \bar{c}_s^2:=c^2_s(\bar{\rho}) \leq \frac{1}{3}.
	\end{equation}
	By \eqref{e:btest} and \eqref{e:chiest}, we arrive at
	\begin{equation}\label{e:dtbt2}
		\bar{s} \partial_\tau \beta  \lesssim\beta \partial_\tau \beta  \lesssim \frac{\beta^2}{\tau} \lesssim 1 \AND \bar{s}\beta \lesssim \tau \quad (\text{i.e. }\bar{q}\beta \lesssim \tau)
	\end{equation}
	for $\tau\in[0,1]$.
\end{remark}

\begin{remark}
	We give \textit{three examples} in \S \ref{s:eplMF} which guarantee that the set of the fluids satisfying above assumptions is \textit{not empty}. These three examples are
	\begin{enumerate}
		\item \label{i:1} Fluids with the \textit{linear} equation of state  $p=K\rho\;(K\in(0,\frac{1}{3}])$. This is a standard fluid model in cosmology and the future stability has been well investigated in \cite{RodnianskiSpeck:2013,Speck2012,Oliynyk2016a,Liu2017,Liu2018}. Especially, \cite{Liu2017,Liu2018} also conclude the statements of cosmological Newtonian limits on large scale based on such a linear model of perfect fluids;
		\item \label{i:2}  \textit{Chaplygin gas} with the equation of state that
		 \begin{equation}
		 	p=-\frac{\Lambda^{1+\vartheta}}{(\rho+\Lambda)^{\vartheta}}+\Lambda \qquad \Bigl(\vartheta\in\Bigl(0,
		 	\sqrt{\frac{1}{3}}\Bigr)\Bigr) \label{e:Chapnew};
		 \end{equation}
		 is a revised version of Chaplygin gases for system \eqref{e:Ein1}--\eqref{e:Eu1} which is equivalent\footnote{In \cite{LeFloch2015a}, the authors studied the Einstein equation \begin{equation}\label{Ein-L}
		\tilde{G}^{\mu\nu}=\tilde{T}^{\mu\nu}\end{equation} with Chaplygin equation of state
			\begin{equation}\label{e:Chapori}
			\tilde{p}=-\frac{A}{\tilde{\rho}^{\vartheta}},
			\end{equation}
			where $\tilde{\rho}-\Lambda\geq0$. It is easy to check that
			equation \eqref{e:Ein1} is equivalent to \eqref{Ein-L} if we use the function transformation $\tilde{\rho}=\rho+\Lambda$ and
			 $A=\Lambda^{1+\vartheta}$. }  to the one studied in \cite{LeFloch2015a}.
The Chaplygin gas is a cosmological gas model for dark energy. In other words, it essentially plays the role of positive cosmological constant. \cite{LeFloch2015a} adopts the method of \cite{Oliynyk2016a} in deriving the long time stability of irrotational Chaplygin gas filled universe.
		 \item \label{i:3}  \textit{Polytropic gas} with  $p=K\rho^{\frac{n+1}{n}}\; (n\in(1,3))$. The polytropic gas is also a well-known fluid model in stellar system. However, there is no results known about polytropic fluids filled general relativistic universe to the best of our knowledge. In fact, this article, after proving the future stability of the fluids satisfying our assumptions, implies the corresponding results of polytropic ones.
	\end{enumerate}
 Hence, the fluid considered here is a decent model at least containing all these common models in astrophysics.
\end{remark}

\begin{remark}
	Assumption \ref{a:Maksym} is a generalization of the standard symmetrization condition of the transformation of the Makino-type variable (see, for example, \cite[Page $114$]{Brauer2014}).
\end{remark}

\begin{remark}
	Assumption \ref{a:tLip} implies that
\begin{align}\label{e:p-p}
\frac{\rho-\bar{\rho}}{\tau^2}= \tau^{\varsigma-2}\varrho\bigl(\tau, \beta^{-1}(\alpha-\bar{\alpha})\bigr) \AND \frac{p-\bar{p}}{\tau^2}= c_s^2\bigl(\rho_{K_4}\bigr)\tau^{\varsigma-2}\varrho\bigl(\tau, \beta^{-1}(\alpha-\bar{\alpha})\bigr),
\end{align}
which will be used in \eqref{e:t2rho} to make sure $\frac{\rho-\bar{\rho}}{\tau^2}$ and $\frac{p-\bar{p}}{\tau^2}$ are not $\tau$-singular terms in the remainder terms of the Einstein equation (see \eqref{e:t2rho} for details).
\end{remark}

\begin{remark}
Let us state an alternative but slightly stronger expression of 	Assumption \ref{a:tLip}.\eqref{e:rhopro} which is given by
	\begin{equation}
		|\mu'(\alpha)|\lesssim \frac{\tau^\varsigma}{\beta}, \qquad \varsigma\geq 2,  \label{e:111}
	\end{equation}
	for $\alpha\in I_\alpha$.  Noting the expansion of $\mu$,
	\begin{align*}
		\mu(\alpha)=\mu(\bar{\alpha})+\beta\mu'(\alpha_{K_8})\beta^{-1} (\alpha-\bar{\alpha})
	\end{align*}
	for some constant $K_8\in(0,1)$, we can, with the help of \eqref{e:111},  arrive at Assumption \ref{a:tLip}.\eqref{e:rhopro} readily.
\end{remark}

\begin{remark}
	Once we rewrite Einstein--Euler equations in terms of the singular symmetric hyperbolic formulations \eqref{e:model1} given in Appendix \ref{section:3.5}, then
	Assumption \ref{a:postvty}.$(1)$ and $(2)$ ensure the positivity of the matrix $\mathbf{B}$ of the singular term in the singular hyperbolic system \eqref{e:model1}. While Assumption \ref{a:postvty}.\eqref{e:B0bd} ensures that the coefficient matrix $B^0\in C^1([0,1],C^\infty(\mathbb{R}^M))$ in the model equation \eqref{e:model1}.
\end{remark}



\subsection{Conformal Einstein--Euler system} \label{s:cees}

The main tool of this article is  Oliynyk's conformal singular system which was first established by T. Oliynyk \cite{Oliynyk2016a}, then developed by \cite{Liu2017,Liu2018} to prove the cosmological Newtonian limits on large scales and applied by \cite{LeFloch2015a} to future stability of Chaplygin gas filled universe.
Instead of describing the spacetime in terms of physical metric  $\widetilde{g}$ directly, we turn to the conformal one
\begin{equation}\label{1.9}
g_{\mu\nu}=e^{-2\Phi}\widetilde{g}_{\mu\nu}, \quad (\text{i.e. } g^{\mu\nu}=e^{2\Phi}\widetilde{g}^{\mu\nu}),
\end{equation}
and the fluid four velocity governed by
\begin{equation}\label{4.29}
u^{\mu}=e^{\Phi}\tilde{u}^{\mu},
\end{equation}
where, throughout this article, we take explicitly
\begin{equation}\label{1.10}
\Phi=-\ln(\tau).
\end{equation}
In addition, by \eqref{e:nol1}, there is a normalization relation of the conformal four velocity,
\begin{equation}\label{e:nol2}
	u^\mu u_\mu=-1.
\end{equation}

Under the conformal transformation \eqref{1.9}--\eqref{1.10}, the conformal background metric becomes
\begin{equation*}\label{3.9}
\eta=-\frac{1}{\omega^{2}(\tau)}d\tau^{2}+\delta_{ij}dx^i d x^j,
\end{equation*}
and recalling identity
\begin{align*}
	\tilde{R}_{\mu\nu}-R_{\mu\nu}=-g_{\mu\nu}\Box \Phi - 2\nabla_\mu \nabla_\nu \Phi + 2(\nabla_\mu \nabla_\nu \Phi-|\nabla\Phi|^2_g g_{\mu\nu}),
\end{align*}
where $\nabla$ and $R_{\mu\nu}$ are the covariant derivative and the Ricci tensor of $g_{\mu\nu}$, respectively, $\Box=\nabla_\nu \nabla^\nu$ and $|\nabla\Phi|^2_g=g^{\mu\nu} \nabla_\mu \Phi \nabla_\nu \Phi$, and then
the equations \eqref{e:Ein1}--\eqref{e:Eu1} become \textit{conformal Einstein--Euler equations}, 
\begin{align}
G^{\mu\nu} =T^{\mu\nu}:=&e^{4\Phi}\widetilde{T}^{\mu\nu}-e^{2\Phi}\Lambda g^{\mu\nu}+2(\nabla^{\mu}\nabla^{\nu}\Phi-\nabla^{\mu}\Phi\nabla^{\nu}\Phi)-(2\Box_{g}\Phi+|\nabla\Phi|_{g}^{2})g^{\mu\nu},  \label{e:cEin1} \\
\nabla_{\mu}\widetilde{T}^{\mu\nu} =&-6\widetilde{T}^{\mu\nu}\nabla_{\mu}\Phi+g_{\kappa\lambda}\widetilde{T}^{\kappa\lambda}g^{\mu\nu}\nabla_{\mu}\Phi,  \label{e:cEu1}
\end{align}
where here and in the following, unless otherwise specified, we raise and lower all coordinate tensor indices using
the conformal metric $g_{\mu\nu}$. Note that \eqref{e:cEu1} can be derived due to the identity of the difference $\tilde{\Gamma}^\gamma_{\mu\nu}-
\Gamma^\gamma_{\mu\nu}$,
\begin{equation*} \label{Gammadif}
\tilde{\Gamma}^\gamma_{\mu\nu}-\Gamma^\gamma_{\mu\nu} =
g^{\gamma \alpha}\bigl(g_{\mu \alpha} \nabla_\nu\Phi
+ g_{\nu\alpha} \nabla_\mu\Phi - g_{\mu\nu} \nabla_\alpha\Phi \bigr).
\end{equation*}
Contracting the free indices of \eqref{e:cEin1} gives $R=4\Lambda-T$,
where $T=g_{\mu\nu}T^{\mu\nu}$ and $R$ is
the Ricci scalar of the conformal metric. Using this and the definition $G^{\mu\nu} =
R^{\mu\nu}-\frac{1}{2}R g^{\mu\nu}$ of the Einstein tensor, we can write \eqref{e:cEin1} as
\begin{align}\label{3.11}
-2R^{\mu\nu} =&-4\nabla^{\mu}\nabla^{\nu}\Phi+4\nabla^{\mu}\Phi\nabla^{\nu}\Phi
-2[\Box_{g}\Phi+2|\nabla\Phi|^{2}_{g} \notag \\
&+\Bigl(\frac{\rho-p}{2}+\Lambda\Bigr)g^{\mu\nu}
-2e^{2\Phi}\left(\rho+p\right)u^{\mu}u^{\nu},
\end{align}
which we call as the \textit{conformal Einstein equations}.

\subsection{The wave gauge}\label{section:3.1}
In order to rewrite the conformal Einstein equation into a hyperbolic system, wave gauge is a useful technique. In this section, we introduce the wave gauge in the spirit of \cite{Oliynyk2016a,Liu2017,Liu2018,LeFloch2015a} which is the conformal version of the original one used in \cite{Ringstroem2008,Speck2012,RodnianskiSpeck:2013}.
Straightforward calculations show that the Christoffel
symbols, contracted Christoffels and Ricci tensors are
\begin{align*}
\gamma^{\lambda}_{\mu\nu}=&\frac{1}{2}\eta^{00}(\partial_{\tau}\eta_{00})\delta_{0}^{\lambda}\delta_{\mu}^{0}\delta_{\nu}^{0}
=-\frac{\partial_{\tau}\omega}{\omega}\delta_{0}^{\lambda}\delta_{\mu}^{0}\delta_{\nu}^{0}, \quad \gamma^{\mu}=-\frac{\Omega}{\tau}\delta^{\mu}_{0}
\intertext{and}
\mathcal{R}_{\mu\nu}=&\partial_{\alpha}\gamma^{\alpha}_{\mu\nu}-\partial_{\mu}\gamma^{\alpha}_{\alpha \nu}+\gamma^{\alpha}_{\alpha \lambda}\gamma^{\lambda}_{\mu\nu}-\gamma^{\alpha}_{\mu\lambda}\gamma^{\alpha}_{\alpha \nu}=0.
\end{align*}
where for the simplicity of notations, we denote
\begin{equation}\label{e:Omeg}
\Omega(\tau) =-\tau \omega \partial_\tau \omega.
\end{equation}

Define the wave coordinates as
\begin{align}\label{E:WAVEGAUGE}
Z^\mu = 0,
\end{align}
where
\begin{equation} \label{Zdef}
Z^\mu = X^\mu+Y^\mu
\end{equation}
with
\begin{align}
X^\mu&:= \Gamma^\mu
- \gamma^\mu =- \partial_\nu g^{\mu\nu}
+\frac{1}{2} g^{\mu\sigma} g_{\alpha\beta} \partial_\sigma g^{\alpha\beta}
+\frac{\Omega}{\tau}\delta^\mu_0 \qquad  \bigl( \Gamma^\mu= g^{\sigma\nu} \Gamma^\mu_{\sigma\nu}\bigr)
\label{E:XY}
\intertext{and}
 Y^\mu&:= -2( g^{\mu\nu}-\eta^{\mu\nu}) \nabla_\nu\Phi= -2 g^{\mu\nu}\nabla_\nu\Phi+2\frac{\omega^2}{\tau} \delta^\mu_0. \label{Ydef}
\end{align}
Then,
\begin{equation}\label{4.5}
Z^{\mu}=X^{\mu}+Y^{\mu}=\Gamma^{\mu}+\frac{2}{\tau}\left(g^{\mu0}+\Bigl(\omega^{2}+\frac{\Omega}{2}\Bigr)
\delta^{\mu}_{0}\right)= \Gamma^{\mu}+\frac{2}{\tau}\left(g^{\mu0}+\frac{\psi(\tau)}{3}\delta^{\mu}_{0}\right)=0,
\end{equation}
where we, for simplicity, denote
\begin{equation}\label{e:lamtau}
	\frac{\psi(\tau)}{3}:=\omega^{2}+\frac{\Omega}{2}.
\end{equation}

\begin{remark}\label{remark:3.1}
	A well known result from Zengino\u{g}lu \cite{Zenginoglu2008} implies if initially $Z^{\mu}=0$, then $Z^{\mu}\equiv0$ in the whole evolution of the Einstein--Euler system.

\end{remark}

\subsection{Transformed field variables}

The gravitational and matter field variables, such as $\{g^{\mu\nu}(x^\alpha), \rho(x^\alpha), u^\mu(x^\alpha)\}$, as they stand for, are not suitable
for establishing the global existence of solutions.
In order to obtain suitable variables, we employ the following
field variables, which are adopted first by \cite{Oliynyk2016a} and then used in \cite{Liu2017,Liu2018,LeFloch2015a} via their variations.

Define the densitized three-metric
$
\textbf{g}^{ij}=\det(\check{g}_{lm})^{\frac{1}{3}}g^{ij},
$
where
$
\check{g}_{lm}=(g^{lm})^{-1},
$
and introduce the variable
\begin{equation}\label{e:q}
\textbf{q}=g^{00}-\eta^{00}+\frac{\eta^{00}}{3}\ln(\det(g^{ij})).
\end{equation}
Then let
\begin{eqnarray} \label{4.12}
\textbf{u}^{0\nu}&=&\frac{g^{0\nu}-\eta^{0\nu}}{2\tau},\\\label{4.13}
\textbf{u}_{0}^{0\nu}&=&\partial_{\tau}(g^{0\nu}-\eta^{0\nu})-\frac{3(g^{0\nu}-\eta^{0\nu})}{2\tau},\\\label{4.14}
\textbf{u}_{i}^{0\nu}&=&\partial_{i}(g^{0\nu}-\eta^{0\nu}),\\\label{4.15}
\textbf{u}^{ij}&=&\textbf{g}^{ij}-\delta^{ij},\\\label{4.16}
\textbf{u}_{\mu}^{ij}&=&\partial_{\mu}\textbf{g}^{ij},\\\label{4.17}
\textbf{u}&=&\textbf{q},\\\label{4.18}
\textbf{u}_{\mu}&=&\partial_{\mu}\textbf{q}. \\
\label{e:densty}\alpha &=&\beta(\tau)\zeta,\\
\label{e:densty2}\bar{\alpha} &=&\beta(\tau)\bar{\zeta},\\
\label{e:velcty}u^{i}&=&\beta(\tau)v^{i}, \\
\label{e:denstydif}\delta\zeta&=& \zeta-\bar{\zeta}, 
\end{eqnarray}
where $\beta(\tau)$ comes from Assumption \ref{a:tLip}.

\subsection{Initial data}\label{remark:1.2}
It is well known that the initial data for the conformal Einstein--Euler equations cannot be chosen freely on the initial hypersurface
\begin{equation*}
\Sigma_{1} = \{1\}\times \mathbb{T}^3 \subset \mathfrak{M}=(0,1]\times \mathbb{T}^3 .
\end{equation*}
Indeed, a number of constraints, which we can separate into gravitational, gauge and velocity normalization, must be satisfied on
$\Sigma_{1}$. In specific,
the initial data are governed by the constraint equations which are essentially an elliptic system. We have to specify some part of the data, then the other data will be derived from these free ones via constraints. Let
\begin{align*}
	g^{\mu\nu}|_{\tau=1} =&g_{0}^{\mu\nu}(x),\quad \partial_{\tau}g^{\mu\nu}|_{\tau=1}=g_{1}^{\mu\nu}(x),\quad \rho|_{\tau=1}=\rho_0(x),\quad v^{\alpha}|_{\tau=1}=\nu^{\alpha}(x).
\end{align*}
Above initial data set $(g_{0}^{\mu\nu}(x),g_{1}^{\mu\nu}(x),\rho_0(x),\nu^{\alpha}(x))$ can not be chosen arbitrarily. They must satisfy the Gauss-Codazzi equations, which are equivalent to $(G^{\mu0}-T^{\mu0})|_{\tau=1}=0$. Moreover, they also satisfy the wave coordinates condition $Z^{\mu}|_{\tau=1}=0$, the precise definition of $Z^{\mu}$ can be found in Section \ref{section:3.1}.

There are a number of distinct methods available to solve these constraint equations. However, we do not intent to state or prove the exact initialization theorem in the current article. One can always use the similar method in \cite{Oliynyk2009a,Oliynyk2009b,Liu2017,Liu2018} which is an adaptation of the method introduced by Lottermoser in \cite{Lottermoser1992} to recover the statements and proofs. The data exist but may not be uniquely selected. In this article, we \textit{assume} the data have already been selected properly and only focus on the further evolutions as $\tau\searrow 0$.

\subsection{Main Theorem} \label{section:1.3}

With above notations, our main result is stated here and the proof will be given in \S\ref{s:theorem:1.4}.
\begin{theorem}\label{theorem:1.4}
Suppose $k\in \mathbb{Z}_{\geq3}$, $\Lambda>0$,  $g_{0}^{\mu\nu} \in H^{k+1}(\mathbb T^{3})$, $g_{1}^{\mu\nu} ,\,\rho_0 ,\,\nu^{\alpha} \in H^{k}(\mathbb T^{3})$, $\rho_0>0	$ for all $x\in\mathbb T^{3}$ and the unknowns are determined by the data on the initial hypersurface that
\begin{equation}\label{e:data}
	(g^{\mu\nu},\partial_\tau g^{\mu\nu}, \rho, u^i)|_{\tau=1}=(g_{0}^{\mu\nu},g_{1}^{\mu\nu},\rho_0,\nu^{\alpha})
\end{equation}
which solve the constraint equations
\begin{equation*}
	(G^{0\mu}-T^{0\mu})|_{\tau=1}=0 \AND Z^\mu|_{\tau=1}=0.
\end{equation*}
The fluids for the Einstein--Euler system satisfy Assumption \ref{a:Maksym}--Assumption \ref{a:postvty}.
Then there exists a constant\footnote{This constant $\sigma$ can be considered small eventually from the proof of this Theorem.  } $\sigma>0$, such that if
$$
\|g_{0}^{\mu\nu}-\eta^{\mu\nu}(1)\|_{H^{k+1}}+\|g_{1}^{\mu\nu}-\partial_{\tau}\eta^{\mu\nu}(1)\|_{H^{k}}+\|\rho_0 -\bar{\rho}(1)\|_{H^{k}}+
\|\nu^{i} \|_{H^{k}}<\sigma,
$$
there exists a unique classical solution $g^{\mu\nu}  \in C^{2}((0,1]\times\mathbb T^{3})$, $\rho, \, v^i\in C^{1}((0,1]\times\mathbb T^{3})$ to the conformal Einstein--Euler system  \eqref{e:cEin1}--\eqref{e:cEu1} that satisfies the initial data \eqref{e:data}, the wave gauge $Z^\mu=0$ in $(0,1]\times \mathbb{T}^3$ and the following regularity conditions
\begin{align}
(g^{\mu\nu},u^{\mu}, \rho )\in \bigcap_{\ell=0}^2 C^{\ell}((T_1,1],H^{k+1-\ell}(\mathbb{T}^3 ))
\times \bigcap_{\ell=0}^1 &  C^{\ell}((T_1,1],H^{k-\ell}(\mathbb{T}^3))  \notag \\
& \times \bigcap_{\ell=0}^1  C^{\ell}((T_1,1],H^{k-\ell}(\mathbb{T}^3)), \label{e:gvrred}
\end{align}
and the estimates that
\begin{align*}
	\|g^{\mu\nu}(\tau)-\eta^{\mu\nu}(\tau)\|_{H^{k+1}}+\|\partial_{\kappa}& g^{\mu\nu}(\tau)-\partial_{\kappa}\eta^{\mu\nu}(\tau)\|_{H^{k}} \notag \\
	&+\|\rho(\tau)-\bar{\rho}(\tau)\|_{H^{k}}+\|u^{i}(\tau)\|_{H^{k}}\lesssim\sigma.
\end{align*}
\end{theorem}

\begin{remark}
	We do not include the asymptotic behaviors of the solutions in this article, but they can be derived in a similar way to \cite{Oliynyk2016a} involving some calculations of decay exponents.
\end{remark}

\subsection{Prior and related work}
\label{section:1.2}

The stability problems of certain exact solutions to Einstein--matter systems are  important in mathematical general relativity and there are a lot of works related to them. We only mention some of them which directly related to our current article, readers can find more in these references. First, a  well known and groundbreaking work by D. Christodoulou and S. Klainerman is the stability of Minkowski spacetime \cite{Christodoulou1993} as a solution to the Einstein-vacuum equations. Another remarkable approach based on the wave coordinates was given by  H. Lindblad and I. Rodnianski \cite{Lindblad2005,Lindblad2010}.

To serve our purpose, we turn to  the fully nonlinear future stability of Friedmann-Lema\^{\i}tre-Robertson-Walker (FLRW) solutions with the positive cosmological constant which has been well studied during this decade. H. Ringstr\"{o}m \cite{Ringstroem2008} first investigated the future global non-linear stability in
the case of Einstein's equations coupled with a non-linear scalar field $\widetilde{T}^{\mu\nu}=\widetilde{\partial}^{\mu}\Psi\widetilde{\partial}^{\nu}\Psi-[\frac{1}{2}\widetilde{g}_{\mu\nu}\widetilde{\partial}^{\mu}\Psi\widetilde{\partial}^{\nu}\Psi+V(\Psi)]\widetilde{g}^{\mu\nu}$ where, under the assumption that $V(0)>0,\, V^{'}(0)=0,\,V^{''}(0)>0$, $V(\Psi)$ plays the role of the positive cosmological constant which appears in the most of the following works. The main observation of this paper is two-fold: First, the Einstein-nonlinear scalar field system can be formulated as a system of nonlinear wave equations with the help of generalized wave coordinates; Second, the problem under considerations describes the accelerated expansion of the universe, and the accelerated expansion provides dispersive terms, which lead to exponential decay for solutions. Inspired by H. Ringstr\"{o}m's work, I. Rodnianski and J. Speck \cite{RodnianskiSpeck:2013} established the future non-linear stability of these FLRW solutions with positive cosmological constant and linear equation of state under the condition $0 < K < 1/3$ and the assumption of zero fluid vorticity. After that, M. Had\v{z}i\'c and J. Speck \cite{Hadzic2015} and J. Speck \cite{Speck2012} answered that this future non-linear stability result remains true
for fluids with non-zero vorticity and also for the equation of state $p = 0$. By employing the conformal method developed by H. Friedrich \cite{Friedrich1986,Friedrich1991}, C. L\"ubble and J. A. V. Kroon \cite{Luebbe2013} proved the above question for the equation of state with parameter values $K = 1/3$ that is the pure radiation universe. After these, T. Oliynyk \cite{Oliynyk2016a} gave an alternative proof for non-linear future stability problems of FLRW solutions based on conformal singular hyperbolic formulations of Einstein--Euler equations, a completely different method comparing with above works, which provides the basic tool of the current article. 
One  advantage of the this method is that, under a conformal transformation,  by choosing the conformal factor and the source term of the wave gauge and variables judiciously, the whole Einstein--Euler system can be turned into a symmetric hyperbolic system with singular in time terms (with ``good'' sign) and solutions defined on finite interval of time. By a variation of standard energy estimates, one can then get the global nonlinear stability of a family of  FLRW solutions and establish the asymptotic behavior of perturbed solutions in the far future. Finally, we point out  that, in the regime $K>\frac{1}{3}$,   \cite{Rendall2004} has found some evidence for instability using formal expansions. Recently, T. Oliynyk \cite{Oliynyk2020} studied the relativistic perfect fluids with linear equation of state on exponentially expanding FLRW spacetimes, and proved the future stability of nonlinear perturbations of a class of homogeneous solutions when $1/3<K<1/2$ under small initial data hypothesis.

Inspired by the structure of the conformal singular hyperbolic equations in \cite{Oliynyk2016a}, C. Liu and T. Oliynyk \cite{Liu2017,Liu2018,LiuThesis} answered a fundamental question:\textit{ On what scales can Newtonian cosmological simulations be trusted to approximate realistic relativistic cosmologies?} They investigated the fully nonlinear long time behavior of our universe based on a small perturbation of the FLRW metric with perfect fluid of linear equation of state and positive cosmological constant, and rigorously proved the errors of solutions between Newtonian gravity (governed by Poisson--Euler equations) and general relativity (governed by Einstein--Euler equations) are well-controlled for the long time evolution in future direction in suitable norms by judicious initial data selections. That is, they established the existence of $1$-parameter families of $\epsilon$-dependent solutions to the Einstein--Euler equations
with a positive cosmological constant $\Lambda >0$ and a linear equation of state $p=\epsilon^2 K \rho$, $0<K\leq 1/3$, for the parameter
values $0<\epsilon < \epsilon_0$. These solutions
exist globally to the future, converge as $\epsilon \searrow 0$ to solutions of the cosmological Poisson--Euler equations of Newtonian gravity,
and are inhomogeneous nonlinear perturbations of FLRW fluid solutions. This cosmological Newtonian limit problem have been investigated for both the isolated version (periodic universe) of cosmology \cite{Liu2017} and the cosmological large space-time scales (multi-body version of universe) \cite{Liu2018}. These two results answered the feasibility of Newtonian approximations used in astrophysics for 100 years old with rigid mathematical proofs (astrophysicists use this approximation as a hypothesis for a century) and intuitively, this question can be viewed as a stability problem of general relativistic cosmological solution around the Newtonian gravity in large scales.

We emphasize that the above results are obtained by assuming linear equation of state of the perfect fluids. P. LeFloch and C. Wei \cite{LeFloch2015a} investigated the Einstein-Chaplygin fluids with equation of state $p=-\frac{A}{\rho^{\vartheta}}$ with $\Lambda=0$. This model is widely investigated by physicists as a candidate for the dark energy. In \cite{LeFloch2015a}, they only considered the irrotational fluids, such that the matter is indeed a scalar field, under which the mechanism for the accelerated expansion of the spacetime is the negativity of the pressure (an analogue of positive cosmological constant). The main idea comes from T. Oliynyk's conformal singular hyperbolic equations \cite{Oliynyk2016a} and the main difference is that they choose different conformal factor to make sure that the irrotational Euler equations are regular. However, this technique can not be applied to the general system due to the normalization condition $\tilde{u}^{\mu}\tilde{u}_{\mu}=-1$. Thus, \textit{one of the main motivations} of this article is to investigate the evolution of the general Chaplygin fluids in accelerated expanding spacetime. 

Recently, Beyer, Oliynyk and Olvera-Santamaría \cite{Beyer2019} generalized the global existence theory of this singular formulations and applied this theory to semilinear wave equations near spatial infinity on Minkowski and Schwarzschild spacetimes, and to the relativistic Euler equations with Gowdy symmetry on Kasner spacetimes.

It is worth noting that in $1994$, U. Brauer, A. Rendall and O. Reula \cite{Brauer1994} proved a fully non-linear future stability problem on the Newtonian cosmological model, which laid the foundation on the corresponding stability problem in the Newtonian setting. By the classical energy estimate, they showed the fluids with equation of state $p=K\rho^{\frac{n+1}{n}}$ is globally stable in accelerated expanding universe when the initial data are small.

On the other direction, I. Rodnianski and J. Speck \cite{Rodnianski2018a,Rodnianski2018} has proven in the collapsing direction $t\searrow 0 $ the FLRW solution is globally nonlinearly stable under small perturbations of its initial data at $t=1$. They formulated their results in the constant mean curvature (CMC)-transported spatial coordinates gauge. Along this direction, we point out another two important works by L. Andersson and A. Rendall \cite{Andersson2001}, and H. Rinstr\"om \cite{Ringstroem2001}, respectively.

We also mention that there are a series of works about anti-de-Sitter spacetimes (that is, for $\Lambda<0$), for example, by Friedrich, Dafermos, Holzegel and Smulevici, Moschidis \cite{Friedrich1995,Dafermos2006,Holzegel2011,HOLZEGEL2012,Moschidis2018}, respectively and references cited therein.

\subsection{Overview of the method}
\subsubsection{Motivations of the perfect fluids}
We first introduce a new class of fluids, which includes several common and frequently used fluid models such as, by directly verifying the definition in the next section, fluids with linear equation of state, Chaplygin gases and polytropic gases, etc. As we have pointed out at the beginning of this paper, our \textit{direct purpose} of current article is to find out what types of fluids can guarantee the validity that the future non-linear stability of FLRW solutions to the Einstein--Euler system with a positive cosmological constant and try to identify as many fluids as possible. Once this type of fluids covers a large class of reasonable fluids, then the stability problem of FLRW metric is \textit{``stable'' with respect to equations of state of this class of fluids} in some sense. The assumptions on the perfect fluids come from the proof of the future non-linear stability of FLRW solutions to the Einstein--Euler system with a positive cosmological constant. We carefully give these assumptions to make sure the stability result holds and this class contains as many types of fluids as possible.

\subsubsection{Conformal singular formulations of Einstein equations}
The main tool of the proof of Theorem \ref{theorem:1.4} is  the conformal singular hyperbolic formulation of \cite{Oliynyk2016a}. The idea of writing conformal Einstein equations to this formulation belongs to \cite{Oliynyk2016a} and is based on the following three main ingredients:
\begin{enumerate}
	\item Conformal transformation with certain conformal factor $\Phi$ defined by \eqref{1.10}. This is reasonable because the conformal FLRW metric is close to the Minkowski metric which can simplify the geometry and calculations, and the conformal time coordinate $\tau\in(0,1]$ compatible with this conformal transformation compacifies the original time $t\in[0,\infty)$ to a finite interval, which \textit{converts a long time problem into a short finite time one}. The cost of this compactification of time is the equation becomes singular in time but singular in a ``good'' way.
	\item Wave gauge defined by \eqref{E:WAVEGAUGE}--\eqref{4.5}. This is not surprising, since this wave gauge comes from the standard one without source term via conformal transformations, and this new source term eliminates and adjusts the bad terms of the remainders to the suitable singular terms and regular terms consistent with the singular hyperbolic system \eqref{e:model1}.
	\item Good analyzable variables given by \eqref{e:q}--\eqref{e:denstydif} are the key part of writing Einstein equations into the appropriate singular formulations \eqref{e:model1}.
\end{enumerate}
By using these three ingredients, it is not hard to rewrite conformal Einstein system \eqref{3.11} into a symmetric hyperbolic system with $\tau$-singular term like  \eqref{e:model1}.

\subsubsection{Conformal singular formulations of Euler equations}
In order to rewrite the Euler equations as the target singular equation, more attention is required than the one with the linear equation of state or Chaplygin gases  due to the generality of this large class of fluids, and this provides the \textit{main innovations} of this article. Let us briefly point out the difficulties here and address the key ideas about how to overcome these difficulties. Because of differences between the structures of \textit{Fluids} $(I)$ and \textit{Fluids} $(II)$, we have to write them to the aimed singular system \eqref{e:model1}--\eqref{e:model2}  in different ways.

Let us first concentrate on the \textit{Fluids} $(I)$, and there are four steps to arrive at the expected formulations. The \textit{first step} is standard, by using the normalization relation of velocity $u^\mu u_\mu=-1$, Rendall's projection \eqref{e:Ldef} and reflection operator $M_{ki}$ defined in \eqref{e:Mdef} to write the conformal Euler equations \eqref{e:cEu1} to a symmetric hyperbolic system \eqref{Euler-original1}--\eqref{Euler-original2} (this formulation is widely used and can be found, for example, in \cite{Rendall1992,Oliynyk2015,Liu2017,Liu2018}). For the problems involving time away from $\tau=0$, this formulation is symmetric hyperbolic equation. However, when $\tau$ is close to $0$, the coefficient matrix will \textit{degenerate}, which \textit{destroys the structure of hyperbolicity} of this system. In specific, multiplying suitable factors, for example, $s^2/(\rho+p)$ on both sides of \eqref{Euler-original1} and $\rho+p$ on both sides of \eqref{Euler-original2}, then \eqref{Euler-original1}--\eqref{Euler-original2} become symmetric. However, from the behavior of the background density $\bar{\rho}\searrow 0$ and pressure $\bar{p}\searrow 0$, as $\tau\rightarrow 0$ (see \eqref{e:pest}--\eqref{e:bgrhoest1}), the expected behaviors of $\rho$ and $p$ due to small perturbation of initial data $\bar{\rho}(1)$ and $\bar{p}(1)$ are also $\rho\searrow 0$ and $p\searrow 0$ as $\tau\rightarrow 0$, which leads to a fact that some of elements of the coefficient matrix of above symmetric hyperbolic system tend to $0$ or $\infty$ as $\tau\rightarrow 0$. The degenerated coefficients when time is close to $0$ violate the Condition \eqref{c:5} in Appendix \ref{section:3.5}.

In \textit{step two}, the idea to overcome above difficulty is to generalize the idea of the non-degenerated symmetrization of Euler equations originated by Makino  \cite{Makino1986} who firstly came up this formulation for polytropic fluids in a compact set in Newtonian setting. We extract the key properties of his idea to write above degenerated equations to a non-degenerated hyperbolic system by introducing a new density variable $\alpha$ defined by Assumption \ref{a:Maksym}. However, this is still not enough to apply the theorem in Appendix \ref{section:3.5}, because the current variable $\alpha-\bar{\alpha}$ can not make sure some of the remainders in Einstein equation to be regular in $\tau$ (eventually, we need the complete Einstein--Euler system to be the target singular system in Appendix \ref{section:3.5}). In other words, The variable $\alpha-\bar{\alpha}$ in above system is not a suitable one in the singular hyperbolic system.

In order to eliminate the singular terms in Einstein equation, we choose, in \textit{step three}, $\delta\zeta$ as the variable describing the density instead of $\alpha-\bar{\alpha}$. With the help of Assumption \ref{a:tLip}.\eqref{e:rhopro}, the remainder terms including $(\rho-\bar{\rho})/\tau^2$ and $(p-\bar{p})/\tau^2$ in the Einstein equations are \textit{regular} in $\tau$ and analytic in $\delta\zeta$, see \eqref{e:t2rho}. 
 After changing $\alpha-\bar{\alpha}$ to $\delta\zeta$, $u^i$ have to been changed to $v^i$ defined by $u^{i}= \beta(\tau)v^{i}$ to make sure this system to be \textit{symmetric}. Hence, we have to change the variables of fluids from $(\alpha-\bar{\alpha},u^i)$ to $(\delta\zeta,v^i)$ to rewrite the equations. 

However, the process of rescaling the velocity $u^i$ to $v^i$ by $\beta$ brings a new singular term $\frac{\partial_{\tau}g^{0l}}{\beta(\tau)}$ to Euler equations. 
Although equations \eqref{Euler-final1a}--\eqref{Euler-final1b} seem to be consistent with the non-degenerated singular hyperbolic system given in Appendix \ref{section:3.5}, the remainders involves $\tau$-singular terms of $\textbf{u}^{00}_i$, $\textbf{u}^{0j}$ and $\textbf{u}^{0j}_0$ in a ``bad'' way in \eqref{e:Si}, which destroys the structure of the singular term in the system of Appendix \ref{section:3.5}. In order to conquer this difficulty, we introduce a new variable $
\textbf{v}^{k}=v^{k}-Ag^{0k}=v^{k}-2\tau A \textbf{u}^{0k}$,
where $
A= -\frac{3\bar{s}^2}{\omega\beta(\tau)}$.
This new variable adjusts the relations of $\textbf{u}^{0j}$ and $\textbf{u}^{0j}_0$ to a ``good'' form, which leads to the ``right'' formulation of the full set of Einstein--Euler system agreeing with the model \eqref{e:model1} in Appendix \ref{section:3.5}.

For \textit{Fluids} $(II)$, when $\beta\equiv $ constant, by \eqref{e:rhopro} in Assumption \ref{a:tLip}, we only need to  proceed the first and second steps as above \textit{Fluids} $(I)$, but give up Step $3$--$4$. However,
we have to change $u^j$ to $u_q$, otherwise Condition \eqref{c:6} in Appendix \ref{section:3.5} fails due to the degeneracy of $\textbf{BP}$. Since $\textbf{P}u$ is the velocity $u^j$, by letting $\textbf{P}u=0$ in $\textbf{P}^\perp B(t,u) \textbf{P}$ and noting $u_q=g_{q0} u^0+g_{qi} u^i$, then $\textbf{P}^\perp B(t,\textbf{P}^\perp u) \textbf{P}=\textbf{P} B(t,\textbf{P}^\perp u) \textbf{P}^\perp \neq 0$. A good expression of the Euler equations verifying Condition \eqref{c:6} in Appendix \ref{section:3.5} relies on the variable $u_q$ instead of $u^i$. We only need to rewrite the Euler equations in terms of variables $(\delta\zeta, u_q)$ further to obtain the system consistent with the model equation  \eqref{e:model1} in Appendix \ref{section:3.5}.


\subsection{Paper outline}
In \S \ref{S:MAK}, we first analyze some key asymptotic properties of the background solutions, and then verify that  several common and frequently used examples belong to the fluids satisfying our assumptions.

In \S \ref{section:2}, we employ the variables \eqref{e:q}--\eqref{e:denstydif} and the wave gauge \eqref{E:WAVEGAUGE}--\eqref{4.5} to write the conformal
Einstein system, given by \eqref{3.11}, as a symmetric hyperbolic system, see \eqref{4.26}--\eqref{4.28}, that
is singular in $\tau$. The method of this transformation comes from \cite{Oliynyk2016a} originally.

In \S \ref{section:3}, we further write conformal Euler equation \eqref{e:cEu1} as the non-degenerated singular symmetric hyperbolic system satisfying the conditions in the Appendix \ref{section:3.5} for \textit{Fluids} $(I)$ and $(II)$ fluids respectively. This section is the main innovation of this article and the goal of this section is to derive equations \eqref{symmetric-Euler} and \eqref{lower-case}. The key to achieve this goal is to choose the right variables. For \textit{Fluids} $(I)$, we propose a four-step process and for \textit{Fluids} $(II)$, we use the appropriate variable of velocity $u_i$ instead of $u^i$.

In \S \ref{s:theorem:1.4}, we complete the proof of Theorem \ref{theorem:1.4} by using the results from \S \ref{section:2} to \S \ref{section:3} to verify that all the
conditions from the model equation \eqref{e:model1} in Appendix \ref{section:3.5} hold for the formulation \eqref{symmetric-Euler1Ma1} and \eqref{symmetric-Euler1Ma2} of the conformal Einstein--Euler equations. This
allows us to apply Theorem \ref{pro:3.16} to obtain the desired conclusion.


In the Appendix \ref{section:3.5}, The slightly revised version of  Theorem of singular hyperbolic equations has been introduced and we point out how to revise the proof due to the revisions of the statement.


\section{Examples of the perfect fluids}\label{S:MAK}


Let us first derive some general properties of the background solutions which will play the role of the center of perturbations.

\subsection{Analysis of FLRW spacetimes}\label{section:2.3}

Einstein equations for the FLRW solutions reduce to the Friedman equations    	
\begin{align}
-3\omega^2-\Omega+\Bigl(\frac{\bar{\rho}-\bar{p}}{2}+\Lambda\Bigr)=&0  \label{3.18}
\intertext{and}
-6\Omega-6\omega^{2}+2\Lambda-(\bar{\rho}+3\bar{p})=&0 \label{3.19}
\end{align}
where we recall $\Omega(\tau) =-\tau \omega \partial_\tau \omega$ is defined by \eqref{e:Omeg}.
The background solution verifies the Euler equations that
\begin{align}\label{e:bgrhoeq}
\partial_0 \bar{\rho} =\frac{3}{\tau} (\bar{\rho} +\bar{p} ).
\end{align}
In fact, equations \eqref{3.18}, \eqref{3.19} and \eqref{e:bgrhoeq} form the Einstein--Euler system for the background FLRW solutions.

Let us estimate several crucial quantities which characterize the asymptotic behaviors of the background solutions. First note that \eqref{3.18} and \eqref{3.19} yield
\begin{equation}\label{3.20}
\omega^{2}=\frac{1}{3}(\Lambda+\bar{\rho}) \quad \text{and} \quad
\Omega=-\frac{\bar{\rho}+\bar{p}}{2}.
\end{equation}

Due to the fact that $p=p(\rho)= c_s^2(K_0\rho)\rho$ where $K_0\in (0,1)$, with te help of \eqref{e:srange}, we derive
\begin{equation}\label{e:pest}
0 \leq \bar{p} \leq \frac{1}{3} \bar{\rho}.
\end{equation}
Then integrating \eqref{e:bgrhoeq}, with the help of \eqref{e:pest} and \eqref{3.20}, yields
\begin{gather}
\tau^4\bar{\rho}(1)\leq \bar{\rho}(\tau) \leq \tau^3\bar{\rho}(1),
 \label{e:bgrhoest1} \\
 \frac{1}{3}\bar{\rho}(1)\tau^4\leq\omega^2-\frac{\Lambda}{3}\leq \frac{1}{3} \bar{\rho}(1)\tau^3, \label{e:bgrhoest4}  \\
-\frac{2}{3} \tau^3\bar{\rho}(1)  
\leq \Omega
\leq -\frac{1}{2} \tau^4\bar{\rho}(1) \label{e:bgrhoest3}
\end{gather}
and
\begin{align}
3\tau^3\bar{\rho}(1)
 \leq\partial_\tau \bar{\rho}
  \leq 4 \tau^2\bar{\rho}(1).   \label{e:bgrhoest2}
\end{align}
Moreover, recalling $\psi(\tau)$ in \eqref{e:lamtau}, let us estimate its time derivative which will be used in the derivations of reduced conformal Einstein equations. By \eqref{3.20}, we arrive at
\begin{align}\label{e:dtlam}
	\partial_\tau\psi(\tau)=\partial_\tau \Bigl(3\omega^2+\frac{3}{2}\Omega\Bigr) =\partial_\tau\bar{\rho}  -\frac{3}{4} \partial_\tau (\bar{\rho}+\bar{p})=\frac{1}{4}\partial_\tau\bar{\rho}  - \frac{3}{4}  \partial_\tau \bar{p}.
\end{align}
Note that
\begin{align}\label{e:dtbgp}
	\partial_\tau \bar{p}=\bar{c}_s^2  \partial_\tau \bar{\rho} \in \Bigl[0, \frac{1}{3}\partial_\tau \bar{\rho}\Bigr].
\end{align}
Gather \eqref{e:bgrhoest2}--\eqref{e:dtbgp} together, we arrive at
\begin{align}\label{e:dtlambda}
	0\leq \partial_\tau\psi(\tau) \leq \tau^2\bar{\rho}(1).
\end{align}
In addition, direct calculations show that
\begin{align}\label{e:dt2omg}
	\partial^2_\tau \omega=\frac{1}{\tau^2\omega}\frac{\bar{\rho}+\bar{p}}{2}\Bigl(\frac{\Omega}{\omega^2}-1\Bigr)+\frac{1}{\tau\omega}\frac{1}{2} (\partial_\tau \bar{\rho} +\partial_\tau \bar{p})
\end{align}
with the help of \eqref{e:bgrhoest1}--\eqref{e:dtlambda}, \eqref{e:dt2omg} implies that $\omega \in C^2([0,1])$.

\subsection{Examples of the target fluids}\label{s:eplMF}

In this section, we give three examples of the target fluids, which demonstrates that the set of the fluids considered in this paper is \textit{non-empty}.
\subsubsection{The polytropic gas}\label{s:pogas}
Suppose the equation of state of the polytropic gas is given by
\begin{equation}\label{e:eospo}
p=K\rho^{\frac{n+1}{n}}
\end{equation}
for $n\in(1,3)$. We also assume the initial homogeneous, isotropic density is bounded by
\begin{align}\label{e:denstrg}
	0<\bar{\rho}(1)\leq \frac{1}{(4Kn(n+1))^n} \biggl(2n\sqrt{\frac{1}{3}\Bigl(1-\frac{3-\epsilon}{2n}\Bigr)}\biggr)^{2n}.
\end{align}

We introduce the standard relationships between $\rho$ and $\alpha$ to define Makino density $\alpha$ which are (for details, see, for example, \cite{Oliynyk2008a,Rendall1992})
\begin{align}\label{e:polyra}
\rho=&\mu(\alpha)=\frac{1}{\bigl(4Kn(n+1)\bigr)^n}\alpha^{2n}  \\
\intertext{and}
\lambda=&\lambda(\alpha)=\Bigl(1+\frac{1}{4n(n+1)}\alpha^2\Bigr)^{-1}. \label{e:pollab}
\end{align}

Direct calculations imply that
\begin{equation}\label{e:pspol}
p=\frac{K}{(4Kn(n+1))^{n+1}}\alpha^{2(n+1)} \AND s=\frac{\alpha}{2n}.
\end{equation}
Furthermore,
\begin{align*}
\frac{d\mu(\alpha)}{d\alpha} = \frac{2n}{\bigl(4Kn(n+1)\bigr)^n}\alpha^{2n-1}=\frac{\lambda(\alpha)(\rho+p)}{s(\alpha)},
\end{align*}
which verifies Assumption \ref{a:Maksym}.

In order to verify Assumption \ref{a:tLip} and \ref{a:postvty}, we first calculate the background density $\bar{\alpha}$. Integrating \eqref{e:bgrhoeq} yields
\begin{align}\label{e:intbg}
\int^{\bar{\rho}(\tau)}_{\bar{\rho}(1)}\frac{d\xi}{\xi+p(\xi)}=\ln \tau^3.
\end{align}
Combining the equation of state \eqref{e:eospo}, transformation \eqref{e:polyra} and  \eqref{e:intbg} yields
\begin{align*}
2n \int^{\bar{\alpha}(\tau)}_{\bar{\alpha}(1)}\frac{dy}{y\bigl(1+\frac{1}{4n(n+1)}y^2\bigr)}=\ln \tau^3
\end{align*}
which, in turn, implies
\begin{align}\label{e:alpbg}
\bar{\alpha}(\tau)=\tau^{\frac{3}{2n}}Q(\tau) \in [0,  \bar{\alpha} (1)  ]
\end{align}
where
\begin{equation*}
	Q(\tau)=\biggl(\frac{1}{\bar{\alpha}^2(1)}+\frac{1}{4n(n+1)}-\frac{1}{4n(n+1)}\tau^{\frac{3}{n}}\biggr)^{-\frac{1}{2}}
\end{equation*}
for $\tau\in [0,1]$. Now let us continue to verify Assumptions \ref{a:tLip}--\ref{a:postvty}.

For every $n\in(1,3)$ and for every $\varepsilon \in (0,1]$ there is $n \in (\frac{3-\varepsilon}{2},3-\varepsilon]\cap (1,3)$, we take  $\varsigma=3-\varepsilon$,   $\beta(\tau)=(C^*\hat{\delta})^{-1} \tau^{(3-\varepsilon)/(2n)}\in C[0,1]\cap C^1(0,1]$ (where, we recall, $C^*$ and $\hat{\delta}$ are defined in Assumption \ref{a:postvty}.$(1)$ and \ref{a:Maksym}, respectively) and 
\begin{align}\label{e:vrhopl}
\varrho\bigl(\tau,y\bigr) 
=&\frac{(C^*\hat{\delta})^{-2n}}{\bigl(4Kn(n+1)\bigr)^n}\bigl((y+(C^*\hat{\delta}) \tau^{-\frac{3-\varepsilon}{2n}}\bar{\alpha})^{2n} -((C^*\hat{\delta}) \tau^{-\frac{3-\varepsilon}{2n}}\bar{\alpha})^{2n}\bigr),
\end{align}
then $\varrho\in C\bigl([0,1], C^\infty( \mathbb{R})\bigr)$, $\varrho(\tau,0)=0$, $\tau^{(3-\varepsilon)/(2n)}\gtrsim \tau$ and $\chi(\tau) =(3-\varepsilon)/(2n)$. Moreover, transformation \eqref{e:vrhopl} yields
\begin{align*}
\mu(\alpha)-\mu(\bar{\alpha})=&\tau^{3-\varepsilon}[\mu(\tau^{-(3-\varepsilon)/(2n)}\alpha)-\mu(\tau^{-(3-\varepsilon)/(2n)} \bar{\alpha})] \notag \\
=& \tau^{3-\varepsilon}\varrho \bigl(\tau,(C^*\hat{\delta})\tau^{-(3-\varepsilon)/(2n)}(\alpha-\bar{\alpha})\bigr),
\end{align*}
which verifies Assumption \ref{a:tLip}.


Equations \eqref{e:pollab}, \eqref{e:pspol} and \eqref{e:alpbg} lead to
\begin{align*}
\bar{\lambda}=\lambda(\bar{\alpha})=\Bigl(1+\frac{1}{4n(n+1)}\bar{\alpha}^2\Bigr)^{-1} \AND \bar{s}=s(\bar{\alpha})=\frac{\bar{\alpha}}{2n}.
\end{align*}
It is clear that $\bar{s} \lesssim \beta(\tau) $ by \eqref{e:alpbg}. Noting that $n\in (\frac{3-\varepsilon}{2},3-\varepsilon]$ leads to $\tau^\frac{3}{2n}<\tau^{\frac{6-\varepsilon}{2n}-1} \lesssim \tau^\frac{\epsilon}{2n}\lesssim 1$ and $\tau^{1+\frac{\epsilon}{n}}<\tau^{\frac{3}{n}-1}  \lesssim \tau^\frac{\epsilon}{n}\lesssim 1$, we calculate the quantity
\begin{align*}
  \lambda^\prime (\bar{\alpha}) \partial_\tau \beta(\tau)
=   -\frac{1}{2n(n+1)} \biggl( \frac{3-\varepsilon}{2n}\tau^{\frac{6-\varepsilon}{2n}-1}Q(\tau) \biggr) \Bigl(1+\frac{1}{4n(n+1)} \bar{\alpha}^2\Bigr)^{-2} (C^*\hat{\delta})^{-1}
\lesssim   1
\end{align*}
and
\begin{align*}
 \frac{\bar{s}}{\tau}   \lambda^\prime (\bar{\alpha})
=   -\frac{1}{2n(n+1)} \biggl(  \frac{\tau^{\frac{3}{n}-1}Q(\tau)^2}{2n}\biggr) \Bigl(1+\frac{1}{4n(n+1)} \bar{\alpha}^2\Bigr)^{-2}
\lesssim  1.
\end{align*}

Then, calculate quantity
\begin{align*}
\frac{1}{3}\chi(\tau)+ \frac{\varepsilon}{6n}=\frac{1}{2n} \leq q^\prime(\bar{\alpha}) =  
\frac{1}{2n}+\frac{3\bar{\alpha}^2}{8n^2(n+1)} \leq \frac{1}{3}\chi(\tau)+ \frac{\varepsilon}{6n}+\frac{3\bar{\alpha}^2(1)}{8n^2(n+1)} ,
\end{align*}
and noting $n\in (\frac{3-\varepsilon}{2},3-\varepsilon]$ implies $1-\frac{3-\varepsilon}{2n}>0$, then \eqref{e:denstrg} implies
\begin{align*}
0<\bar{\alpha}(1)< 2n\sqrt{\frac{1}{3}\bigl(1-\frac{3-\epsilon}{2n}\bigr)},
\end{align*}
which, in turn, yields
\begin{align*}
1-3\bar{s}^2-\chi(\tau) = & 1-\frac{3\bar{\alpha}^2}{4n^2}-\frac{3-\varepsilon}{2n} \geq  1-\frac{3\bar{\alpha}^2(1)}{4n^2}-\frac{3-\varepsilon}{2n} 
> \frac{3-\epsilon}{2n}-\frac{3-\varepsilon}{2n}=0.
\end{align*}
It is evident that \eqref{e:btest} holds and then,
Assumption \ref{a:postvty} can be concluded. We have verified all the Assumptions \ref{a:Maksym}--\ref{a:postvty} now and this means this type of \textit{polytropic fluids} belongs to \textit{Fluids} $(I)$.

\subsubsection{Fluids with the \textit{linear} equation of state} \label{s:linf}
Suppose the equation of state of this fluid is
\begin{equation}\label{e:eosln}
	  p=K\rho
\end{equation}
for $K\in(0,\frac{1}{3}]$. Using \eqref{e:eosln} and \eqref{e:bgrhoeq}, we arrive at
\begin{align*}
	\bar{\rho}=\bar{\rho}(1)\tau^{3(1+K)}.
\end{align*}
Let the set $\{\lambda(\alpha), \varrho, \varsigma, \beta(\tau)\}$ be
\begin{gather*}
\alpha=\mu^{-1}(\rho)=\int_{\bar{\rho}(1)}^{\rho}\frac{d\xi}{\xi+p(\xi)}=\frac{1}{1+K}\ln \frac{\rho}{\bar{\rho}(1)}, \quad \lambda(\alpha) \equiv \sqrt{K},\quad \beta(\tau)\equiv 1,\\
 \varsigma= 3(1+K)  \AND \varrho(\tau,y)=\bar{\rho}(1) [e^{(1+K)y}-1].
\end{gather*}
Then $\varrho\in C\bigl([0,1], C^\infty( \mathbb{R})\bigr)$ and $\varrho(\tau,0)=0$.
It is easy to calculate that
\begin{gather*}
\bar{\alpha}=\mu^{-1}(\bar{\rho})=\int_{\bar{\rho}(1)}^{\bar{\rho}}\frac{d\xi}{\xi+p(\xi)}=\frac{1}{1+K}\ln \frac{\bar{\rho}}{\bar{\rho}(1)}= \ln \tau^{3} , \\
	s(\alpha)\equiv \lambda(\alpha)
	\equiv\sqrt{K},
	\AND \mu(\alpha)
	=\bar{\rho}(1)  e^{(1+K)\alpha},
\end{gather*}

Direct calculations, with the help of $\bar{s}=s(\bar{\alpha})=\bar{\lambda}=\lambda(\bar{\alpha})= \sqrt{K}$, show that $\bar{s} \lesssim \beta(\tau)$, and
\begin{gather*}
	\frac{d\mu(\alpha)}{d\alpha}=(1+K)\mu(\alpha) =(1+K)\bar{\rho}(1)  e^{(1+K)\alpha}= \frac{\lambda(\alpha)(\rho+p)}{s(\alpha)},  \\
	\mu(\alpha)-\mu(\bar{\alpha})
	= \tau^{3(1+K)}(\bar{\rho}(1)  e^{(1+K)(\alpha-\bar{\alpha})}-\bar{\rho}(1))=\tau^{3(1+K)} \varrho(\tau,\alpha-\bar{\alpha} ) ,  \\
	\intertext{and}
   \lambda^\prime (\bar{\alpha})\partial_\tau \beta(\tau) \equiv 0,  \qquad \frac{\bar{s}}{\tau}  \lambda^\prime (\bar{\alpha})\equiv 0,
\end{gather*}
which is enough to conclude Assumptions \ref{a:Maksym}, \ref{a:tLip} and \eqref{e:B0bd} in Assumption \ref{a:postvty}.

Note that $q\equiv 1$,
and
\begin{align*}
	 1-3\bar{s}^2 =1-3K =
	 \begin{cases}
	 0 \quad &\text{if} \quad K=\frac{1}{3}, \\
	 1-3K>0 \quad &\text{if} \quad 0<K<\frac{1}{3},
	 \end{cases}
\end{align*}
which verifies Assumption \ref{a:postvty}. Therefore, we conclude that this fluid belongs to \textit{Fluids} $(II)$.

\subsubsection{Chaplygin gases} \label{s:chapf}
Suppose the equation of state of Chaplygin gases is expressed by
\begin{equation}\label{e:eoschp}
	p=-\frac{\Lambda^{1+\vartheta}}{(\rho+\Lambda)^{\vartheta}}+\Lambda
\end{equation}
for $\vartheta\in\bigl(0,\sqrt{\frac{1}{3}}\bigr)$. It is easy to check that $p(0)=0$. 
Equation \eqref{e:intbg}, with the help of \eqref{e:eoschp}, leads to
\begin{align*}
	(\bar{\rho}+\Lambda)^{\vartheta+1}-\Lambda^{\vartheta+1}=\tau^{3(1+\vartheta)}[(\bar{\rho}(1)+\Lambda)^{\vartheta+1}-\Lambda^{\vartheta+1}].
\end{align*}

We take $\{\lambda(\alpha), \varrho, \varsigma, \beta(\tau)\}$ as
\begin{align}
\alpha=\mu^{-1}(\rho)=&\int_{\bar{\rho}(1)}^{\rho}\frac{d\xi}{\xi+p(\xi)}=\frac{1}{\vartheta+1}\ln\frac{(\rho+\Lambda)^{\vartheta+1}-\Lambda^{\vartheta+1}}{(\bar{\rho}(1)+\Lambda)^{\vartheta+1}-\Lambda^{\vartheta+1}},   \label{e:chppar1}  \\
\lambda(\alpha)
=&\frac{\vartheta \Lambda^{1+\vartheta}}{\Lambda^{\vartheta+1}+ [(\bar{\rho}(1)+\Lambda)^{\vartheta+1}-\Lambda^{\vartheta+1}]e^{(\vartheta+1)\alpha}},  \label{e:chppar2}  \\
\beta(\tau)=&1, \quad  \varsigma=3(1+\vartheta) \AND \varrho(\tau,y)=\tau^{-3(1+\vartheta)}(\mu(y+\bar{\alpha})-\mu(\bar{\alpha})). \label{e:chppar3}
\end{align}
It is easy to calculate that
\begin{align}
    \bar{\alpha}=&\mu^{-1}(\bar{\rho})= \frac{1}{\vartheta+1}\ln\frac{(\bar{\rho}+\Lambda)^{\vartheta+1}-\Lambda^{\vartheta+1}}{(\bar{\rho}(1)+\Lambda)^{\vartheta+1}-\Lambda^{\vartheta+1}}= \ln \tau^3  , \label{e:chapalp}\\
	\mu(\alpha) 
	= & \Lambda\biggl\{1+ \bigl[\bigl(1+\frac{\bar{\rho}(1)}{\Lambda}\bigr)^{\vartheta+1}-1\bigr]e^{(\vartheta+1)\alpha}\biggr\}^\frac{1}{\vartheta+1}-\Lambda,   \label{e:chaprho} \\
	\mu(\bar{\alpha}) 
	= & \Lambda\biggl\{1+ \bigl[\bigl(1+\frac{\bar{\rho}(1)}{\Lambda}\bigr)^{\vartheta+1}-1\bigr] \tau^{3(\vartheta+1) }\biggr\}^\frac{1}{\vartheta+1}-\Lambda,   \label{e:chaprho2} \\
	\intertext{and}
	s(\alpha)=&\lambda(\alpha)=\sqrt{p^{\prime}(\rho)}
	= \frac{\vartheta \Lambda^{1+\vartheta}}{\Lambda^{\vartheta+1}+ [(\bar{\rho}(1)+\Lambda)^{\vartheta+1}-\Lambda^{\vartheta+1}]e^{(\vartheta+1)\alpha}}  \notag \\
	>&\frac{\vartheta}{1+[(\frac{\bar{\rho}(1)}{\Lambda}+1)^{\vartheta+1}-1]e^{(\vartheta+1)K_1}}  >0, \label{e:chapslm}
\end{align}
for $|\alpha|\leq K_1$ (recall the fact that $\alpha \in I_\alpha$ and $I_\alpha$ is compact). Then,  \eqref{e:chaprho}, \eqref{e:chaprho2} and \eqref{e:chppar3}, with the help of mean value theorem, imply
\begin{align*}
	\varrho(\tau,\alpha-\bar{\alpha})=&\tau^{-3(1+\vartheta)}(\mu(\alpha)-\mu(\bar{\alpha}))  \notag \\
	= & e^{-(\vartheta+1)\bar{\alpha}} \Lambda \bigl[\bigl(1+ K_3 e^{(\vartheta+1)\alpha}\bigr)^\frac{1}{\vartheta+1} - \bigl(1+ K_3  e^{(\vartheta+1)\bar{\alpha}}\bigr)^\frac{1}{\vartheta+1} \bigr]  \notag \\
	= &    \Lambda (1+K_3 \tau^{3(1+\vartheta)} e^{(\vartheta+1) K_2 (\alpha-\bar{\alpha}) })^{\frac{1}{\vartheta+1}-1} K_3 e^{(\vartheta+1) K_2(\alpha-\bar{\alpha}) }(\alpha-\bar{\alpha})
\end{align*}
for some $K_2\in (0,1)$,
where $K_3=\bigl[\bigl(1+\frac{\bar{\rho}(1)}{\Lambda}\bigr)^{\vartheta+1}-1\bigr]$. Therefore,
$\varrho\in C\bigl([0,1], C^\omega( \mathbb{R})\bigr)$ and $\varrho(\tau,0)=0$.

Direct calculations using \eqref{e:intbg}, \eqref{e:chppar3}, with the help of that (by \eqref{e:chapalp} and \eqref{e:chapslm})
\begin{align*}
	\bar{s}=s(\bar{\alpha})=\bar{\lambda}=\lambda(\bar{\alpha})= \frac{\vartheta \Lambda^{1+\vartheta}}{\Lambda^{\vartheta+1}+\tau^{3(1+\vartheta)}[(\bar{\rho}(1)+\Lambda)^{\vartheta+1}-\Lambda^{\vartheta+1}] } <\vartheta\lesssim \beta   
\end{align*}
and \begin{align*}
	\lambda^\prime(\bar{\alpha})=-\frac{(1+\vartheta)\vartheta \Lambda^{1+\vartheta}[(\bar{\rho}(1)+\Lambda)^{1+\vartheta}-\Lambda^{1+\vartheta}]\tau^{3 (1+\vartheta)}}{\bigl[\Lambda^{\vartheta+1}+\tau^{3(1+\vartheta)}[(\bar{\rho}(1)+\Lambda)^{\vartheta+1}-\Lambda^{\vartheta+1}] \bigr]^2}
\end{align*}
show that
\begin{gather*}
\frac{d\mu(\alpha)}{d\alpha}= \frac{\lambda(\alpha)(\rho+p)}{s(\alpha)} = \rho+p,  \\
\mu(\alpha)-\mu(\bar{\alpha})= \tau^{3(1+\vartheta)} \varrho(\tau, \alpha-\bar{\alpha})  , 
\end{gather*}
and $\lambda^\prime (\bar{\alpha})\partial_\tau \beta(\tau) \equiv 0$,
\begin{align*}
	\frac{\bar{s}}{\tau}  \lambda^\prime (\bar{\alpha})= & -\frac{(1+\vartheta)\vartheta^2 \Lambda^{2(1+\vartheta)}[(\bar{\rho}(1)+\Lambda)^{1+\vartheta}-\Lambda^{1+\vartheta}]\tau^{ 3(1+\vartheta)-1}}{\bigl[\Lambda^{\vartheta+1}+\tau^{3(1+\vartheta)}[(\bar{\rho}(1)+\Lambda)^{\vartheta+1}-\Lambda^{\vartheta+1}] \bigr]^3}
	\lesssim   1.
\end{align*}
This is enough to conclude Assumptions \ref{a:Maksym}, \ref{a:tLip} and \eqref{e:B0bd} in Assumption \ref{a:postvty}.
Note $q\equiv 1$,
and because $0<\vartheta<\sqrt{\frac{1}{3}}$,
\begin{align*}
1-3\bar{s}^2 = & 1-3\biggl(\frac{\vartheta \Lambda^{1+\vartheta}}{\Lambda^{\vartheta+1}+\tau^{3(1+\vartheta)}[(\bar{\rho}(1)+\Lambda)^{\vartheta+1}-\Lambda^{\vartheta+1}] }\biggr)^2 \notag \\
\geq & 1-3 \vartheta ^2 
>0
\end{align*}
for $\tau \in [0,1]$, which verifies Assumption \ref{a:postvty}. Therefore, we conclude the generalized Chaplygin gases belong to \textit{Fluids} $(II)$.

\section{Singular hyperbolic formulations of the conformal Einstein equations}\label{section:2}

This section contributes to the singular hyperbolic formulations of the conformal Einstein equations by using the method initiated in \cite{Oliynyk2016a} (and applied in \cite{LeFloch2015a,Liu2017,Liu2018}). Readers who are familiar with this approach can quickly browse but pay more attentions to the next section on conformal Euler equations which involve main difficulties and contribute to the main innovations of this article.

\subsection{The reduced conformal Einstein equations}\label{section:3.2}

With the help of the wave coordinates $Z^{\mu}$ defined by \eqref{4.5}, we transform conformal Einstein equation to the gauge reduced one,  
\begin{align} \label{4.6}
-2R^{\mu\nu} +&2\nabla^{(\mu}Z^{\nu)}+A_{\kappa}^{\mu\nu}Z^{\kappa}=-4\nabla^{\mu}\nabla^{\nu}\Phi+4\nabla^{\mu}\Phi\nabla^{\nu}\Phi\nonumber\\
 -&2\left[\Box_{g}\Phi+2|\nabla\Phi|^{2}_{g}+(\frac{\rho-p}{2}+\Lambda)e^{2\Phi}\right]g^{\mu\nu} -2e^{2\Phi}(\rho+p)u^{\mu}u^{\nu},
\end{align}
where
\begin{equation*}
A^{\mu\nu}_{\kappa}=-X^{(\mu}\delta^{\nu)}_{\kappa}+Y^{(\mu}\delta^{\nu)}_{\kappa}.
\end{equation*}
Direct calculations yield
\begin{align}
2\nabla^{(\mu}Z^{\nu)}
&= 2\nabla^{(\mu}\Gamma^{\nu)}+\partial_{\tau}\left(\frac{2\psi(\tau)}{3\tau}\right)(g^{\mu 0}\delta^{\nu}_{0}+g^{\nu0}\delta^{\mu}_{0})
-\frac{2\psi(\tau)}{3\tau}
\partial_{\tau}g^{\mu\nu}-4\nabla^{\mu}\nabla^{\nu}\Phi  \label{e:NZ}
\intertext{and}
A_{k}^{\mu\nu}Z^{\kappa}
&= -\Gamma^{(\mu}\Gamma^{\nu)}+4\nabla^{\mu}\Phi\nabla^{\nu}\Phi-\frac{4\psi(\tau)}{3\tau}(\nabla^{\mu}\Phi\delta^{\nu}_{0}+\nabla^{\nu}\Phi\delta^{\mu}_{0})
+\frac{4\Lambda^{2}(\tau)}{9\tau^{2}}\delta^{\mu}_{0}\delta^{\nu}_{0}.\label{e:AZ}
\end{align}
Reduced Einstein equations \eqref{4.6} and \eqref{e:NZ}--\eqref{e:AZ} lead to
\begin{align}\label{e:Rein1}
-2R^{\mu\nu} +&2\nabla^{(\mu}\Gamma^{\nu)}-\Gamma^{\mu}\Gamma^{\nu}=\frac{2\psi(\tau)}{3\tau}\partial_{\tau}g^{\mu\nu}
-\frac{2\partial_{\tau}\psi(\tau)}{3\tau}
(g^{\mu0}\delta^{\nu}_{0}+g^{\nu0}\delta^{\mu}_{0})\nonumber\\
 -&\frac{4\psi(\tau)}{3\tau^{2}}\left(g^{00}+\frac{\psi(\tau)}{3}\right)\delta^{\mu}_{0}\delta^{\nu}_{0}-
\frac{4\psi(\tau)}{3\tau^{2}}g^{0i}\delta_{0}^{(\mu}\delta^{\nu)}_{i}-\frac{2}{\tau^{2}}g^{\mu\nu}\left(g^{00}+\frac{\psi(\tau)}{3}\right)\nonumber\\
 +&\frac{2}{\tau^{2}}\left(\psi(\tau)-\Lambda-\frac{\rho-p}{2}\right)g^{\mu\nu}
-2e^{2\Phi}(\rho+p)u^{\mu}u^{\nu}.
\end{align}
Recalling the formula (e.g., see \cite{Friedrich2000, Ringstroem2009})
\begin{align} \label{E:REPOFR}
R^{\mu\nu}=\frac{1}{2}g^{\lambda\sigma}\ \partial_\lambda \partial_\sigma g^{\mu\nu}+ \nabla^{(\mu} \Gamma^{\nu)}+\frac{1}{2}(Q^{\mu\nu}-X^\mu X^\nu),
\end{align}
where
\begin{align} \label{E:QMUNU}
Q^{\mu\nu}(g,\partial g)  =&\bar{g}^{\lambda\sigma}\bar{\partial}_\lambda(\bar{g}^{\alpha\mu}\bar{g}^{\rho\nu})
\bar{\partial}_\sigma\bar{g}_{\alpha\rho}
+2\bar{g}^{\alpha\mu}\bar{\Gamma}^\eta_{\lambda\alpha}\bar{g}_{\eta\delta}\bar{g}^
{\lambda\gamma}\bar{g}^{\rho\nu}\bar{\Gamma}^\delta_{\rho
	\gamma}  \notag \\
 & + 4\bar{\Gamma}^\lambda_{\delta\eta}
\bar{g}^{\delta\gamma}\bar{g}_{\lambda(\alpha}\bar{\Gamma}^\eta_{\rho)
	\gamma}\bar{g}^{\alpha\mu}\bar{g}^{\rho\nu}
+(\bar{\Gamma}^\mu
-\bar{\gamma}^\mu)(\bar{\Gamma}^\nu-\bar{\gamma}^\nu),
\end{align}
and inserting \eqref{E:REPOFR} and \eqref{E:QMUNU} into \eqref{e:Rein1}, with the help of \eqref{e:lamtau} (that is,  $\frac{\psi(\tau)}{3}=\omega^{2}+\frac{\Omega}{2}$), yields
\begin{align*}
-g^{\kappa\lambda}\partial_{\kappa}\partial_{\lambda}g^{\mu\nu} =&\frac{2\omega^{2}}{\tau}\partial_{\tau}g^{\mu\nu}-\frac{4\omega^{2}}{\tau^{2}}(g^{00}+\omega^{2})
\delta^{\mu}_{0}\delta^{\nu}_{0}
-\frac{4\omega^{2}}{\tau^{2}}g^{0i}\delta_{0}^{(\mu}\delta^{\nu)}_{i}
\nonumber\\
 &-\frac{2}{\tau^{2}}g^{\mu\nu}(g^{00}+\omega^{2})
-\frac{2\Omega}{\tau^{2}}(g^{00}+\omega^{2}+\frac{\Omega}{2})\delta^{\mu}_{0}\delta^{\nu}_{0}
-\frac{\Omega}{\tau}\partial_{\tau}g^{\mu\nu}
\nonumber\\
 &-\frac{2}{\tau^{2}}\omega^{2}\Omega\delta^{\mu}_{0}\delta^{\nu}_{0}
-\frac{2\Omega}{\tau^{2}}g^{0i}\delta_{0}^{(\mu}\delta^{\nu)}_{i}-\frac{\Omega}{\tau^{2}}g^{\mu\nu}
-\frac{2\partial_{\tau}\psi(\tau)}{3\tau}(g^{\mu0}\delta^{\nu}_{0}+g^{\nu0}\delta^{\mu}_{0})\nonumber\\
 &+\frac{2}{\tau^{2}}\left(3\omega^{2}+\Omega -\Lambda-\frac{\rho-p}{2}\right)g^{\mu\nu}\nonumber\\
 &-2e^{2\Phi}(\rho+p)u^{\mu}u^{\nu}+Q^{\mu\nu}(g,\partial g) ,
\end{align*}
which is equivalent to
\begin{align}\label{4.7}
-g^{\kappa\lambda}\partial_{\kappa}\partial_{\lambda}(g^{\mu\nu}-\eta^{\mu\nu})
 =&\frac{2\omega^{2}}{\tau}\partial_{\tau}(g^{\mu\nu}-\eta^{\mu\nu})-\frac{4\omega^{2}}{\tau^{2}}(g^{00}+\omega^{2})\delta^{\mu}_{0}\delta^{\nu}_{0}\nonumber\\
& -\frac{4\omega^{2}}{\tau^{2}}g^{0i}\delta_{0}^{(\mu}\delta^{\nu)}_{i}-\frac{2}{\tau^{2}}g^{\mu\nu}(g^{00}+\omega^{2})+\mathfrak{H}^{\mu\nu},
\end{align}
where
\begin{align*}
\mathfrak{H}^{\mu\nu}
 =&(g^{\kappa\lambda}-\eta^{\kappa\lambda})\partial_{\kappa}\partial_{\lambda}\eta^{\mu\nu}-\frac{2\Omega}{\tau^{2}}(g^{00}+\omega^{2})\delta^{\mu}_{0}\delta^{\nu}_{0}
-\frac{\Omega}{\tau}\partial_{\tau}(g^{\mu\nu}-\eta^{\mu\nu})-\frac{2\Omega}{\tau^{2}}g^{0i}\delta_{0}^{(\mu}\delta^{\nu)}_{i}\nonumber\\
 &-\frac{\Omega}{\tau^{2}}(g^{\mu\nu}-\eta^{\mu\nu})-\frac{2\partial_{\tau}\psi(\tau)}{3\tau}\left((g^{\mu0}-\eta^{\mu0})\delta^{\nu}_{0}+(g^{\nu0}-\eta^{\nu0})\delta^{\mu}_{0}\right)\nonumber\\
 &-\frac{1}{\tau^2}  (\rho-\bar{\rho}-p+\bar{p}) (g^{\mu\nu}-\eta^{\mu\nu})
-\frac{2}{\tau^2}\left(\frac{\rho-\bar{\rho}-(p-\bar{p})}{2}\right)\eta^{\mu\nu}\nonumber\\
 &-\frac{2}{\tau^2}\Big[(\rho-\bar{\rho}+p-\bar{p})u^{\mu}u^{\nu}+
(\bar{\rho}+\bar{p})(u^{\mu}u^{\nu}-\bar{u}^{\mu}\bar{u}^{\nu})\Big]\nonumber\\
 &+Q^{\mu\nu}(g,\partial g)-Q^{\mu\nu}(\eta,\partial\eta),
\end{align*}
and we have applied the identity $3\omega^2-\Lambda+ \Omega -\frac{\rho-p}{2}=-\frac{1}{2}(\rho-\bar{\rho}-p+\bar{p})$ due to \eqref{3.20}.

\begin{remark}\label{remark:3.2}
Note that, by \eqref{E:QMUNU}, $Q^{\mu\nu}(g,\partial g)$ are quadratic in $\partial g=(\partial_{\kappa}g^{\mu\nu})$ and analytical in $g=(g^{\mu\nu})$.
\end{remark}
\begin{remark}
	Note that terms with $\frac{2\partial_{\tau}\psi(\tau)}{3\tau}$ is regular in $\tau$ due to the estimate \eqref{e:dtlambda}.
\end{remark}





Definitions \eqref{4.12}--\eqref{4.14} imply
\begin{align}\label{e:ug}
	g^{0\mu}-\eta^{0\mu}=2\tau \textbf{u}^{0\mu}, \quad
	\partial_{i}(g^{0\mu}-\eta^{0\mu})=\textbf{u}_{i}^{0\mu} \AND \partial_{\tau}(g^{0\mu}-\eta^{0\mu})=\textbf{u}_{0}^{0\mu}+3\textbf{u}^{0\mu}.
\end{align}
Then, by letting $\nu=0$, with the help of \eqref{e:ug}, the equation \eqref{4.7} becomes
\begin{align}\label{4.19}
 &-g^{00}\partial_{\tau} \textbf{u}_{0}^{0\mu}-2g^{0i}\partial_{i}\textbf{u}_{0}^{0\mu}-g^{ij}\partial_{j}\textbf{u}_{i}^{0\mu}\nonumber\\
 =&\frac{1}{\tau}\left[-\frac{g^{00}}{2}(\textbf{u}_{0}^{0\mu}+\textbf{u}^{0\mu})\right]+6\textbf{u}^{0i}\textbf{u}^{0\mu}_i+4\textbf{u}^{00}\textbf{u}_{0}^{0\mu}
-4\textbf{u}^{00}\textbf{u}^{0\mu}+\mathfrak{H}^{0\mu}.
\end{align}
Differentiating \eqref{4.14}, which is the definition of $\textbf{u}_{j}^{0\mu}$, implies
\begin{align*}
\partial_{\tau}\textbf{u}_{j}^{0\mu} =\partial_{\tau}\partial_{j}(g^{0\mu}-\eta^{0\mu})
=\partial_{j}(\textbf{u}_{0}^{0\mu}+3\textbf{u}^{0\mu})
=\partial_{j}\textbf{u}_{0}^{0\mu}+\frac{3}{2\tau}\textbf{u}_{j}^{0\mu},
\end{align*}
namely
\begin{equation*}
\partial_{\tau}\textbf{u}_{j}^{0\mu}-\partial_{j}\textbf{u}_{0}^{0\mu}=\frac{3}{2\tau}\textbf{u}_{j}^{0\mu}.
\end{equation*}
Then, differentiating \eqref{4.12} yields
\begin{equation}\label{e:dtu1}
\partial_{\tau}\textbf{u}^{0\mu}=\partial_{\tau}\left(\frac{g^{0\mu}-\eta^{0\mu}}{2\tau}\right)=\frac{\textbf{u}_{0}^{0\mu}+3\textbf{u}^{0\mu}}{2\tau}-\frac{g^{0\mu}-\eta^{0\mu}}{2\tau^2}
=\frac{1}{2\tau}(\textbf{u}_{0}^{0\mu}+\textbf{u}^{0\mu}).
\end{equation}

For the spatial components, a more delicate transformation is applied to the $\mu=i, \nu=j$ components of \eqref{4.7} in order to rewrite those equations into the desired singular hyperbolic form. The first step is
to contract $(\mu,\nu)=(i,j)$ components of \eqref{4.7} with $\check{g}_{ij}$, where we recall
that $(\check{g}_{pq})=(g^{pq})^{-1}$. A straightforward calculation,
using the identity $\check{g}_{pq} \partial_\mu g^{pq}=-\det(g^{pq})  \partial_\mu \det(g^{pq})^{-1}$ and \eqref{4.7}
with $(\mu,\nu)=(0,0)$, recalling the definition of $\textbf{q}$ in \eqref{e:q}, which is
\begin{equation*}
\textbf{q}=g^{00}-\eta^{00}+\frac{\eta^{00}}{3}\ln(\det(g^{ij})),
\end{equation*}
leads to
\begin{equation}\label{4.22}
\partial_{\lambda}\textbf{q}=\partial_{\lambda}(g^{00}+\omega^{2})
-\frac{\omega^{2}}{3}\check{g}_{pq}\partial_{\lambda}g^{pq}-\frac{2\omega\delta_{\lambda}^{0}\partial_{\tau}\omega}{3}\ln(\det(g^{pq})).
\end{equation}
Then differentiating \eqref{4.22}, we arrive at
\begin{align}\label{4.23}
\partial_{\kappa}\partial_{\lambda}\textbf{q}
 =\partial_{\kappa}\partial_{\lambda}(g^{00}+\omega^{2})-\frac{\omega^{2}}{3}\check{g}_{pq}\partial_{\kappa}\partial_{\lambda}g^{pq}+\mathfrak{F} _{\kappa\lambda},
\end{align}
where
\begin{align*}
\mathfrak{F} _{\kappa\lambda} =& -\frac{\omega^{2}}{3}\partial_{\kappa}\check{g}_{pq}\partial_{\lambda}g^{pq}
-\frac{2\omega\partial_{\tau}\omega\delta^{0}_{\kappa}}{3}\check{g}_{pq}\partial_{\lambda}g^{pq}
-\frac{2(\partial_{\tau}\omega)^{2}}{3}\delta^{0}_{\lambda}\delta^{0}_{\kappa}\ln(\det(g^{pq}))\nonumber\\
 &-\frac{2\omega\partial^{2}_{\tau}\omega}{3}\delta^{0}_{\lambda}\delta^{0}_{\kappa}\ln(\det(g^{pq}))
-\frac{2\omega\partial_{\tau}\omega}{3}\delta^{0}_{\lambda}\check{g}_{pq}\partial_{\kappa}g^{pq}.
\end{align*}
Contracting \eqref{4.23} by $-g^{\kappa\lambda}$ and substituting the corresponding second derivative terms by \eqref{4.7}, with the help of \eqref{4.17}--\eqref{4.18}, result the equation
\begin{align}\label{4.24}
-g^{\kappa\lambda}\partial_{\kappa}\partial_{\lambda}\textbf{q} =&-g^{\kappa\lambda}\partial_{\kappa}\partial_{\lambda}(g^{00}+\omega^{2})
+\frac{\omega^{2}}{3}\check{g}_{pq}g^{\kappa\lambda}
\partial_{\kappa}\partial_{\lambda}g^{pq}-g^{\kappa\lambda}\mathfrak{F} _{\kappa\lambda}\nonumber\\
 =&\frac{2\omega^{2}}{\tau}\partial_{\tau}(g^{00}+\omega^{2})-\frac{4\omega^{2}}{\tau^{2}}(g^{00}+\omega^{2})-\frac{2}{\tau^{2}}g^{00}(g^{00}+\omega^{2})+\mathfrak{H}^{00}\nonumber\\
 &-\frac{\omega^{2}}{3}\check{g}_{pq}\Bigl(\frac{2\omega^{2}}{\tau}\partial_{\tau}g^{pq}-\frac{2}{\tau^{2}}g^{pq}(g^{00}+\omega^{2})+\mathfrak{H}^{pq}\Bigr)-g^{\kappa\lambda}\mathfrak{F} _{\kappa\lambda}\nonumber\\
 =&-\frac{2}{\tau}g^{00}\partial_{\tau}\textbf{q}+4\textbf{u}^{00}\partial_{\tau}\textbf{q}-8(\textbf{u}^{00})^{2}+\mathfrak{F},
\end{align}
where
\begin{equation*}
\mathfrak{F}  =\mathfrak{H}^{00}-\frac{\omega^{2}}{3}\check{g}_{pq}\mathfrak{H}^{pq}
+\frac{4\omega^{3}\partial_{\tau}\omega}{3\tau}\ln(\det(g^{pq}))-g^{\kappa\lambda}\mathfrak{F} _{\kappa\lambda}.
\end{equation*}
Furthermore, by \eqref{4.18}, the equation \eqref{4.24} is equivalent to
\begin{align}\label{4.25}
 -g^{00}\partial_{\tau}\textbf{u}_{0}-2g^{0i}\partial_{i}\textbf{u}_{0}-g^{ij}\partial_{i}\textbf{u}_{j} = -\frac{2}{\tau}g^{00}\textbf{u}_{0}+4\textbf{u}^{00}\textbf{u}_{0}-8(\textbf{u}^{00})^{2}+\mathfrak{F} .
\end{align}

Using \eqref{4.17}--\eqref{4.18} and differentiating them imply that
\begin{align}
\partial_{\tau}\textbf{u}_{j}=&\partial_{\tau}\partial_{j}\textbf{q}=\partial_{j}\partial_{\tau}\textbf{q}=\partial_{j}\textbf{u}_{0},  \label{e:dtuj3}
\intertext{and}
\partial_{\tau}\textbf{u}=&\partial_{\tau}\textbf{q}=\textbf{u}_{0}.\label{e:dtu3}
\end{align}

Introduce an operator
\begin{equation}\label{e:Lop}
\textbf{L}_{lm}^{ij}=\delta^{i}_{l}\delta^{j}_{m}-\frac{1}{3}\check{g}_{lm}g^{ij}.
\end{equation}
Direct calculations yield
\begin{equation}
\partial_{\mu}\textbf{g}^{ij}=(\det(\check{g}_{pq}))^{\frac{1}{3}}\textbf{L}_{lm}^{ij}\partial_{\mu}g^{lm} \AND \textbf{L}_{lm}^{ij}g^{lm}=0.\label{4.11}
\end{equation}
Let us turn to  $\textbf{g}^{ij}-\delta^{ij}$. Applying $(\det(\check{g}_{pq}))^{\frac{1}{3}}\textbf{L}_{lm}^{ij}$ to \eqref{4.7} with $(\mu,\nu)=(l,m)$, with the help of \eqref{e:Lop}--\eqref{4.11}, direct calculations give
\begin{align}\label{4.20}
-g^{\kappa\lambda}\partial_{\kappa}\partial_{\lambda}\textbf{g}^{ij} =&-g^{\kappa\lambda}\partial_{\kappa}[(\det(\check{g}_{pq}))^{\frac{1}{3}}\textbf{L}_{lm}^{ij}\partial_{\lambda}g^{lm}]\nonumber\\
 =&(\det(\check{g}_{pq}))^{\frac{1}{3}}\textbf{L}_{lm}^{ij}(-g^{\kappa\lambda}\partial_{\kappa}\partial_{\lambda}g^{lm})-g^{\kappa\lambda}\partial_{\kappa}[(\det(\check{g}_{pq}))^{\frac{1}{3}}
\textbf{L}_{lm}^{ij}]\partial_{\lambda}g^{lm}\nonumber\\
 =&(\det(\check{g}_{pq}))^{\frac{1}{3}}\textbf{L}_{lm}^{ij}\left(\frac{2\omega^{2}}{\tau}\partial_{\tau}(g^{lm}-\eta^{lm})
-\frac{2}{\tau^{2}}g^{lm}(g^{00}+\omega^{2})+\mathfrak{H}^{lm}\right)\nonumber\\
 &-g^{\kappa\lambda}\partial_{\kappa}[(\det(\check{g}_{pq}))^{\frac{1}{3}}
\textbf{L}_{lm}^{ij}]\partial_{\lambda}g^{lm}\nonumber\\
 =&\frac{2\omega^{2}}{\tau}\partial_{\tau}\textbf{g}^{ij}+\mathfrak{R}^{ij},
\end{align}
where
\begin{align*}
\mathfrak{R}^{ij}=(\det(\check{g}_{pq}))^{\frac{1}{3}}\textbf{L}_{lm}^{ij}\mathfrak{H}^{lm}-
g^{\kappa\lambda}\partial_{\kappa}[(\det(\check{g}_{pq}))^{\frac{1}{3}}
\textbf{L}_{lm}^{ij}]\partial_{\lambda}g^{lm}.
\end{align*}
Then, inserting \eqref{4.15}--\eqref{4.16} into \eqref{4.20}, we arrive at
\begin{align}\label{4.21}
-g^{00}\partial_{\tau}\textbf{u}_{0}^{ij}-2g^{0i}\partial_{i}\textbf{u}_{0}^{ij}-g^{pq}\partial_{p}\textbf{u}_{q}^{ij}
= -\frac{2}{\tau}g^{00}\textbf{u}_{0}^{ij}+4\textbf{u}^{00}\textbf{u}_{0}^{ij}+\mathfrak{R}^{ij}.
\end{align}

Differentiating \eqref{4.16} and \eqref{4.15} also gives
\begin{align}
\partial_{\tau}\textbf{u}_{j}^{lm}=&\partial_{\tau}\partial_{j}\textbf{g}^{lm}=\partial_{j}\partial_{\tau}\textbf{g}^{lm}=\partial_{j}\textbf{u}_{0}^{lm} \label{e:dtuj2}
\intertext{and}
\partial_{\tau}\textbf{u}^{lm}=&\partial_{\tau}(\textbf{g}^{lm}-\delta^{lm})=\partial_{\tau}\textbf{g}^{lm}=\textbf{u}_{0}^{lm}.\label{e:dtu2}
\end{align}

Collecting \eqref{4.19}--\eqref{e:dtu1}, \eqref{4.21}--\eqref{e:dtu2} and \eqref{4.25}--\eqref{e:dtu3} together, we  transform the reduced conformal Einstein equations into the following singular symmetric hyperbolic formulation,
\begin{align}\label{4.26}
A^{\kappa}\partial_{\kappa}\left(\begin{array}{ll}
\textbf{u}^{0\mu}_{0}\\\textbf{u}^{0\mu}_{j}\\\textbf{u}^{0\mu}
\end{array}\right)=\frac{1}{\tau}\textbf{A}\textbf{P}^\star\left(\begin{array}{ll}
\textbf{u}^{0\mu}_{0}\\\textbf{u}^{0\mu}_{j}\\\textbf{u}^{0\mu}
\end{array}\right)+F_1,
\end{align}
\begin{align}\label{4.27}
A^{\kappa}\partial_{\kappa}\left(\begin{array}{ll}
\textbf{u}^{lm}_{0}\\ \textbf{u}^{lm}_{j}\\ \textbf{u}^{lm}
\end{array}\right)=\frac{1}{\tau}(-2g^{00})\Pi\left(\begin{array}{ll}
\textbf{u}^{lm}_{0}\\ \textbf{u}^{lm}_{j}\\ \textbf{u}^{lm}
\end{array}\right)+F_2,
\end{align}
and
\begin{align}\label{4.28}
A^{\kappa}\partial_{\kappa}\left(\begin{array}{ll}
\textbf{u}_{0}\\ \textbf{u}_{j}\\ \textbf{u}
\end{array}\right)=\frac{1}{\tau}(-2g^{00})\Pi\left(\begin{array}{ll}
\textbf{u}_{0}\\ \textbf{u}_{j}\\ \textbf{u}
\end{array}\right)+F_3,
\end{align}
where
\begin{align*}
A^{0}=\left(
\begin{array}{ccc}
-g^{00}& 0& 0\\
0& g^{ij}& 0\\
0&0& -g^{00}
\end{array}\right),\quad
A^{k}=\left(
\begin{array}{ccc}
-2g^{0k}& -g^{jk}& 0\\
-g^{ik}& 0& 0\\
0&0& 0
\end{array}\right),
\end{align*}
\begin{align*}
\textbf{P}^\star=\left(
\begin{array}{ccc}
\frac{1}{2}& 0& \frac{1}{2}\\
0& \delta^{j}_{l}& 0\\
\frac{1}{2}&0& \frac{1}{2}
\end{array}\right),\quad
\textbf{A}=\left(
\begin{array}{ccc}
-g^{00}& 0& 0\\
0& \frac{3}{2}g^{li}& 0\\
0&0& -g^{00}
\end{array}\right),
\end{align*}
\begin{align*}
\Pi=\left(
\begin{array}{ccc}
1& 0& 0\\
0& 0& 0\\
0&0& 0\end{array}\right),\quad
F_1=\left(
\begin{array}{ccc}
6\textbf{u}^{0i}\textbf{u}^{0\mu}_i+4\textbf{u}^{00}\textbf{u}_{0}^{0\mu}-4\textbf{u}^{00}\textbf{u}^{0\mu}+\mathfrak{H}^{0\mu}\\
0\\
0
\end{array}\right),
\end{align*}
and
\begin{align*}
F_2=\p{
4\textbf{u}^{00}\textbf{u}_{0}^{ij}+\mathfrak{R}^{ij}\\
0\\
-g^{00}\textbf{u}_{0}^{lm} },\quad
F_3=\p{
4\textbf{u}^{00}\textbf{u}_{0}-8(\textbf{u}^{00})^{2}+\mathfrak{F} \\
0\\
-g^{00}\textbf{u}_{0} }.
\end{align*}


\subsection{Regular remainder terms }\label{section:3.4}
In order to apply the singular symmetric hyperbolic system of Theorem \ref{pro:3.16}, we have to verify that $\mathfrak{H}^{0\mu}$, $\mathfrak{H}^{ij}$, $\mathfrak{R}^{ij}$, $\mathfrak{F} $ belong to $C^0([0,1], C^\infty(\mathbb{V}))$. The similar examinations and the key calculations have appeared in \cite{LeFloch2015a,Liu2017,Liu2018}\footnote{Let $\epsilon=v_T/c \equiv 1$
	. For more details please refer to \cite{Liu2017,Liu2018}. }. We only state the main ideas and list the key results which are minor variations of the corresponding quantities in  \cite{Liu2017,Liu2018} and the similar calculations apply.

Let us denote
\begin{equation}\label{e:UV}
	\textbf{U}: =(\textbf{u}^{0\mu}_{0},\textbf{u}^{0\mu}_{i},\textbf{u}^{0\mu},\textbf{u}^{lm}_{0},\textbf{u}^{lm}_{i},\textbf{u}^{lm}, \textbf{u} _{0},\textbf{u} _{i},\textbf{u} )^{T} \AND \textbf{V}:=(\delta\zeta, v^i)
\end{equation}
and
\begin{equation}\label{e:bbV}
\mathbb{V} = \mathbb{R}^4\times\mathbb{R}^{12}\times \mathbb{R}^4\times\mathbb{S}_3\times  (\mathbb{S}_3)^3\times\mathbb{S}_3\times \mathbb{R}
\times\mathbb{R}^3\times\mathbb{R} \times \mathbb{R}\times\mathbb{R}^3,
\end{equation}

First note
	\begin{equation}\label{4.50}
	\det({g^{ij}})=e^{\frac{3(2\tau \textbf{u}^{00}-\textbf{u})}{\omega^{2}}},
	\end{equation}
then we have the following expansions
\begin{align}
g^{ij} =&\delta^{ij}+ \textbf{u}^{ij}+\frac{\eta^{ij}}{\omega^{2}}(2\tau \textbf{u}^{00}-\textbf{u})+\mathcal{S}^{ij}(\tau,\textbf{u},\textbf{u}^{\mu\nu} ),  \label{e:gij}\\
\partial_{\lambda} g^{ij} = &   \textbf{u}_{\lambda}^{ij}+\frac{\eta^{ij}}{\omega^{2}}\bigl(2\textbf{u}^{00}\delta_{\lambda}^{0}+2\tau \textbf{u}^{00}_{i}\delta_{\lambda}^{i}+(\textbf{u}_{0}^{00}+\textbf{u}^{00})\delta^{0}_{\lambda}-\textbf{u}_{\lambda}\bigr)  \notag
\\
 &+\frac{2\Omega\eta^{ij}\delta_{\lambda}^{0}}{\tau \omega^{4}}(2\tau\textbf{u}^{00}-\textbf{u})+\mathcal{S}_\lambda^{ij}(\tau,\textbf{u},\textbf{u}^{\mu\nu} , \textbf{u}^{\mu\nu}_\sigma, \textbf{u}_\sigma ),\\
g_{\mu\nu} = & \eta_{\mu\nu} +  \mathcal{S}_{\mu\nu}(\tau,\textbf{u},\textbf{u}^{\mu\nu} ),\\
 \check{g}_{ij}
= & \delta_{ij}+  \mathcal{S}_{ij}( \tau,\textbf{u},\textbf{u}^{\mu\nu}), \\
g^{0\mu}=&\eta^{0\mu}+2\tau\textbf{u}^{0\mu}, \label{e:g00exp}\\
\partial_{\tau} g^{\mu0}=&
\textbf{u}_{0}^{\mu0}+3\textbf{u}^{\mu0}+2\omega\partial_{\tau}\omega \delta^\mu_0,  \label{e:dtg00exp}  \\
\partial_\sigma g_{\mu\nu}= & \partial_\sigma \eta_{\mu\nu}+  \mathcal{S}_{\mu\nu\sigma}(\tau,  \textbf{u}^{\alpha\beta}, \textbf{u}, \textbf{u}^{\alpha\beta}_\gamma, \textbf{u}_\gamma),  \\
g_{00}=& \eta_{00} +\tau\mathcal{S}_{00}(\tau,\textbf{U}), \\
g_{0i}=&\tau \mathcal{S}_{0i}(\tau,\textbf{U}),  \\
u^{0}
=&-\sqrt{-\eta^{00}}+\tau\mathcal{S}(\tau,\textbf{U})
+\beta^2(\tau)\mathcal{W}(\tau,\textbf{U},\textbf{V})+\tau\beta(\tau) \mathcal{V}(\tau,\textbf{U},\textbf{V}), \label{e:u0up} \\
u_{0}
=&\frac{1}{\sqrt{-\eta^{00}}}+\tau\mathcal{S}(\tau,\textbf{U})
+ \beta^2(\tau)\mathcal{W}(\tau,\textbf{U},\textbf{V})+\tau\beta(\tau) \mathcal{V}(\tau,\textbf{U},\textbf{V}), \\
u_{i}=&
\beta(\tau)g_{ij}v^{j}+2\tau\textbf{u}^{0j}\frac{g_{kj}u^{0}}{g_{il}g^{0i}g^{0l}-g^{00}} \notag \\
=&  \beta(\tau)g_{ij}v^{j}+ \tau\mathcal{S}_i(\tau, \textbf{U}) + \beta^2(\tau)\mathcal{W}_i(\tau,\textbf{U},\textbf{V})+\tau\beta(\tau) \mathcal{V}_i(\tau,\textbf{U},\textbf{V}) . \label{e:ui}
\end{align}
where we recall that the upper case callgraphic letters above and below obey the convention in \S \ref{remainder} and they vanish in $\xi=0$. We also remark that because the exact forms of these callgraphic remainders are not important, the remainders of the same letter may change from line to line. In specific, in above \eqref{e:u0up}--\eqref{e:ui}, terms like $\mathcal{V}$ and $\mathcal{W}$ vanish in $\mathbf{V}=0$ as well.

The similar arguments to \cite[Proposition $2.3$]{Liu2017} can be applied to
\begin{equation}\label{e:Qdiff}
Q^{\mu\nu}(g,\partial g)-Q^{\mu\nu}(\eta,\partial\eta)=\mathcal{S}^{\mu\nu}(\tau,u^{\alpha\beta},u,u^{\alpha\beta}_\sigma, u_\sigma).
\end{equation}
By Assumption \ref{a:tLip}, we have identity
\begin{align}
	\frac{\rho-\bar{\rho}}{\tau^2}=\tau^{\varsigma-2}\varrho(\tau, \delta\zeta) \AND \frac{p-\bar{p}}{\tau^2}=c_s^2\bigl(\rho_{K_4}\bigr)\tau^{\varsigma-2}\varrho(\tau, \delta\zeta),  \label{e:t2rho}
\end{align}
where $\rho_{K_4}:=\bar{\rho}+K_4(\rho-\bar{\rho})$ for some constant $K_4$, where $\varrho$ is defined in Assumption \ref{a:tLip} and we have used the mean value theorem above.

By the definitions of field variables \eqref{4.12}--\eqref{e:velcty} and expansions \eqref{4.50} and  \eqref{e:gij}--\eqref{e:ui}, with the help of \eqref{e:pest}--\eqref{e:dtlambda}, \eqref{e:Qdiff} and \eqref{e:t2rho}, we conclude that $\mathfrak{H}^{\mu\nu}$, $\mathfrak{R}^{ij}$, $\mathfrak{F} $ belong to $C^0([0,1], C^\infty(\mathbb{V}))$ and can be expressed as
\begin{align}\label{e:remds1}
	\mathfrak{H}^{\mu\nu} + \mathfrak{R}^{ij}+\mathfrak{F}  =\texttt{S}(\tau,\textbf{u}^{\alpha\beta},\textbf{u},\textbf{u}^{\alpha\beta}_\sigma, \textbf{u}_\sigma, \delta\zeta, v^i), 
\end{align}
where $\texttt{S}(\tau,0,0,0,0,0,0)=0$.

\section{Singular hyperbolic formulations of the conformal Euler equations}\label{section:3}
The main goal of this section is to derive the conformal Euler equations \eqref{symmetric-Euler} and \eqref{lower-case}. The idea of this section is new and contributes to the main innovation of this paper.

\subsection{Formulation of conformal Euler equations for \textit{Fluids} $(I)$}\label{section:3.3}

In this section, we rewrite the conformal Euler equation \eqref{e:cEu1} as a \textit{non-degenerated singular symmetric hyperbolic system} for \textit{Fluids} $(I)$. To achieve this, a series of transformations have been applied to and we demonstrate them in the following four steps to clearly express the ideas of overcoming the difficulties.

\subsubsection{Step $1$: symmetric hyperbolic formulations}

Recall the conformal Euler equations \eqref{e:cEu1}
\begin{align*}
	\nabla_{\mu}\widetilde{T}^{\mu\nu} =&-6\widetilde{T}^{\mu\nu}\nabla_{\mu}\Phi+g_{\kappa\lambda}\widetilde{T}^{\kappa\lambda}g^{\mu\nu}\nabla_{\mu}\Phi.
\end{align*}
First note that differentiating \eqref{e:nol2} yields
\begin{align*}\label{e:unu1}
	u_\mu\nabla_\nu u^\mu=0  \quad \Bigl(\text{that is } \nabla_\nu u^0= - \frac{u_i}{u_0}\nabla_\nu u^i \Bigr),
\end{align*}
and define
\begin{equation}\label{e:Ldef}
	L_{i}^{\mu}=\delta_{i}^{\mu}-\frac{u_{i}}{u_{0}}\delta_{0}^{\mu} \AND L_{k\nu}=g_{\nu\lambda}L^\lambda_k.
\end{equation}
Then using the conformal fluid four velocity $u^\mu$ and $L_{k\nu}$ acting on above conformal Euler equations \eqref{e:cEu1}, respectively yields the following formulation of conformal Euler equations (for more details, we refer to \cite[\S $2.2$]{Oliynyk2015})
\begin{align}
u^{\mu}\partial_{\mu}\rho+(\rho+p)L_{i}^{\mu}\nabla_{\mu}u^{i}= & -3(\rho+p) u^{\mu}\nabla_{\mu}\Phi,\label{Euler-original1} \\
\frac{c_s^2L_{i}^{\mu}}{\rho+p}\partial_{\mu}\rho+M_{ki}u^{\mu}\partial_{\mu}u^{i}= & -L_{k}^{\mu}\partial_{\mu}\Phi, \label{Euler-original2}
\end{align}
where $c_s^{2}=\frac{dp}{d\rho}$
and
\begin{align}\label{e:Mdef}
M_{ki}=&g_{ki}-\frac{u_{i}}{u_{0}}g_{0k}-\frac{u_{k}}{u_{0}}g_{0i}+\frac{u_{i}u_{k}}{u_{0}^2}g_{00}.
\end{align}
Then, by the normalization
\begin{equation*}
g_{00}(u^{0})^2+2g_{0i}u^{0}u^{i}+g_{ij}u^{i}u^{j}=-1,
\end{equation*}
we are able to recover $u^0$,
\begin{equation*}\label{e:uup0}
u^{0}=-\frac{g_{0i}u^{i}-\sqrt{(g_{0i}u^{i})^2-g_{00}(g_{ij}u^{i}u^{j}+1)}}{g_{00}}.
\end{equation*}

\subsubsection{Step $2$: non-degenerated symmetric hyperbolic formulations}
We expect to rewrite \eqref{Euler-original1}--\eqref{Euler-original2} into a singular symmetric hyperbolic system. Once multiplying suitable factors, for example, $s^2/(\rho+p)$ on both sides of \eqref{Euler-original1} and $\rho+p$ on both sides of \eqref{Euler-original2}, then \eqref{Euler-original1}--\eqref{Euler-original2} become symmetric. However, since the behavior of the background density $\bar{\rho}\searrow 0$ and pressure $\bar{p}\searrow 0$, as $\tau\rightarrow 0$ (see \eqref{e:pest}--\eqref{e:bgrhoest1}), the expected behaviors of $\rho$ and $p$ due to small perturbation of initial data $\bar{\rho}(1)$ and $\bar{p}(1)$ also satisfy $\rho\searrow 0$ and $p\searrow 0$ as $\tau\rightarrow 0$, which leads to a fact that some of elements of the coefficient matrix of above symmetric hyperbolic system tend to $0$ or $\infty$ as $\tau\rightarrow 0$. The aim of our method is to use the Theorem \ref{pro:3.16} of the singular symmetric hyperbolic equation in Appendix \ref{section:3.5}, but the probably degenerated coefficients around $\tau=0$ violate the condition \eqref{c:5} in Appendix \ref{section:3.5}. Hence it is not possible to represent this system in a non-degenerate form by multiplying these factors.

In order to overcome this difficulty, we adopt and generalize the idea of the non-degenerated symmetrization of Euler equations originated by Makino
\cite{Makino1986} (the idea is also clearly stated in, for example, \cite{Brauer2014, Oliynyk2008a}).
The key point is to introduce a new density variable $\alpha$ defined by Assumption \ref{a:Maksym},
such that  \eqref{Euler-original1}--\eqref{Euler-original2} become
\begin{align}
u^{\mu}\frac{d\mu }{d\alpha}\partial_{\mu}\alpha+(\mu+\mu^*p)L_{i}^{\mu}\nabla_{\mu}u^{i}=&-3(\mu+\mu^*p) u^{\mu}\nabla_{\mu}\Phi,\label{Euler-transform1a} \\
M_{ki}u^{\mu}\partial_{\mu}u^{i}+\frac{s^2L_{i}^{\mu}}{\mu+\mu^*p}\frac{d\mu }{d\alpha}\partial_{\mu}\alpha=&-L_{k}^{\mu}\partial_{\mu}\Phi.  \label{Euler-transform1b}
\end{align}
Recalling the non-degenerate function $\lambda(\alpha)$ of the fluids defined in Assumption \ref{a:Maksym} and multiplying both sides of \eqref{Euler-transform1a} by $\lambda^2(\alpha)\frac{d(\mu^{-1})}{d\rho }$, with the help of that
$\frac{d\mu(\alpha)}{d\alpha}\frac{d\mu^{-1}(\rho)}{d\rho}=1$, we obtain
\begin{align}
\lambda^2u^{\mu}\partial_{\mu}\alpha+\lambda^2\frac{d(\mu^{-1})}{d\rho }(\mu+\mu^*p)L_{i}^{\mu}\nabla_{\mu}u^{i}=&-3(\mu+\mu^*p)\lambda^2\frac{d(\mu^{-1})}{d\rho } u^{\mu}\nabla_{\mu}\Phi, \label{Euler-transform2a}\\
M_{ki}u^{\mu}\partial_{\mu}u^{i}+\frac{s^2L_{i}^{\mu}}{\mu+\mu^*p}\frac{d\mu }{d\alpha}\partial_{\mu}\alpha=&-L_{k}^{\mu}\partial_{\mu}\Phi. \label{Euler-transform2b}
\end{align}
The relation \eqref{e:Maksym} in Assumption \ref{a:Maksym} which is equivalent to
\begin{equation}\label{e:Maksym2}
\frac{d\mu(\alpha)}{d\alpha}=\frac{\lambda(\alpha)(\mu+\mu^*p)(\alpha)}{s(\alpha)}
\end{equation}
ensures that the above system \eqref{Euler-transform2a}--\eqref{Euler-transform2b} is symmetric. This is because \eqref{e:Maksym2} implies
\begin{equation}\label{Makino-ODE}
 \lambda^2(\mu+\mu^*p)\frac{d(\mu^{-1})}{d\rho}=\frac{s^2 d\mu/d\alpha}{\mu+\mu^*p}=\lambda s.
\end{equation}
In view of \eqref{Makino-ODE}, equations \eqref{Euler-transform2a}--\eqref{Euler-transform2b} turn to
\begin{align}
\lambda^2u^{\mu}\partial_{\mu}\alpha+\lambda sL_{i}^{\mu}\nabla_{\mu}u^{i}=&-3\lambda s u^{\mu}\nabla_{\mu}\Phi, \label{Euler-transform3a}\\
\lambda sL_{i}^{\mu}\partial_{\mu}\alpha+M_{ki}u^{\mu}\nabla_{\mu}u^{k}=&-L_{k}^{\mu}\partial_{\mu}\Phi,  \label{Euler-transform3b}
\end{align}
which is a \textit{non-degenerated} symmetric hyperbolic equation and the coefficient matrix does not include the troublesome variables $\rho$ and $p$.

\begin{remark}
	In \cite{Oliynyk2015} (also see in \cite{Oliynyk2016a,LeFloch2015a,Liu2017,Liu2018}), the author took another non-degenerated symmetric hyperbolic formulation for \eqref{Euler-original1}--\eqref{Euler-original2}, in which the new density variable is defined by $\xi=\xi(\rho)=\int_{\rho(1)}^{\rho}\frac{dy}{y+p(y)}$. Under this variable transformation, \eqref{Euler-original1}--\eqref{Euler-original2} become
\begin{align}
s^2u^{\mu}\partial_{\mu}\xi+s^2L_{i}^{\mu}\nabla_{\mu}u^{i}=&-3s^2u^{\mu}\partial_{\mu}\Phi,\label{Euler-transform4a}\\
s^2L_{i}^{\mu}\partial_{\mu}\xi+M_{ij}u^{\mu}\nabla_{\mu}u^{j}=&-L_{i}^{\mu}\partial_{\mu}\Phi. \label{Euler-transform4b}
\end{align}
It is evident that \eqref{Euler-transform3a}--\eqref{Euler-transform3b} coincide with \eqref{Euler-transform4a}--\eqref{Euler-transform4b} by choosing $\lambda=s$ and $\alpha=\xi$ provided $s$ is non-degenerate (in fact, $s=\sqrt{K}$ in \cite{Oliynyk2015}).
\end{remark}

It is evident that \eqref{Euler-transform3a}--\eqref{Euler-transform3b} admit a homogeneous solution $(\alpha,u^\mu)=(\bar{\alpha}(\tau),-\omega\delta^\mu_0)$, which leads to
\begin{equation}\label{background-solution1}
\partial_{\tau}\bar{\alpha}=\frac{3}{\tau} \bar{q}.
\end{equation}
Subtracting  \eqref{Euler-transform3a}--\eqref{Euler-transform3b} by \eqref{background-solution1}, we obtain the formulation of conformal Euler equations as follows,
\begin{align}
\lambda^2u^{\mu}\partial_{\mu}(\alpha-\bar{\alpha})+\lambda sL_{i}^{\mu}\partial_{\mu}u^{i}=& \frac{3}{\tau}\lambda s u^0 -\frac{3}{\tau}\lambda^2u^0\bar{q}-\lambda s L_{i}^{\mu}\Gamma_{\mu\nu}^{i}u^{\nu},\label{Euler-uppera} \\
M_{ki}u^{\mu}\partial_{\mu}u^{k}+\lambda sL_{i}^{\mu}\partial_{\mu}(\alpha-\bar{\alpha})=& L_{i}^{0}\frac{1}{\tau}-\frac{3}{\tau}\lambda s L_{i}^{0}\bar{q}-M_{ki}u^{\mu}\Gamma_{\mu\nu}^{k}u^{\nu}. \label{Euler-upperb}
\end{align}

\subsubsection{Step $3$: non-degenerated symmetric hyperbolic formulations of new ``good'' variables} \label{s:stp3}
Recall  \eqref{e:densty}--\eqref{e:denstydif} which are
\begin{align}\label{e:resvar}
	\alpha =\beta(\tau)\zeta, \quad \bar{\alpha} =\beta(\tau)\bar{\zeta},\quad u^{i}= \beta(\tau)v^{i} \AND \delta\zeta = \zeta-\bar{\zeta},
\end{align}
then re-express \eqref{Euler-uppera}--\eqref{Euler-upperb} in terms of these new variables $(\delta\zeta, v^i)$. They become
\begin{align}
\lambda^2u^{\mu}\partial_{\mu}\delta\zeta+\lambda sL_{i}^{\mu}\partial_{\mu}v^{i}=&S ,\label{Euler-final1a} \\
\lambda sL_{i}^{\mu}\partial_{\mu}\delta\zeta+M_{ki}u^{\mu}\partial_{\mu}v^{k}=&S_{i},  \label{Euler-final1b}
\end{align}
where
\begin{align*}
S =&\frac{1}{\tau} \frac{1}{\beta} 3\lambda^2 u^{0}(q-\bar{q})-\frac{1}{\tau} \chi(\tau) \lambda^2u^{0}\delta\zeta\nonumber\\
&+\frac{1}{\tau}  \chi(\tau) \left( \frac{\lambda sg_{ij}\beta(\tau) v^{j}}{ u_{0}}\right)v^{i}
+\frac{\lambda s\beta^{\prime}(\tau)g_{0i}u^{0}}{\beta(\tau)u_{0}}v^{i}-\frac{1}{\beta(\tau)}(\lambda s L_{i}^{\mu}\Gamma_{\mu\nu}^{i}u^{\nu}), 
\intertext{and}
S_{i}=&\frac{1}{\tau}\left(-\frac{g_{ik}}{u_{0}}\Bigl(1- 3 s^2 \frac{\bar{q}}{q} -\chi(\tau) \beta(\tau) \lambda s \delta\zeta \Bigr)-\chi(\tau)  M_{ki}u^{0} \right)v^{k}\nonumber\\
&+\frac{1}{\tau}\left(-\frac{1}{\beta(\tau)}(1-3\bar{s}^2)\right)\frac{g_{0i}u^{0}}{u_{0}}
+\frac{1}{\tau}\left(\frac{3(\lambda s-\bar{\lambda}\bar{s}) \bar{s}}{\bar{\lambda}\beta(\tau)}+\chi(\tau)\lambda s \delta\zeta \right)\frac{g_{0i}u^{0}}{u_{0}}\nonumber\\
&-\frac{1}{\beta(\tau)}(M_{ki}u^{\mu}\Gamma_{\mu\nu}^{k}u^{\nu}). 
\end{align*}

\begin{remark}\label{R:chgvar1}
	The variable $\alpha-\bar{\alpha}$ in above system is not suitable for the purpose of formulating the target singular hyperbolic system  \eqref{e:model1} in Appendix \ref{section:3.5}, since, with the help of \eqref{e:rhopro}, the remainder terms including $(\rho-\bar{\rho})/\tau^2$ and $(p-\bar{p})/\tau^2$ in the Einstein equations are \textit{regular} in $\tau$ and analytic in $\delta\zeta$ instead of $\alpha-\bar{\alpha}$, see \eqref{e:t2rho}, and the remainders have been represented by \eqref{e:remds1} in terms of variables $(\tau,\textbf{u}^{\alpha\beta},\textbf{u},\textbf{u}^{\alpha\beta}_\sigma, \textbf{u}_\sigma, \delta\zeta, v^i)$. In other words, if one use $\alpha-\bar{\alpha}$ as the variable, the remainders of the Einstein equation contain the $\frac{1}{\tau^2}$-singular terms. After changing $\alpha-\bar{\alpha}$ to $\delta\zeta$, $u^i$ have to been changed to $v^i$ to make sure this system to be \textit{symmetric}. Hence, we have to change the variables of fluids from $(\alpha-\bar{\alpha},u^i)$ to $(\delta\zeta,v^i)$. In summary,  we emphasize that using $\delta\zeta$ instead of $\alpha-\bar{\alpha}$ as the variable of the target equation \eqref{e:model1} is due to the requirement of the regular remainders in the Einstein equations and adopting $v^i$ rather than $u^i$
is due to the requirement of the symmetric coefficient matrix of the Euler equations.
\end{remark}

Next, in order to apply the singular hyperbolic system \eqref{e:model1} in Appendix \ref{section:3.5}, we have to distinguish the singular (in $\tau$) and the regular terms in $S$ and $S_i$. To achieve this, we expand the following key quantities by direct calculations. With the help of \eqref{e:gij}--\eqref{e:ui},  \eqref{e:Mdef} becomes
\begin{align}\label{e:Mexp}
	M_{ik} = & g_{ki} +\beta^2(\tau) \mathcal{W}_{ki}(\tau,\textbf{U},\textbf{V}) +\beta^3(\tau) \mathcal{U}_{ki}(\tau,\textbf{U},\textbf{V}) \notag  \\
	&+\tau\beta(\tau) \mathcal{V}_{ki}(\tau,\textbf{U},\textbf{V}) +\tau  \mathcal{S}_{ki}(\tau,\textbf{U},\textbf{V}).
\end{align}
Note that $\tau \beta$, $\beta^3$ and  $\beta^2 \in C^1([0,1])$ (by Assumption \ref{a:tLip},  \ref{a:postvty} and \eqref{e:dtbt2}) which surely can be absorbed into the Calligraphic remainders.
However, later on we expect $\frac{1}{\tau}(M_{ik}-g_{ik})$ is regular in $\tau$ and since $\beta^2$, $\beta^3$ and $\tau\beta\lesssim \tau$ (due to the fact that $\beta (\tau)\lesssim \sqrt{\tau}$), it is better to expand $M_{ik}$ in this form. Similarly, using \eqref{e:gij}--\eqref{e:ui}, we expand
\begin{align}\label{e:expL}
L^0_i=&-\omega\beta\delta_{ij} v^j+\beta\mathcal{T}_i(\tau,\textbf{U}, \textbf{V})+\tau\mathcal{S}_i(\tau,\textbf{U}, \textbf{V})  \notag  \\
&+\beta^2\mathcal{W}_i(\tau,\textbf{U}, \textbf{V})+\tau\beta\mathcal{V}_i(\tau,\textbf{U}, \textbf{V}).
\end{align}
Calculating $\Gamma^i_{00}$ yields
\begin{align*}
	\Gamma^i_{00} = &  \frac{1}{2} g^{i0} \partial_{0} g_{00} + g^{ij} \partial_{0} g_{j0} -\frac{1}{2} g^{ij} \partial_{j} g_{00} \notag  \\
	= &  \tau \mathcal{S}^i_{00}(\tau, \textbf{u}^{\alpha\beta}_\sigma, \textbf{u}^{\mu\nu}, \textbf{u})  - g_{jk} g^{ij} \eta_{00} (\textbf{u}^{0k}_0+3\textbf{u}^{0k} )+ \frac{1}{2} g^{ij}  (\eta_{00})^2  \textbf{u}^{00}_j .
\end{align*}
Then by \eqref{e:gij}--\eqref{e:ui}, further expansions give
\begin{align}
L_{i}^{\mu}\Gamma_{\mu\nu}^{i}u^{\nu}=&-\frac{u_i}{u_0}\Gamma_{00}^{i}u^{0}-\frac{u_i}{u_0}\Gamma_{0j}^{i}u^{j}
+\delta_{i}^{j}\Gamma_{j0}^{i}u^{0}+\delta_{i}^{j}\Gamma_{j k}^{i}u^{k}\nonumber\\
=&\frac{u^{0}}{2}g^{ik}(\partial_{i}g_{k0}+\partial_{\tau}g_{ki}-\partial_{k}g_{i0})+ \beta(\tau) \texttt{S}(\tau, \textbf{U},\textbf{V}).\label{e:LGU}
\end{align}
Noting
\begin{equation*}
	\partial_\lambda g_{\mu\nu}=-g_{\mu\alpha} g_{\beta\nu} \partial_\lambda g^{\alpha\beta},
\end{equation*}
we derive that
\begin{align}
u^{\mu}\Gamma_{\mu\nu}^{k}u^{\nu} =&u^{0}\Gamma_{00}^{k}u^{0}+2u^{i}\Gamma_{i0}^{k}u^{0}+u^{j}\Gamma^{k}_{ji}u^{i}\nonumber\\
 =& g^{kj}g_{jl}(\textbf{u}_{0}^{0l}+3\textbf{u}^{0l})-\frac{g^{kj}}{2}\eta_{00}\textbf{u}^{00}_{j} +\beta(\tau)\texttt{S}^k(\tau, \textbf{U},\textbf{V}). \label{e:ugu}
\end{align}
By the mean value theorem, express $s(\alpha)-s(\bar{\alpha})$ and $\lambda(\alpha)- \lambda (\bar{\alpha})$ in terms of $\delta \zeta$,
\begin{align}
s(\alpha)-s(\bar{\alpha}) =&s'(\alpha_{K_6})(\alpha-\bar{\alpha})=\beta(\tau)s'(\alpha_{K_6})\delta\zeta  \label{e:s-bs}
\intertext{and}
\lambda(\alpha)- \lambda (\bar{\alpha}) =&\lambda'(\alpha_{K_5}) (\alpha-\bar{\alpha})=\beta(\tau)\lambda'(\alpha_{K_5})\delta\zeta,  \label{e:l-bl}
\end{align}
where $\alpha_{K_\ell}$ ($\ell=5,6$) are intermediate points defined in \S \ref{s:intpt}  
for some constants $K_5, K_6\in (0,1)$. Since  $\alpha-\bar{\alpha}:=\beta(\tau)\delta\zeta$, then
\begin{equation}\label{e:alexp}
\alpha_{K_\ell}:=\bar{\alpha}+K_\ell(\alpha-\bar{\alpha})=\bar{\alpha}+K_\ell\beta(\tau) \delta\zeta =\texttt{S}(\tau, \delta\zeta).
\end{equation}

With the help of $q = s /\lambda $, we have
\begin{align*}
1- 3 s^2 \frac{\bar{q}}{q} =&1- 3(\lambda-\bar{\lambda}) s \bar{q}  -3(s-\bar{s})\bar{s}-3\bar{s}^2\nonumber\\
 =&1-3\bar{s}^2-3\beta(\tau) \lambda'(\alpha_{K_5}) s\bar{q}\delta\zeta-3\beta(\tau)s'(\alpha_{K_6}) \bar{s}\delta\zeta,
\intertext{and}
q-\bar{q}=&\beta q'(\bar{\alpha})\delta\zeta +\frac{1}{2} \beta^2 q''\bigl(\alpha_{K_7}\bigr)(\delta\zeta)^2,  
\end{align*}
where $\alpha_{K_7}:=\bar{\alpha}+K_7(\alpha-\bar{\alpha})$ for constant $K_7\in (0,1)$.
With the help of identities
\begin{align*}
g_{ik}g^{ij} = \delta_{k}^{j}-g_{0k}g^{0j}  \AND
g_{ik}g_{jl}g^{ij} = g_{lk}-g_{ik}g_{0l}g^{0i},
\end{align*}
we can rewrite $S$ and $S_{i}$ as follows
\begin{align}
S =&\frac{1}{\tau}\lambda^2 u^{0}\left[3q'(\bar{\alpha})+\frac{3}{2} q''\bigl(\alpha_{K_7}\bigr)\beta \delta\zeta -\chi(\tau)  \right]\delta\zeta  +\frac{1}{\tau}\chi(\tau) \left(\frac{ \lambda sg_{ij}\beta(\tau) v^{j}}{u_{0}}\right)v^{i}
\nonumber \\
&+\frac{\lambda s\beta^{\prime}(\tau)g_{0i}u^{0}}{\beta(\tau)u_{0}}v^{i}   - \lambda \frac{s}{\beta(\tau)} \frac{u^{0}}{2}g^{ik}(\partial_{i}g_{k0}+\partial_{\tau}g_{ki}-\partial_{k}g_{i0})-  \texttt{S}(\tau, \textbf{U})  \label{e:S}
\intertext{and}
S_{i}
 =&\frac{1}{\tau}\left(-\frac{g_{ik}}{u_{0}}\Bigl(1- 3 s^2 \frac{\bar{q}}{q} -\chi(\tau)\beta(\tau)\lambda s \delta\zeta\Bigr)-\chi(\tau)M_{ki}u^{0}\right)v^{k}\nonumber\\
 &-2\left(\frac{3(\lambda s-\bar{\lambda}\bar{s}) \bar{s}}{\bar{\lambda}\beta(\tau)}+\chi(\tau)\lambda s \delta\zeta\right)g_{ij}\textbf{u}^{0j}-\frac{1}{\tau}\frac{\tau g_{ij}}{\beta(\tau)}\left((1+6\bar{s}^2)\textbf{u}^{0j}+\textbf{u}_{0}^{0j}\right)\nonumber\\
 &
+\frac{\eta_{00} \textbf{u}_{i}^{00}}{2\beta(\tau)}+\texttt{S}_{i}(\tau,\textbf{U},\textbf{V}). \label{e:Si}
\end{align}

\subsubsection{Step $4$: non-degenerated symmetric hyperbolic formulations of new ``better'' variables} \label{s:stp4}
As mentioned in Remark \ref{R:chgvar1}, we need to rescale the velocity to $v^i$ by $\beta$ to ensure the system to be symmetric. However, this brings a new singular term $\frac{\partial_{\tau}g^{0l}}{\beta(\tau)}$ in $S_i$.
Although equations \eqref{Euler-final1a}--\eqref{Euler-final1b} seem to be consistent with the non-degenerated singular hyperbolic system given in Appendix \ref{section:3.5}, $S_i$ involves $\tau$-singular terms of $\textbf{u}^{00}_i$, $\textbf{u}^{0j}$ and $\textbf{u}^{0j}_0$ in a ``bad'' way in \eqref{e:Si}, which destroys the structure of the singular term in the system of Appendix \ref{section:3.5}. In order to overcome this difficulty, we introduce a new variable

\begin{equation}\label{e:NewVar}
\textbf{v}^{k}=v^{k}-Ag^{0k}=v^{k}-2\tau A \textbf{u}^{0k},
\end{equation}
where
\begin{equation}\label{A}
A=A(\tau)=-\frac{3s^2(\bar{\alpha}(\tau))}{\sqrt{-\eta^{00}}\beta(\tau)}=-\frac{3\bar{s}^2}{\omega\beta(\tau)}.
\end{equation}
This new variable adjusts the relations between $\textbf{u}^{0j}$ and $\textbf{u}^{0j}_0$ to a ``good'' form via subtracting $2\tau A \textbf{u}^{0k}$ (hence the equations at the end of this section will precisely agree with the one in Appendix \ref{section:3.5}).
Note that by Assumption \ref{a:postvty}, $\bar{s}\lesssim \beta$, then $A\lesssim 1$ is a regular term.

Expressing \eqref{Euler-final1b} in terms of $\textbf{v}^{k}$ yields
\begin{equation}\label{e:el}
M_{ki}u^{\mu}\partial_{\mu}(\textbf{v}^{k})+\lambda s L_{i}^{\mu}\partial_{\mu}\delta\zeta =M_{ki}u^{\mu}\partial_{\mu}(Ag^{0k})+S_{i}.
\end{equation}
Then direct calculations give
\begin{align}
M_{ki}u^{\mu}\partial_{\mu}(Ag^{0k})+S_{i}
 =&\frac{1}{\tau}\left(-\frac{g_{ik}}{u_{0}}\Bigl(1-3 s^2 \frac{\bar{q}}{q}-\chi(\tau)\beta(\tau)\lambda s \delta\zeta\Bigr)-\chi(\tau)g_{ki}u^{0}\right)v^{k}\nonumber\\
 &\hspace{-1.5cm} -g_{ij}\left(\frac{1}{\beta(\tau)}((1+6\bar{s}^2)\textbf{u}^{0j}+\textbf{u}_{0}^{0j})+\sqrt{-\eta^{00}}A(\textbf{u}_{0}^{0j}+3\textbf{u}^{0j})\right)\nonumber\\
 &\hspace{-1.5cm}+M_{ki}u^{0}(\partial_{\tau}A)g^{0k}
 +M_{ki}u^{j}\partial_{j}(Ag^{0k})
+\frac{\eta_{00} \textbf{u}_{i}^{00}}{2\beta(\tau)}+\hat{\texttt{S}}_{i}(\tau,\textbf{U},\tilde{\textbf{V}}).  \label{e:MS1}
\end{align}
Let us focus on the dangerous terms on the right hand side of \eqref{e:MS1}. By \eqref{A} and the equation \eqref{background-solution1}, we have
\begin{align}
(\partial_{\tau}A)g^{0k}=&-2\tau\textbf{u}^{0k}\left(\frac{6\bar{s}}{\omega\beta}s'(\bar{\alpha})\partial_\tau \bar{\alpha} -\frac{3\bar{s}^2\partial_{\tau}\omega}{\omega^2\beta}\right)
+2\chi\textbf{u}^{0k}\left(\frac{3\bar{s}^2 }{\omega\beta}\right) \notag\\
= & -  \textbf{u}^{0k} \frac{36\bar{s}^2}{\omega\beta\bar{\lambda}}s'(\bar{\alpha}) -  \textbf{u}^{0k}\frac{6\Omega\bar{s}^2}{\omega^3\beta}
+2\chi\textbf{u}^{0k}\left(\frac{3\bar{s}^2 }{\omega\beta}\right). \label{e:dtAg}
\end{align}
Assumption \ref{a:postvty} and \eqref{e:bgrhoest3} imply $(\partial_{\tau}A)g^{0k}$ is not singular in $\tau$.
Then, with the help of identities
\begin{gather*}
1+6\bar{s}^2+3\sqrt{-\eta^{00}}\beta(\tau)A=1-3\bar{s}^2, \\
A\partial_{\tau}g^{0k}=A(\textbf{u}^{0k}_{0}+3\textbf{u}^{0k}) \AND \partial_{j}(Ag^{0k})=A\partial_{j}g^{0k}=A\textbf{u}_{j}^{0k},
\end{gather*}
we derive a further expression of \eqref{e:el}--\eqref{e:MS1},
\begin{align}
\lambda s L_{i}^{\mu}\partial_{\mu}\delta\zeta+M_{ki}u^{\mu}\partial_{\mu}(\textbf{v}^{k}) =&\frac{1}{\tau}\left(-\frac{g_{ik}}{u_{0}}\Bigl(1-3 s^2 \frac{\bar{q}}{q}-\chi(\tau)\beta(\tau)\lambda s \delta\zeta\Bigr)-\chi(\tau)g_{ki}u^{0}\right)\textbf{v}^{k}
\nonumber\\
 &-\frac{g_{ij}}{\tau}\left(\frac{\tau}{\beta(\tau)}(1-3\bar{s}^2)(\textbf{u}^{0j}_{0}+\textbf{u}^{0j})\right)
+\frac{1}{\tau}\frac{\tau\eta_{00} \textbf{u}_{i}^{00}}{2\beta(\tau)}\nonumber\\
 &+\hat{\texttt{S}}_{i}(\tau,\textbf{U},\tilde{\textbf{V}}),   \label{e:Eufin1}
\end{align}
where $\hat{\texttt{S}}_{i}$ satisfies $\hat{\texttt{S}}_i(\tau,0,0)=0$.

Next, let us express \eqref{Euler-final1a} in terms of $\textbf{v}^{k}$. Direct calculation gives
\begin{equation}\label{e:Sieq1}
\lambda^2u^{\mu}\partial_{\mu}\delta\zeta+\lambda sL_{i}^{\mu}\partial_{\mu}\textbf{v}^{i} = S+\lambda sL_{i}^{\mu}\partial_{\mu}(Ag^{0i}).
\end{equation}
Note that
\begin{align*}
	\partial_\mu (A g^{0i})= \delta^0_\mu (\partial_\tau A)g^{0i}+A\delta^0_\mu(\textbf{u}^{0i}_0+3\textbf{u}^{0i}) +A\delta^j_\mu \textbf{u}^{0i}_j
\end{align*}
is regular in $\tau$ due to \eqref{e:dtAg},
 and using \eqref{e:s-bs} and Assumption \ref{a:postvty}.$(1)$ that $\bar{s} \lesssim \beta(\tau)$,
\begin{align*}
	\frac{s}{\beta(\tau)}=\frac{s-\bar{s}}{\beta(\tau)}+\frac{\bar{s}}{\beta(\tau)}= s'(\alpha_{K_6})\delta\zeta  +\frac{\bar{s}}{\beta(\tau)}
\end{align*}
is also regular in $\tau$, which can be absorbed into $\hat{\texttt{F}}(\tau,\textbf{U},\tilde{\textbf{V}})$ in the following equation. Thus, eventually, the equation \eqref{e:Sieq1} becomes
\begin{align}\label{e:Eufin2}
\lambda^2u^{\mu}\partial_{\mu}\delta\zeta+\lambda sL_{i}^{\mu}\partial_{\mu}\textbf{v}^{i}
 =&\frac{1}{\tau}\lambda^2 u^{0}\left[3q'(\bar{\alpha})+\frac{3}{2} q''\bigl(\alpha_{K_7}\bigr)\beta \delta\zeta -\chi(\tau) \right]\delta\zeta  \nonumber\\
 &
+\frac{\chi(\tau)}{\tau}\left(\frac{\beta(\tau)\lambda sg_{ij}\textbf{v}^{j}}{u_{0}}\right)\textbf{v}^{i}
+\hat{\texttt{F}}(\tau,\textbf{U},\tilde{\textbf{V}}),
\end{align}
where $\hat{\texttt{F}}(\tau,0,0)=0$.

Combining \eqref{e:Eufin1} and \eqref{e:Eufin2} together, we bring them to a matrix form
\begin{equation}\label{symmetric-Euler}
N^{\mu}\partial_{\mu}\tilde{\textbf{V}}=\frac{1}{\tau}\textbf{N}\textbf{P}^\dagger\tilde{\textbf{V}}+\frac{1}{\tau}(\textbf{E}_0 \delta_\mu^0+\textbf{E}_{q}\delta_\mu^q) \textbf{U}^{\mu}
+F(\tau,\tilde{\textbf{V}},\textbf{U}),
\end{equation}
where
\begin{align}\label{e::Vtil}
	\tilde{\textbf{V}}=(\delta\zeta,\textbf{v}^{p})^{T}, \qquad  \textbf{U}^{\mu}=(\textbf{u}^{0\mu}_{0},\textbf{u}^{0\mu}_{j},\textbf{u}^{0\mu})^{T}
\end{align}
and
\begin{align*}
N^{\mu}=\left(
\begin{array}{cc}
\lambda^2 u^{\mu}& \lambda s L_{p}^{\mu}\\
\lambda s L_{r}^{\mu}& M_{rp}u^{\mu}
\end{array}\right),\quad
\textbf{E}_0 =\frac{\tau\eta_{00} }{2\beta(\tau)}\left(
\begin{array}{ccc}
0&0&0\\
0&\delta^j_r&0
\end{array}\right),
\end{align*}
\begin{align*}
\textbf{N}=\left(
\begin{array}{cc}
\lambda^2 u^{0}\left[3q'(\bar{\alpha})+\frac{3}{2} q''\bigl(\alpha_{K_7}\bigr)\beta \delta\zeta -\chi(\tau) \right] & \frac{1}{u_{0}}\chi(\tau)\beta (\tau)\lambda s g_{ij}\textbf{v}^{j}\\
\frac{1}{u_{0}}\chi(\tau)\beta (\tau)\lambda s g_{rj}\textbf{v}^{j} & -\frac{g_{ir}}{u_{0}}\bigl(1-3 s^2 \frac{\bar{q}}{q} \bigr)-\chi(\tau)g_{ri}u^{0},
\end{array}\right),
\end{align*}
\begin{align*}
\textbf{E}_{q}=-\frac{\tau}{\beta(\tau)}(1-3\bar{s}^2) \left(
\begin{array}{ccc}
0&0&0\\
g_{rq}&0&g_{rq}
\end{array}\right) ,
\quad
\textbf{P}^\dagger=\left(
\begin{array}{cc}
1&0\\
0&\delta_{p}^i
\end{array}
\right).
\end{align*}
and $F=(\hat{\texttt{F}}, \hat{\texttt{S}}_{i})^T$.

\subsection{Formulation of conformal Euler equations for \textit{Fluids} $(II)$}\label{section:3.3a}
When $\beta\equiv $ constant, by \eqref{e:rhopro} in Assumption \ref{a:tLip}, the suitable variable of density is $\delta\zeta=\alpha-\bar{\alpha}$ (see Remark \ref{R:chgvar1}). Therefore, once obtaining \eqref{Euler-uppera}--\eqref{Euler-upperb}, we do not need to proceed Step $3$--$4$, that is, \S \ref{s:stp3}--\S \ref{s:stp4}, which are performed due to the bad extra decay rate ($\beta(\tau)\lesssim \sqrt{\tau}$) of density and velocity in that case. However, another procedure has been carried out if $q$ is conserved with respect to the perturbations, that is
\begin{equation}\label{e:slcons}
	q= \bar{q}.
\end{equation}
We have to change $u^j$ to $u_q$, otherwise, Condition \eqref{c:6} in Appendix \ref{section:3.5} can not be satisfied due to the degeneracy of $\textbf{BP}$ and $\textbf{P}u$ is the velocity $u^j$, by letting $\textbf{P}u=0$ in $\textbf{P}^\perp B(t,u) \textbf{P}$ and noting
\begin{align}
	u_q=g_{q0} u^0+g_{qi} u^i,  \label{e:u-q}
\end{align}
then $\textbf{P}^\perp B(t,\textbf{P}^\perp u) \textbf{P}=\textbf{P} B(t,\textbf{P}^\perp u) \textbf{P}^\perp \neq 0$. For more details of this case, see \cite{Oliynyk2016a,LeFloch2015a,Liu2017,Liu2018}. In other words, a good expression of the Euler equations verifying Condition \eqref{c:6} in Appendix \ref{section:3.5} relies on the variable $u_q$ instead of $u^i$. In this section, we rewrite the Euler equations in terms of variables $(\delta\zeta, u_q)$ first.

We express $u^k$ in terms of $(g^{\mu\nu}, u_q)$ by performing changes of variables from $u^i$ to $u_q$, which are related
via the map $u^i = u^i(u_q, g^{\mu\nu})$ given by
\begin{equation}\label{velocity}
u^{k}=g^{ki}u_{i}+g^{k0}u_{0}=g^{ki}u_{i}+g^{k0}\frac{-g^{0i}u_{i}-\sqrt{(g^{0i}u_{i})^2-g^{00}(1+g^{ij}u_{i}u_{j})}}{g^{00}}.
\end{equation}
In above derivation, we have used the normalization \eqref{e:nol2}, which is $g^{\alpha\beta}u_\alpha u_\beta=-1$,
to obtain
\begin{equation*}\label{e:ud0}
u_{0}=\frac{-g^{0i}u_{i}-\sqrt{(g^{0i}u_{i})^2-g^{00}(1+g^{ij}u_{i}u_{j})}}{g^{00}}.
\end{equation*}
Denote $J^{ij}$ the Jacobian of $u^i = u^i(u_j, g^{\mu\nu})$.
Differentiating \eqref{velocity} with respect to $u_j$, we calculate
\begin{align}\label{Jacobi}
J^{ij} =&\frac{\partial u^{i}}{\partial u_{j}}=g^{ij}+\frac{g^{0i}}{g^{00}}\left(-g^{0j}-
\frac{2 g^{0j} g^{0k}u_{k}-g^{00}g^{kj}u_{k}}{\sqrt{(g^{0k}u_{k})^2-g^{00}(1+g^{kj}u_{k}u_{j})}}\right)\nonumber\\
 =&\delta^{ij}+\textbf{u}^{ij}+\frac{\eta^{ij}}{\omega^2}(2\tau\textbf{u}^{00}-\textbf{u})+\mathcal{S}^{ij}(\tau,\textbf{u},\textbf{u}^{\mu\nu}) +\frac{2\tau\textbf{u}^{0i}}{\eta^{00}+2\tau \textbf{u}^{00}}\Big(-2\tau\textbf{u}^{0j}\nonumber\\
 &-\frac{4\tau^2\textbf{u}^{0j}\textbf{u}^{0k}u_{k}-(\eta^{00}+2\tau\textbf{u}^{00})g^{kj}u_{k}}{2\sqrt{(2\tau\textbf{u}^{0k}u_{k})^2-(\eta^{00}+2\tau\textbf{u}^{00})
		(1+g^{kl}u_{k}u_{l})}}\Big)\nonumber\\
 =&\delta^{ij}+\mathcal{J}^{ij}(\tau,\textbf{u}^{0\nu},u_{k}).
\end{align}
Where $\mathcal{J}^{ij}(\tau,0,0)=0$.
We differentiate $u^i$ with respect to $x^\mu$, then a simple chain rule gives
\begin{equation}\label{variable-transform}
\partial_{\mu}u^{i}=J^{ij}\partial_{\mu}u_{j}+\frac{\partial u^{i}}{\partial g^{\alpha\beta}}\partial_{\mu}g^{\alpha\beta}.
\end{equation}
Inserting \eqref{variable-transform} into \eqref{Euler-uppera}--\eqref{Euler-upperb}, and  multiplying both sides of \eqref{Euler-upperb} by $J^{ij}$, we can rewrite the Euler equations as
\begin{align}
\lambda^2u^{\mu}\partial_{\mu}(\alpha-\bar{\alpha})+\lambda sL_{i}^{\mu}J^{iq}\partial_{\mu}u_{q}=&
\frac{3}{\tau}\lambda^2 u^0\bigl( q -\bar{q} \bigr)   \notag  \\
&  -\lambda sL_{i}^{\mu}\left(\frac{\partial u^{i}}{\partial (g^{\alpha\beta})}\partial_{\mu}g^{\alpha\beta}+\Gamma_{\mu\nu}^{i}u^{\nu}\right)
,\label{final2} \\
M_{ki}u^{\mu}J^{kj}J^{iq}\partial_{\mu}u_{q}+\lambda sJ^{ij}L_{i}^{\mu}\partial_{\mu}(\alpha-\bar{\alpha})=&J^{jq}\Bigl[
-\frac{1}{\tau} \frac{1}{u_{0}} \bigl(1 -  3 s^2 \frac{\bar{q}}{q} \bigr)  u_{q} \notag\\
&   -M_{ki}u^{\mu}
\Bigl(\frac{\partial u^{k}}{\partial (g^{\alpha\beta})}\partial_{\mu}g^{\alpha\beta}+\Gamma_{\mu\nu}^{k}u^{\nu}\Bigr)
\Bigr]. \label{final2b}
\end{align}

Multiplying $\frac{1}{\lambda^2 u^0}$ on the both sides of above Euler equations  \eqref{final2}--\eqref{final2b} (this step is necessary in order to satisfy Condition \eqref{c:7} of the singular hyperbolic system \eqref{e:model1}, see \eqref{e:PDDP} for the examination of this condition) and rephrasing them to a more compact matrix form yields that
\begin{equation}\label{lower-case}
\hat{N}^{\mu}\partial_{\mu} \hat{\textbf{V}} = \frac{1}{\tau}\hat{\textbf{N}}\hat{\textbf{P}}^\dagger \hat{\textbf{V}}+\hat{H}(\tau,\textbf{U},\hat{\textbf{V}}),
\end{equation}
where $\hat{\textbf{V}}=(\delta\zeta, u_{q})^T$, $\hat{N}^{\mu}$ and $\hat{H}$ are given by
\begin{align*}
\hat{N}^{\mu}=\left(\begin{array}{cc}
 \frac{ u^{\mu}}{  u^0}& q J^{iq}L_{i}^{\mu}\frac{1}{  u^0}\\
q J^{ij}L_{i}^{\mu} \frac{1}{  u^0} & \frac{q^2}{s^2} M_{ki}J^{kj}J^{iq}u^{\mu} \frac{1}{  u^0}
\end{array}
\right),
\end{align*}
and
\begin{align*}
\hat{H}=\frac{1}{\lambda^2 u^0}\left(\begin{array}{c}
-\lambda sL_{i}^{\mu}(\frac{\partial u^{i}}{\partial (g^{\alpha\beta})}\partial_{\mu}g^{\alpha\beta}+\Gamma_{\mu\nu}^{i}u^{\nu})\\
-J^{ij}M_{ki}u^{\mu}
(\frac{\partial u^{k}}{\partial (g^{\alpha\beta})}\partial_{\mu}g^{\alpha\beta}+\Gamma_{\mu\nu}^{k}u^{\nu})
\end{array}
\right).
\end{align*}

In order to use Theorem \ref{pro:3.16}, we list two cases to choose $\hat{\textbf{N}}$ and $\hat{\textbf{P}}^\dagger$:
\begin{enumerate}
\item \label{case1} If $q= \bar{q} $ and $1-3\bar{s}^2\geq \hat{\delta}$ hold, then set
\begin{align}
\hat{\textbf{N}}= \left(\begin{array}{cc}
1 & 0\\
0&-\frac{1}{\lambda^2 u^0 u_{0}}(1-3 s^2 )J^{ij}
\end{array}
\right)
\AND
\hat{\textbf{P}}^\dagger=\left(\begin{array}{cc}
0&0\\
0&\delta^{q}_{i}
\end{array}
\right). \label{e:case1}
\end{align}

\item \label{case1b} If $q= \bar{q}$ and $1-3s^2\equiv 0$ hold, then set
\begin{align}
\hat{\textbf{N}}= \left(\begin{array}{cc}
1 & 0\\
0& \frac{1}{\bar{\lambda}^2}\delta^{iq}
\end{array}
\right)
\AND
\hat{\textbf{P}}^\dagger=\left(\begin{array}{cc}
0&0\\
0&0
\end{array}
\right). \label{e:case1b}
\end{align}

\end{enumerate}

Now we have transformed the Einstein--Euler system to the target singular symmetric hyperbolic form. In the next section, we focus on examining the Conditions in Appendix \ref{section:3.5} by above formulations.



\section{Proof of the main Theorems}\label{s:theorem:1.4}
In this section, we prove the main theorem based on Theorem \ref{pro:3.16}.

\subsection{Proof of Theorem \ref{theorem:1.4}}

\subsubsection{Local existence and continuation of reduced conformal Einstein--Euler equations}

The local existence and continuation of reduced conformal Einstein--Euler equations can be derived using standard local existence and continuation results for
symmetric hyperbolic systems (see, e.g. Theorems 2.1 and 2.2 of \cite{Majda2012}), as long as the reduced conformal Einstein--Euler equations are well-defined when they are  written as a symmetric hyperbolic
system, that is, the conformal metric $g^{\mu\nu}$ remains non-degenerate and the conformal fluid four-velocity remains future directed,
\begin{equation*} \label{e:welldef}
\det(g^{\mu\nu})< 0 \AND u^0 < 0.
\end{equation*}
and $\rho$ remains strictly positive, and the new variables $(\textbf{U},\textbf{V})$ or $(\textbf{U},\tilde{\textbf{V}})$ are equivalent to the original ones $(g^{\mu\nu}, \partial_\sigma g^{\mu\nu}, \rho, u^\mu)$. We omit the details for which we refer readers to \cite[\S $3$]{Liu2017} and \cite[\S $4$]{Liu2018}, but give the proposition to state this result.
\begin{proposition} \label{rcEEexist}
	Suppose $k\in \mathbb{Z}_{\geq 3}$, $\Lambda>0$,    $(g^{\mu\nu}_0)$ $\in$ $H^{k+1}(\mathbb{T}^3,\mathbb{S}_4)$,
	and $(g^{\mu\nu}_1)$ $\in$ $H^{k}(\mathbb{T}^3,\mathbb{S}_4)$,
	$\nu^\alpha$ $\in$  $H^{k}(\mathbb{T}^3,\mathbb{R}^4)$ and $\rho_0$ $\in$ $H^{k}(\mathbb{T}^3)$,
	where $\nu^\alpha$ is normalized by $g_{0\alpha\beta}\nu^\alpha\nu^\beta=-1$, and $\det(g^{\mu\nu}_0) < 0$
	and $\rho_0>0$ on
	$\mathbb{T}^3 $. Then there exists a $T_1 \in (0,1]$ and
	a unique classical solution
	\begin{align*}
	(g^{\mu\nu},u^{\mu}, \rho )\in \bigcap_{\ell=0}^2 C^{\ell}((T_1,1],H^{k+1-\ell}(\mathbb{T}^3 ))
	\times \bigcap_{\ell=0}^1 &  C^{\ell}((T_1,1],H^{k-\ell}(\mathbb{T}^3))  \notag \\
	& \times \bigcap_{\ell=0}^1  C^{\ell}((T_1,1],H^{k-\ell}(\mathbb{T}^3)),
	\end{align*}
	of the conformal Einstein--Euler equations, given by \eqref{e:cEin1} and \eqref{e:cEu1}, on the spacetime region
	$(T_1,1]\times \mathbb{T}^3 $ that satisfies
	\begin{equation*}
	(g^{\mu\nu}, \partial_\tau g ^{\mu\nu},\rho, u^\alpha )|_{\tau=1} =( g^{\mu\nu}_0,g^{\mu\nu}_1, \rho_0,\nu^\alpha ).
	\end{equation*}
	Moreover,
	\begin{itemize}
		\item[(i)] the vector $(\textbf{U},\textbf{V})$, see \eqref{e:UV}, is well-defined, lies in the space
		\begin{equation*}
		(\textbf{U},\textbf{V}) \in  \bigcap_{\ell=0}^1  C^{\ell}((T_1,1],H^{k-\ell}(\mathbb{T}^3,\mathbb{V})), \label{e:locU}
		\end{equation*}
		where
		\begin{equation*}\label{e:bbV}
		\mathbb{V} = \mathbb{R}^4\times\mathbb{R}^{12}\times \mathbb{R}^4\times\mathbb{S}_3\times  (\mathbb{S}_3)^3\times\mathbb{S}_3\times \mathbb{R}
		\times\mathbb{R}^3\times\mathbb{R} \times \mathbb{R}\times\mathbb{R}^3,
		\end{equation*}
		and solves \eqref{4.26}--\eqref{4.28} and the Euler equations \eqref{symmetric-Euler} basing on a transformation \eqref{e:NewVar}, or \eqref{lower-case} via transformation \eqref{e:u-q} on the spacetime region $(T_1,1]\times \mathbb{T}^3$, and
		\item[(ii)] there exists a constant $\sigma > 0$, independent of   $T_1\in (0,1)$, such that if $(\textbf{U},\textbf{V})$ satisfies
		\begin{equation*}
		\|(\textbf{U},\textbf{V})\|_{L^\infty((T_1,1],H^s(\mathbb{T}^3))} < \sigma,
		\end{equation*}
		then the solution $(g^{\mu\nu},u^{\mu}, \rho)$ can be uniquely continued as a classical solution
		with the same regularity
		to the larger spacetime region $(T^*_1,1]\times \mathbb{T}^3$ for some $T^*_1 \in (0,T_1)$.
	\end{itemize}
\end{proposition}

\subsubsection{Proof for \textit{Fluids} $(I)$} \label{s:Mak1}
Let us first gather the Einstein equations \eqref{4.26}--\eqref{4.28} and the Euler equations \eqref{symmetric-Euler} together, recalling \eqref{e:UV} and \eqref{e::Vtil} that
\begin{align}\label{e:UtilV}
	\textbf{U}: =(\textbf{u}^{0\mu}_{0},\textbf{u}^{0\mu}_{j},\textbf{u}^{0\mu},\textbf{u}^{lm}_{0},\textbf{u}^{lm}_{j},\textbf{u}^{lm}, \textbf{u} _{0},\textbf{u} _{j},\textbf{u} )^{T} \AND \tilde{\textbf{V}}=(\delta\zeta,\textbf{v}^{p})^{T},
\end{align}
to get the complete non-degenerated singular symmetric hyperbolic system
\begin{equation}\label{symmetric-Euler1Ma1}
B^{\mu}\partial_{\mu}\p{\textbf{U}\\ \tilde{\textbf{V}}}=\frac{1}{\tau}\textbf{BP}\p{\textbf{U}\\ \tilde{\textbf{V}}}+H,
\end{equation}
where
\begin{gather*}
	\textbf{B}=\left(
	\begin{array}{cccc}
	\textbf{A}&0&0&0\\
	0&-2g^{00}\mathbb I&0&0\\
	0&0&-2g^{00}\mathbb I&0\\
	-(\textbf{E}_0 \delta^0_\mu +\textbf{E}_{q} \delta^q_\mu) &0&0&-\textbf{N}
	\end{array}
	\right), \\
	B^{\mu}=\left(
	\begin{array}{cccc}
	A^{\mu}&0&0&0\\
	0&A^{\mu}&0&0\\
	0&0&A^{\mu}&0\\
	0&0&0&-N^{\mu}
	\end{array}
	\right),
	\quad
	\textbf{P}=\left(
	\begin{array}{cccc}
	\textbf{P}^\star &0&0&0\\
	0&\Pi&0&0\\
	0&0&\Pi&0\\
	0&0&0&\textbf{P}^\dagger
	\end{array}
	\right)
\end{gather*}
and $H=(F_1,F_2,F_3, -F)^T$.

Then by simply reversing the time $\tau\rightarrow -\tau$
on above equations or the model equation in Appendix \ref{section:3.5}, we can apply Theorem \ref{pro:3.16} to above equations \eqref{symmetric-Euler1Ma1} directly.

We need to verify all the conditions in Appendix \ref{section:3.5}. \textbf{Conditions \eqref{c:1}--\eqref{c:3}} are evident and with the help of \eqref{e:alexp}, it is direct to verify $H,\,B^{\mu},\,\textbf{B}\in C^{0}([T_{0},0],C^{\infty}(\mathbb{V} ))$.

The key to \textbf{Condition \eqref{c:4}} are $B^{0}\in C^{1}([T_{0},0],C^{\infty}( \mathbb{V}  ))$ and $[\textbf{P}, \textbf{B}]=\textbf{PB}-\textbf{BP}=0$. We now verify them respectively.

\textbf{$(1)$ Verification of $B^{0}=\diag\{A^0,A^0,A^0,-N^0\}\in C^{1}([0,1],C^{\infty}( \mathbb{V}  ))$:} $\partial_\tau A^0$ has been investigated in \cite{Oliynyk2016a,Liu2017,Liu2018}, we need to verify $N^{0}\in C^{1}([0,1],C^{\infty}( \mathbb{V}  ))$.
First we prove that $\tau A \in C^1([0,1])$ which can be verified by noting that $ A \in C^0([0,1])$ (recall that $A$ is defined by \eqref{A}) and
by \eqref{A} and the equation \eqref{background-solution1},
\begin{align*}
\partial_{\tau} ( \tau A )
= & -\frac{3\bar{s}^2}{\omega\beta } - \frac{18\bar{s}^2}{\omega\beta\bar{\lambda}}s'(\bar{\alpha}) - \frac{3\Omega\bar{s}^2}{\omega^3\beta}
+ \chi \left(\frac{3\bar{s}^2 }{\omega\beta}\right) \lesssim 1 \quad (\text{by Assumption \ref{a:postvty}, } \bar{s}\lesssim \beta),
\end{align*}
for $\tau\in [0,1]$. Therefore, \eqref{e:NewVar} can be expressed as $v^j=\textbf{v}^j+\tau A(\tau) \mathcal{S}^j(\tau,\textbf{U})$ where, as our convention of notations, $\mathcal{S}^j(\tau,\textbf{U})\in C^1([0,1],C^\infty(\mathbb{V}))$ and $\mathcal{S}^j(\tau,0) =0$.
Applying \eqref{e:u0up}--\eqref{e:ui}, \eqref{background-solution1}, \eqref{e:resvar}, \eqref{e:Mexp} and \eqref{e:expL}, we calculate $D_\tau N^0(\tau, \textbf{U}, \tilde{\textbf{V}})$ (recall the convention of notation in \S \ref{s:deriv}, the derivative operator $D_\tau$ is the partial derivative with respect to the first variable $\tau$) and try to see if it is continuous in $[0,1]$.
Directly expanding all the derivatives in $D_\tau N^0$, with the help of that $\tau A \in C^1([0,1])$, we find the key quantities in $D_\tau N^0$ (the other terms are easy to bound in $[0,1]$) are
\begin{align}
D_{\tau}u^{0}(\tau,\textbf{U},\tilde{\textbf{V}}) =&\frac{1}{\omega}\frac{\Omega}{\tau}+\texttt{S}(\tau,\textbf{U},\tilde{\textbf{V}}),  \label{e:du0}\\
D_{\tau}u_{0}(\tau,\textbf{U},\tilde{\textbf{V}}) =& \frac{1}{  \omega^3}  \frac{\Omega}{\tau}  +\texttt{T}(\tau,\textbf{U},\tilde{\textbf{V}}), \\
D_{\tau}M_{ki}(\tau,\textbf{U},\tilde{\textbf{V}}) =&
\texttt{F}_{ki}(\tau,\textbf{U},\tilde{\textbf{V}})
\end{align}
where, by \eqref{e:bgrhoest4}--\eqref{e:bgrhoest3}, $|\Omega/\tau| \lesssim 1$ for $\tau\in[0,1]$ and the typewriter fond remainders $\texttt{S}$, $\texttt{T}$, $\texttt{F}_{ki}\in C^0([0,1],C^\infty(\mathbb{V}))$  agree with the conventions in \S \ref{remainder} and $\texttt{S}(\tau,0,0)=\texttt{T}(\tau,0,0)=\texttt{F}_{ki}(\tau,0,0)=0$. Also note that
\begin{align*}
	s D_{\tau}u_{i}(\tau,\textbf{U},\tilde{\textbf{V}}) =&(\bar{s}+ s'(\alpha_{K_5}) \beta \delta\zeta)\partial_\tau \beta (\tau)g_{ij}\textbf{v}^{j}+\texttt{L}_i(\tau,\textbf{U},\tilde{\textbf{V}}),
\end{align*}
and another crucial quantities are (view $\lambda (\alpha)$ and $s(\alpha)$ are functional of $(\tau,\textbf{U},\tilde{\textbf{V}})$)
\begin{align}
	D_{\tau}\lambda (\alpha)=&D_\tau \lambda\bigl(\bar{\alpha}(\tau)+\beta(\tau) \delta\zeta\bigr)=
	\lambda^\prime(\alpha) \Bigl[\delta\zeta\partial_\tau\beta(\tau)+\frac{3}{\tau}\bar{q}\Bigr] \notag \\
	=&	(\lambda^\prime(\bar{\alpha})+\lambda''(\alpha_{K_8})\beta\delta\zeta) \Bigl[\delta\zeta\partial_\tau\beta(\tau)+\frac{3}{\tau}\bar{q}\Bigr]  \notag  \\
	= & 	 \lambda^\prime(\bar{\alpha})\Bigl[\delta\zeta\partial_\tau\beta(\tau)+\frac{3\bar{s}}{\tau\bar{\lambda}} \Bigr]+\lambda''(\alpha_{K_8})\delta\zeta  \Bigl[\delta\zeta\beta\partial_\tau\beta(\tau)+\frac{3}{\tau}\beta(\tau)\bar{q}\Bigr] ,
	\intertext{and}
	\beta(\tau)D_{\tau}s(\alpha) =&\beta(\tau) D_\tau s\bigl(\bar{\alpha}(\tau)+\beta(\tau) \delta\zeta\bigr)=
	\beta(\tau) s^\prime( \alpha ) \Bigl[\delta\zeta\partial_\tau\beta(\tau)+\frac{3}{\tau}\bar{q}\Bigr]\notag \\=&\beta(\tau)(s^\prime(\bar{\alpha})+s''(\alpha_{K_9})\beta\delta\zeta) \Bigl[\delta\zeta\partial_\tau\beta(\tau)+\frac{3}{\tau}\bar{q}\Bigr] \notag \\
	=& (s^\prime(\bar{\alpha})+s''(\alpha_{K_9}) \beta\delta\zeta) \Bigl[\delta\zeta \beta(\tau)\partial_\tau\beta(\tau)+\frac{3}{\tau}\beta(\tau)\bar{q} \Bigr], \label{e:bds}
\end{align}
where 
$K_8$, $K_9\in(0,1)$. Then Assumption \ref{a:postvty}.\eqref{e:B0bd} and inequalities \eqref{e:dtbt2} guarantee $N^{0}\in C^{1}([0,1],C^{\infty}( \mathbb{V}  ))$, and then we conclude this condition.  

\textbf{$(2)$ Verification of $[\textbf{P}, \textbf{B}]=\textbf{PB}-\textbf{BP}=0$: }To verify this condition, we only need to examine the following three relations,
\begin{align}
	\textbf{P}^\star\textbf{A}=&\textbf{A}\textbf{P}^\star, \label{e:comm1}\\
	\textbf{P}^\dagger\textbf{N}=&\textbf{N}\textbf{P}^\dagger  \label{e:comm2}
	\intertext{and}
	\textbf{P}^\dagger(\textbf{E}_0\delta^0_\mu+\textbf{E}_q\delta^q_\mu) =& (\textbf{E}_0\delta^0_\mu+\textbf{E}_q\delta^q_\mu)\textbf{P}^\star. \label{e:comm3}
\end{align}
It is evident that \eqref{e:comm1} and \eqref{e:comm2} hold by direct calculations. Let us focus on \eqref{e:comm3} and calculate
\begin{align*}
	&(\textbf{E}_0\delta^0_\mu+\textbf{E}_q\delta^q_\mu)\textbf{P}^\star \notag \\
	&=\Biggl(\frac{\tau\eta_{00} }{2\beta(\tau)}\left(
	\begin{array}{ccc}
	0&0&0\\
	0&\delta^l_r&0
	\end{array}\right) \delta^0_\mu -\frac{\tau}{\beta(\tau)}(1-3\bar{s}^2) \left(
	\begin{array}{ccc}
	0&0&0\\
	g_{rq}&0&g_{rq}
	\end{array}\right)\delta^q_\mu\Biggr) \left(
	\begin{array}{ccc}
	\frac{1}{2}& 0& \frac{1}{2}\\
	0& \delta^{j}_{l}& 0\\
	\frac{1}{2}&0& \frac{1}{2}
	\end{array}\right)  \notag \\
	& 
	=  \frac{\tau\eta_{00} }{2\beta(\tau)}\left(
	\begin{array}{ccc}
	0&0&0\\
	0&\delta^j_r&0
	\end{array}\right) \delta^0_\mu -\frac{\tau}{\beta(\tau)}(1-3\bar{s}^2) \left(
	\begin{array}{ccc}
	0&0&0\\
	g_{rq}&0&g_{rq}
	\end{array}\right)\delta^q_\mu  \notag  \\
	&=\left(
	\begin{array}{cc}
	1&0\\
	0&\delta_{r}^i
	\end{array}
	\right)\Biggl(\frac{\tau\eta_{00} }{2\beta(\tau)}\left(
	\begin{array}{ccc}
	0&0&0\\
	0&\delta^j_i&0
	\end{array}\right) \delta^0_\mu -\frac{\tau}{\beta(\tau)}(1-3\bar{s}^2) \left(
	\begin{array}{ccc}
	0&0&0\\
	g_{iq}&0&g_{iq}
	\end{array}\right)\delta^q_\mu\Biggr)  \notag  \\
	&\hspace{9cm} =\textbf{P}^\dagger(\textbf{E}_0\delta^0_\mu+\textbf{E}_q\delta^q_\mu).
\end{align*}
Now we complete the verification of  \textbf{Condition \eqref{c:4}}.

Next, we verify \textbf{Condition \eqref{c:5}}. In order to state this concisely, we address it by using the following simple Lemma,
\begin{lemma}\label{T:Matpos}
	Suppose a real block matrix
	\begin{align*}
		Q=\p{\tilde{A} & 0 \\
		   E & N}
	\end{align*}
	where $\tilde{A}$ is a $m\times m$ symmetric matrix, $N$ a $n\times n$ symmetric matrix and $E$ a $n\times m$ matrix. If there is a constant $\epsilon>0$, such that
	\begin{align*}
		\tilde{A}-\frac{1}{2\epsilon}\mathbb{I} \AND N-\frac{\epsilon}{2}EE^T
	\end{align*}
	are both positive definite,
	then $Q$ is positive definite, that is
	\begin{align*}
		(X^T,Y^T)Q\p{X\\Y} \geq 0
	\end{align*}
	for any $m\times 1$ matrix $X\in \mathbb{R}^m$ and $n\times 1$ matrix $Y\in \mathbb{R}^n$.
\end{lemma}
\begin{proof}
	For any $m\times 1$ matrix $X\in \mathbb{R}^m$ and $n\times 1$ matrix $Y\in \mathbb{R}^n$, note there is an inequality that
	\begin{align*}
		-\frac{\epsilon}{2}Y^TEE^T Y-\frac{1}{2\epsilon}X^T X \leq Y^T E X \leq \frac{\epsilon}{2}Y^TEE^T Y+\frac{1}{2\epsilon}X^T X
	\end{align*}
	for some constant $\epsilon >0$. Then since
	\begin{align*}
		(X^T,Y^T)\p{\tilde{A} & 0 \\
			E & N}\p{X\\Y}=X^T \tilde{A}X+Y^TNY+Y^T E X,
	\end{align*}
	we arrive at
	\begin{align*}
		& X^T \Bigl(\tilde{A}-\frac{1}{2\epsilon}\mathbb{I}\Bigr)X+Y^T\Bigl(N-\frac{\epsilon}{2}EE^T\Bigr)Y \leq (X^T,Y^T)\p{\tilde{A} & 0 \\
			E & N}\p{X\\Y} \notag \\
		&\hspace{6cm}\leq X^T \Bigl(\tilde{A}+\frac{1}{2\epsilon}\mathbb{I}\Bigr)X+Y^T\Bigl(N+\frac{\epsilon}{2}EE^T\Bigr)Y.
	\end{align*}
	We then complete the proof of this Lemma due to the positive definite properties of $\tilde{A}-\frac{1}{2\epsilon}\mathbb{I}$ and $N-\frac{\epsilon}{2}EE^T$ for some $\epsilon>0$.
\end{proof}
\begin{remark}
	Note that if $E=(a_{ij})$, $E^T$ denote $(a^T_{kl})$ where $a^T_{kl}=a_{lk}$, then we can calculate the element $b_{il}$ of $EE^T$ is equal to $\delta^{kj} a_{ik} a_{lj}$.
\end{remark}

Now let us use this Lemma to examine the Condition \eqref{c:5}. Firstly, we introduce a notation
\begin{align*}
	\cir{[M]}:= \ring{M}(\tau) \AND \til{[M]}:=\tilde{M}(\tau, \textbf{U}, \tilde{\textbf{V}})
\end{align*}
respectively, provided there is a decomposition of a functional $M(\tau, \textbf{U}, \tilde{\textbf{V}}) =\ring{M} (\tau)+\tilde{M} (\tau, \textbf{U}, \tilde{\textbf{V}})$
where $\tilde{M}(\tau,0,0)=0$. Then, it is easy to conclude simple but \textit{crucial identities} that
\begin{gather}
	\cir{[M]}=\ring{M}(\tau)=M(\tau, 0,0),  \label{e:cir0}  \\
	\cir[M_1M_2]=\cir[M_1]\cir[M_2]=\ring{M}_1\ring{M}_2
	\intertext{and} \cir[M_1+M_2]=\cir[M_1]+\cir[M_2]=\ring{M}_1+\ring{M}_2
\end{gather}
for $M_\ell=\ring{M}_\ell(\tau)+\tilde{M}_\ell(\tau, \textbf{U}, \tilde{\textbf{V}})$ $(\ell=1,2)$.

In order to verify \eqref{e:Bineq}, we only concern
\begin{align*}
Q=\frac{1}{\kappa}\ring{\textbf{B}}-\ring{B}^0=\cir{\Bigl[\frac{1}{\kappa} \textbf{B} - B^0\Bigr]}=\frac{1}{\kappa} \textbf{B}(\tau,0,0) - B^0(\tau,0,0),
\end{align*}
due to the fact that evidently $\ring{B}^0\geq \frac{1}{\gamma_1}\mathbb{I}$ for some constant $\gamma_1>0$ and $\ring{\textbf{B}}$ is bounded.
We only need to verify $Q$ is positive definite, which implies $\ring{B}^0\leq \frac{1}{\kappa}\ring{\textbf{B}}$.
Then, the corresponding matrices in the Lemma are selected by
\begin{align}
	\tilde{A}=&\cir\p{\frac{1}{\kappa}\textbf{A}-A^{0}&0&0 \\
		0&-\frac{2}{\kappa}g^{00} \mathbb{I}-A^{0}&0 \\
		0&0&-\frac{2}{\kappa}g^{00} \mathbb{I}-A^{0}}, \label{e:A} \\ N=&\cir\p{-(\frac{1}{\kappa}\textbf{N}-N^0)} \label{e:D}
	\intertext{and}
	E=&\cir\p{-\frac{1}{\kappa}(E_0\delta^0_\mu +E_q \delta^q_\mu) & 0 & 0}. \label{e:E}
\end{align}
Then by Lemma \ref{T:Matpos}, we only need to concentrate on examining
\begin{align}
	\tilde{A}-\frac{1}{2\epsilon}\mathbb{I} =& \cir \p{\frac{1}{\kappa}\textbf{A}-A^{0}-\frac{1}{2\epsilon}\mathbb{I}&0&0 \\
		0&-\frac{2}{\kappa}g^{00} \mathbb{I}-A^{0}-\frac{1}{2\epsilon}\mathbb{I}&0 \\
		0&0&-\frac{2}{\kappa}g^{00} \mathbb{I}-A^{0}-\frac{1}{2\epsilon}\mathbb{I}}   \label{e:A-I}
	\intertext{and}
	N-\frac{\epsilon}{2}EE^T= & \cir\p{-(\frac{1}{\kappa}\textbf{N}-N^0)-\frac{\epsilon}{2}\frac{1}{\kappa^2}(E_0\delta^0_\mu +E_q \delta^q_\mu)(E_0\delta^0_\nu +E_l \delta^l_\nu)^T \delta^{\mu\nu}}  \label{e:D-EE}
\end{align}
are both positive definite. Let us, with the help of \eqref{e:cir0}, calculate
\begin{align}
	\cir\Bigl[\frac{1}{\kappa}\textbf{A}-A^{0}-\frac{1}{2\epsilon}\mathbb{I}\Bigr]
	= & \p{  \bigl(\frac{1}{\kappa}-1\bigr) \omega^2 -\frac{1}{2\epsilon} & 0& 0\\
		0& \bigl(\frac{3}{2\kappa}-1-\frac{1}{2\epsilon}\bigr)\delta^{ij}& 0\\
		0&0& \bigl(\frac{1}{\kappa}-1\bigr) \omega^2 -\frac{1}{2\epsilon}},  \label{e:AAI}\\
	\cir\Bigl[-\frac{2}{\kappa}g^{00} \mathbb{I}-A^{0}-\frac{1}{2\epsilon}\mathbb{I} \Bigr]
= & \p{	\bigl(\frac{2}{\kappa}-1\bigr)\omega^2 -\frac{1}{2\epsilon} & 0& 0\\
	0& 	 \bigl(\frac{2}{\kappa}\omega^2 -1-\frac{1}{2\epsilon}\bigr) \delta^{ij}& 0\\
	0&0& 	\bigl(\frac{2}{\kappa}-1\bigr)\omega^2 -\frac{1}{2\epsilon}} \label{e:gAI}
\end{align}
and by noting that $\cir{[u^0]}=-\omega$ and $\cir{[u_0]}=\frac{1}{\omega}$, and since $q''(\alpha)\lesssim 1$, then
\begin{align} \label{e:NN}
	&\cir\Bigl[-\Bigl(\frac{1}{\kappa}\textbf{N}-N^0\Bigr)-\frac{\epsilon}{2}\frac{1}{\kappa^2}(E_0\delta^0_\mu +E_q \delta^q_\mu)(E_0\delta^0_\nu +E_l \delta^l_\nu)^T \delta^{\mu\nu} \Bigr] \notag \\
	= & \p{  \frac{1}{\kappa}\bar{\lambda}^2 \omega \bigl( 3q'(\bar{\alpha})-\chi(\tau)\bigr)  -\bar{\lambda}^2 \omega & 0 \\
	0 & \bigl[\frac{1}{\kappa}   \bigl(1- 3   \bar{s}^2  \bigr)-\frac{1}{\kappa}\chi(\tau)  -1\bigr]\delta_{ri} \omega -\frac{\epsilon}{2}\frac{1}{\kappa^2}\frac{\tau^2}{\beta^2} S(\tau)  \delta_{ri}}
\end{align}
where
\begin{align*}
S(\tau)
=& \Bigl(-\frac{1 }{2\omega^2 }\Bigr)^2 +2  (1-3\bar{s}^2)^2.
\end{align*}
Now we examine all the elements in above \eqref{e:AAI}--\eqref{e:NN} are positive by recalling
$\hat{\delta}$ defined by \eqref{e:hdel} which is $
\hat{\delta}\in \bigl(0,  \bigl(1+\sqrt{\frac{3}{\Lambda}}\bigr)\min \{ \frac{3}{4} ,  \frac{ \Lambda}{3+\Lambda} \} \bigr)$
and 
taking
\begin{equation*}
	\kappa =  \frac{1}{2}\biggl(1+\sqrt{\frac{3}{\Lambda}}\biggr)^{-1} \hat{\delta}   \AND \epsilon=\frac{1}{2},
\end{equation*}
which implies \eqref{e:A-I} and \eqref{e:D-EE} are both positive definite. With the help of \eqref{e:bgrhoest4} (which implies $\omega^2\geq \frac{\Lambda}{3} $),
we arrive at
\begin{align*}
	\Bigl(\frac{1}{\kappa}-1\Bigr) \omega^2 -1    \geq  & \Bigl[\frac{2}{ \hat{\delta} } \Bigl(1+\sqrt{\frac{3}{\Lambda}}\Bigr)-1\Bigr]\frac{\Lambda}{3}-1 >1+\frac{\Lambda}{3}>0   \quad  \Bigl(\text{by } \hat{\delta}<\Bigl(1+\sqrt{\frac{3}{\Lambda}}\Bigr)\frac{ \Lambda}{3+\Lambda}\Bigr), \\
	\frac{3}{2\kappa}-1-1  = & 2\Bigl(1+\sqrt{\frac{3}{\Lambda}}\Bigr)\frac{3}{ 2\hat{\delta} }-2>2 >0   \quad  \Bigl(\text{by } \hat{\delta}<\Bigl(1+\sqrt{\frac{3}{\Lambda}}\Bigr)\frac{3}{4}\Bigr), \\
	\frac{2}{\kappa}\omega^2 -1-1 \geq  & 2\Bigl(1+\sqrt{\frac{3}{\Lambda}}\Bigr)\frac{2\Lambda}{3\hat{\delta}}-2  > 2+\frac{4}{3}\Lambda >0   \quad  \Bigl(\text{by } \hat{\delta}<\Bigl(1+\sqrt{\frac{3}{\Lambda}}\Bigr)\frac{ \Lambda}{3+\Lambda}\Bigr).
\end{align*}
With the help of Assumption \ref{a:postvty}.$(1)$, we obtain
\begin{align*}
	\frac{1}{\kappa} \bigl( 3q'(\bar{\alpha})-\chi(\tau)\bigr)  -1 \geq 2  \Bigl(1+\sqrt{\frac{3}{\Lambda}}\Bigr)\frac{3\hat{\delta}}{\hat{\delta}}-1 >5+6\sqrt{\frac{3}{\Lambda}} >0.
\end{align*}
By using \eqref{e:btest} and noting $C^*$ is defined in Assumption \ref{a:postvty}.$(1)$, we have
\begin{align*}
	\Bigl[\frac{1}{\kappa}   \bigl(1 - 3 \bar{s}^2  - \chi(\tau) \bigr) -1\Bigr] \omega & -\frac{1}{4}\frac{1}{\kappa^2}\frac{\tau^2}{\beta^2}S(\tau) \geq  \Bigl(\frac{\hat{\delta}}{\kappa}     -1\Bigr) \sqrt{\frac{\Lambda}{3}} - \frac{1}{4 \kappa^2}(C^*\hat{\delta})^2 \Bigl(\frac{9}{4\Lambda^2}+2\Bigr)   \notag \\
	& \geq
	\Bigl(2\Bigl(1+\sqrt{\frac{3}{\Lambda}}\Bigr)\frac{ \hat{\delta}}{\hat{\delta}}     -1\Bigr) \sqrt{\frac{\Lambda}{3}} - 1 =1+\sqrt{\frac{\Lambda}{3}}>0 
\end{align*}
for any $\tau\in[0,1]$.
Then we verified \textbf{Condition \eqref{c:5}}.

Let us turn to \textbf{Condition \eqref{c:6}}.
First calculate
\begin{align*}
	\textbf{P}^\perp=\left(
	\begin{array}{cccc}
		(\textbf{P}^\star)^\perp &0&0&0\\
		0&\Pi^\perp&0&0\\
		0&0&\Pi^\perp&0\\
		0&0&0&(\textbf{P}^\dagger)^\perp
	\end{array}
	\right).
\end{align*}
Then this condition means we need to examine
\begin{align*}
	(\textbf{P}^\dagger)^\perp N^0\bigl(\tau,\textbf{P}^\perp(\textbf{U},\tilde{\textbf{V}})^T\bigr) (\textbf{P}^\dagger)=(\textbf{P}^\dagger) N^0\bigl(\tau,\textbf{P}^\perp(\textbf{U},\tilde{\textbf{V}})^T\bigr)  (\textbf{P}^\dagger)^\perp=0.
\end{align*}
This always holds since $(\textbf{P}^\dagger)^\perp\equiv 0$. The other parts are easy to check and the same as the corresponding part in  \cite{Oliynyk2016a}.

The last \textbf{Condition \eqref{c:7}} can be examined by noting $(\textbf{P}^\dagger)^\perp\equiv 0$ and using the derivations in \cite[\S $7.1$]{Liu2017}, which we omit here and readers can consult \cite[\S $7.1$]{Liu2017} for the details.

Having verified that all of the hypotheses of Theorem \ref{pro:3.16} are satisfied, we conclude that there exists a constant $\sigma>0$, such that if
\begin{align*}
	\|g_{0}^{\mu\nu}-\eta^{\mu\nu}(1)\|_{H^{k+1}}+\|g_{1}^{\mu\nu}-\partial_{\tau}\eta^{\mu\nu}(1)\|_{H^{k}}+\|\rho_0 -\bar{\rho}(1)\|_{H^{k}}+
	\|\nu^{i} \|_{H^{k}}<\sigma,
\end{align*}
which, by Assumption \ref{a:Maksym},  \eqref{e:q}--\eqref{e:denstydif}, \eqref{e:NewVar}--\eqref{A} and $\delta\zeta=\beta^{-1}(\mu^{-1})'(\rho_{K_{10}})(\rho-\bar{\rho})$,  implies
\begin{align*}
	\|(\mathbf{U},\tilde{\mathbf{V}})|_{\tau=1}\|_{H^k}\leq C\sigma ,
\end{align*}
then by Theorem \ref{pro:3.16}, there exists a $T_*\in(0,1)$, and  a unique classical solution $(\mathbf{U},\tilde{\mathbf{V}})\in C^1([T_*,1]\times \mathbb{T}^3)$ that satisfies
\begin{equation}\label{e:uvreg}
	(\mathbf{U},\tilde{\mathbf{V}})\in C^{0}([T_{\ast},1],H^{k})\cap C^{1}([T_{\ast},1],H^{k-1}),
\end{equation}
and the energy estimate
\begin{equation}\label{e:apriest}
	\|(\mathbf{U},\tilde{\mathbf{V}}) \|_{H^{k}}  \leq
	C \|(\mathbf{U},\tilde{\mathbf{V}})|_{\tau=1}\|_{H^{k}}  \leq C\sigma
\end{equation}
for all $1 \geq \tau > T_*$, and can be uniquely continued to a larger time interval $(T^*,1]$ for all $T^*\in [0,T_*)$. Furthermore, above bound leads to the solutions $(\mathbf{U},\tilde{\mathbf{V}})$ exist globally on
$\mathfrak{M} = (0, 1] \times  \mathbb{T}^3$ and satisfy the estimates \eqref{e:apriest} with $T_* = 0 $. In particular, this implies, via the definition \eqref{e:UtilV} of $(\mathbf{U},\tilde{\mathbf{V}})$, that
\begin{align*}
\|(\textbf{u}^{0\mu}_{0},\textbf{u}^{0\mu}_{j},\textbf{u}^{0\mu},\textbf{u}^{lm}_{0},\textbf{u}^{lm}_{j},\textbf{u}^{lm}, \textbf{u} _{0},\textbf{u} _{j},\textbf{u}, \delta\zeta,\textbf{v}^{p}) \|_{H^{k}} \leq C\sigma.
\end{align*}
Furthermore, by \eqref{e:gij}--\eqref{e:dtg00exp}, we obtain
\begin{align*}
	\|g^{\mu\nu}-\eta^{\mu\nu}\|_{H^{k+1}}+ \|\partial_\sigma g^{\mu\nu}-\partial_\sigma \eta^{\mu\nu}\|_{H^k} \leq C\|\mathbf{U}\|_{H^k} \leq C\sigma,
\end{align*}
and the transformations in  Assumption \ref{a:Maksym} and \ref{a:tLip} imply there exists a function $\varrho\in C\bigl([0,1], C^\infty( \mathbb{R})\bigr)$ satisfying $\varrho(\tau,0)=0$
such that
\begin{align*}
\rho-\bar{\rho} =\tau^{\varsigma}  \varrho\bigl(\tau,  \beta^{-1}(\tau)(\alpha- \bar{\alpha})\bigr) ,\quad \varsigma \geq 2.
\end{align*}
Then Moser estimates for composition of functions (see, e.g. \cite[Lemma A.$3$]{Liu2017}) yield
\begin{align*}
	\|\rho-\bar{\rho}\|_{H^k} \leq C\|\beta^{-1}(\alpha-\bar{\alpha} )\|_{H^k} \leq C\sigma,
\end{align*}
and using \eqref{e:NewVar}--\eqref{A}, that is, $u^{p}=\beta(\tau) (\mathbf{v}^p + 2\tau A \textbf{u}^{0p})$, we derive that
\begin{equation*}
	\|u^p\|_{H^k}\leq \beta(\tau) \| \mathbf{v}^p\|_{H^k}+2\tau \beta(\tau) A(\tau) \|\mathbf{u}^{0p} \|_{H^k} \leq C\sigma
\end{equation*}
for $\tau \in(0,1]$. In addition, \eqref{e:uvreg} with above transformations and estimates yield \eqref{e:gvrred}. Then, we complete the proof of the main Theorem \ref{theorem:1.4} for Fluids $(I)$.

\subsubsection{Proof for \textit{Fluids} $(II)$}  \label{s:Mak2}

As the previous section \S \ref{s:Mak1},
Let us first gather the Einstein equations \eqref{4.26}--\eqref{4.28} and the Euler equations \eqref{lower-case} together, recalling the definition of
\begin{align*}
\textbf{U}: =(\textbf{u}^{0\mu}_{0},\textbf{u}^{0\mu}_{j},\textbf{u}^{0\mu},\textbf{u}^{lm}_{0},\textbf{u}^{lm}_{j},\textbf{u}^{lm}, \textbf{u} _{0},\textbf{u} _{j},\textbf{u} )^{T} \AND \hat{\textbf{V}}:=(\delta\zeta, u_q)^{T},
\end{align*}
to get the complete non-degenerated singular symmetric hyperbolic system
\begin{equation}\label{symmetric-Euler1Ma2}
B^{\mu}\partial_{\mu}\p{\textbf{U}\\ \hat{\textbf{V}}}=\frac{1}{\tau}\textbf{BP}\p{\textbf{U}\\ \hat{\textbf{V}}}+H
\end{equation}
where
\begin{gather*}
\textbf{B}=\left(
\begin{array}{cccc}
\textbf{A}&0&0&0\\
0&-2g^{00}\mathbb I&0&0\\
0&0&-2g^{00}\mathbb I&0\\
0 &0&0& \hat{\textbf{N}}
\end{array}
\right), \\
B^{\mu}=\left(
\begin{array}{cccc}
A^{\mu}&0&0&0\\
0&A^{\mu}&0&0\\
0&0&A^{\mu}&0\\
0&0&0& \hat{N}^{\mu}
\end{array}
\right),
\quad
\textbf{P}=\left(
\begin{array}{cccc}
\textbf{P}^\star &0&0&0\\
0&\Pi&0&0\\
0&0&\Pi&0\\
0&0&0&\hat{\textbf{P}}^\dagger
\end{array}
\right)
\end{gather*}
and $H=(F_1,F_2,F_3, \hat{H})^T$. To prove the main theorem in this case, our purpose is still to apply Theorem \ref{pro:3.16} in Appendix \ref{section:3.5}. Thus we need to examine Conditions \eqref{c:1}--\eqref{c:7} for the equation \eqref{symmetric-Euler1Ma2}.

It is easy to verify \textbf{Conditions \eqref{c:1}--\eqref{c:2}} 
as the previous section.

Next, in order to verify \textbf{Condition \eqref{c:3}}, we demonstrate $H$ is regular in all the variables, that is $H\in C^{0}([0,1],C^{\infty}(\mathbb{V}))$, by expanding the following crucial quantities.
Firstly differentiating \eqref{velocity} with respect to $g^{\mu\nu}$, we arrive at
\begin{align*}
	\frac{\partial u^{k}}{\partial(g^{00})} =&g^{k0}\left(\frac{(1+g^{ij}u_{i}u_{j})}{2g^{00}\sqrt{(g^{0i}u_{i})^2-g^{00}(1+g^{ij}u_{i}u_{j})}}
	+\frac{\sqrt{(g^{0i}u_{i})^2-g^{00}(1+g^{ij}u_{i}u_{j})}}{(g^{00})^2}\right)\\  =&\tau\mathcal{M}^{k}_{00}(\tau,\textbf{U},\hat{\textbf{V}}),   \\
	\frac{\partial u^{k}}{\partial (g^{0i})} =&\delta^{k}_{i}\left(\frac{-g^{0q}u_{q}-\sqrt{(g^{0q}u_{q})^2-g^{00}(1+g^{ij}u_{i}u_{j})}}{g^{00}}\right)\\ &+g^{k0}\left(-u_{i}-\frac{g^{0i}u_{i}^2}{g^{00}\sqrt{(g^{0q}u_{q})^2-g^{00}(1+g^{ij}u_{i}u_{j})}}\right)\\  =&-\frac{\delta^{k}_{i}}{\sqrt{-\eta^{00}}}+\mathcal{M}^{k}_{0i}(\tau,\textbf{U},\hat{\textbf{V}}),
	\intertext{and}
	\frac{\partial u^{k}}{\partial (g^{ij})} =&\delta^{k}_{j}u_{i}+\frac{g^{k0}u_{i}u_{j}}{2\sqrt{(g^{0q}u_{q})^2-g^{00}(1+g^{ij}u_{i}u_{j})}} = \mathcal{M}^{k}_{ij}(\tau,\textbf{U},\hat{\textbf{V}}),
\end{align*}
where $\mathcal{M}^{k}_{\alpha\beta}(\tau,\textbf{U},\hat{\textbf{V}})$ agrees with the convention of notation of remainders in \S\ref{remainder} 
and satisfies $\mathcal{M}^{k}_{\alpha\beta}(\tau,0,0)=0$.
Then by using \eqref{e:gij}--\eqref{e:ui} and letting $\beta(\tau)=1$, we easily get
\begin{equation}\label{e:Ldudg}
L_{k}^{\mu}\frac{\partial u^{k}}{\partial (g^{\alpha\beta})}\partial_{\mu}g^{\alpha\beta}
= \mathcal{N}(\tau,\textbf{U},\hat{\textbf{V}}),
\end{equation}
where $\mathcal{N}(\tau,\textbf{U},\hat{\textbf{V}})$ agrees with the convention of notation of remainders 
and satisfies $\mathcal{N}(\tau,0,0)=0$. Next let us calculate terms involving the Christoffel terms.
Direct calculations, with the help of that
\begin{equation*}
u^{0}=g^{00}u_{0}+g^{k0}u_{k}=-\sqrt{(g^{0i}u_{i})^2-g^{00}(1+g^{ij}u_{i}u_{j})},
\end{equation*}
yield
\begin{align}
\Gamma_{\mu\nu}^{k}u^{\mu}u^{\nu}=&\Gamma^{k}_{00}u^{0}u^{0}+2\Gamma_{0i}^{k}u^{0}u^{i}+\Gamma_{lm}^{k}u^{l}u^{m}\nonumber\\
=&\Gamma_{00}^{k}\left((g^{0i}u_{i})^2-g^{00}(1+g^{ij}u_{i}u_{j})\right)\nonumber\\
 &-2\Gamma^{k}_{0i}\sqrt{(g^{0i}u_{i})^2-g^{00}(1+g^{ij}u_{i}u_{j})}(g^{i0}u_{0}+g^{im}u_{m})\nonumber\\
 &+\Gamma_{lm}^{k}(g^{l0}u_{0}+g^{lq}u_{q})(g^{m0}u_{0}+g^{mq}u_{q})\nonumber\\
 =&\mathcal{N}^{k}(\tau,\textbf{U},\hat{\textbf{V}}), \label{e:Guu}
\end{align}
where $\mathcal{N}^{k}(\tau,\textbf{U},\hat{\textbf{V}})$ agrees with the convention of notation of remainders 
and satisfies $\mathcal{N}^{k}(\tau,0,0)=0$. Some direct expansions, with the help of \eqref{e:Mexp}--\eqref{e:ugu}, \eqref{velocity}, \eqref{Jacobi},  \eqref{e:Ldudg}--\eqref{e:Guu}, lead to $H\in C^{0}([0,1],C^{\infty}(\mathbb{V}))$.

As the previous section, in order to verify \textbf{Condition \eqref{c:4}}, the crucial condition is $B^0\in C^1([0,1], C^\infty(\mathbb{V}))$ and we only need to check $\hat{N}^0\in C^1([0,1], C^\infty(\mathbb{V}))$. Using \eqref{e:B0bd} in Assumption \ref{a:postvty}, with the help of \eqref{velocity} which implies $  u^i(\tau,\mathbf{U},\hat{\textbf{V}}) \in C^1([0,1], C^\infty(\mathbb{V})) $,  the similar derivations to the previous section \eqref{e:du0}--\eqref{e:bds} yield, 
\begin{align}
	D_{\tau}u^{0}(\tau,\textbf{U},\hat{\textbf{V}}) =&\frac{1}{\omega}\frac{\Omega}{\tau} +\texttt{S}(\tau,\textbf{U},\hat{\textbf{V}}),  \label{e:Duu0} \\
	D_{\tau}u_{0}(\tau,\textbf{U},\hat{\textbf{V}}) =&\frac{1}{  \omega^3}  \frac{\Omega}{\tau} +\texttt{T}(\tau,\textbf{U},\hat{\textbf{V}}), \\
	D_{\tau}M_{ki}(\tau,\textbf{U},\hat{\textbf{V}}) =&
	\texttt{F}_{ki}(\tau,\textbf{U},\hat{\textbf{V}}),  \\
	D_{\tau} q(\alpha) =& q^\prime(\alpha) \frac{3}{\tau}\bar{q} \equiv 0  \quad(\text{Becasue } q\equiv \bar{q}) ,  \label{e:Dq}
\end{align}
and by $q\equiv \bar{q}$,  Assumption \ref{a:Maksym}.\eqref{e:Maksym} and \ref{a:tLip}.\eqref{e:rhopro}, \eqref{background-solution1}, \eqref{e:p-p}, \eqref{e:bgrhoest1} (the bounds of $\bar{\rho}(\tau)$), and
\begin{align}
	(s^2)^\prime (\alpha) = \frac{d}{d\alpha} \mu^*\bigl( dp/d\rho\bigr)=\frac{d^2p }{d\rho^2}\Bigr|_{\mu(\alpha)}\frac{d\mu}{d\alpha}\Bigr|_{\alpha}
\end{align}
we arrive at
\begin{align}\label{e:Ds2}	
	D_{\tau} & s^2(\alpha)  =   (s^2)'(\alpha) \frac{3}{\tau} \bar{q}   =p''\bigr|_{\mu(\alpha)}\frac{d\mu}{d\alpha}\Bigr|_{\alpha} \frac{3}{\tau} \bar{q} =p''\bigr|_{\mu(\alpha)}\frac{\mu(\alpha)+\mu^*p(\alpha)}{q(\alpha)}  \frac{3}{\tau} \bar{q}   \notag  \\
	= & p''\bigr|_{\mu(\alpha)}\frac{\mu(\alpha)-\mu(\bar{\alpha})+\mu^*p(\alpha)-\mu^*p(\bar{\alpha})}{q(\alpha)} \frac{3}{\tau} \bar{q} +p''\bigr|_{\mu(\alpha)}\frac{\mu(\bar{\alpha})+\mu^*p(\bar{\alpha})}{q(\alpha)}  \frac{3}{\tau} \bar{q}\notag  \\
	= & 3 p''\bigr|_{\mu(\alpha)} (1+c_s^2(\rho_{K_4}))\tau^{\varsigma-1}  \varrho\bigl(\tau,  \beta^{-1}(\tau)(\alpha- \bar{\alpha})\bigr)    +p''\bigr|_{\mu(\alpha)} \bigl( \bar{\rho}+ \bar{p}(\rho) \bigr) \frac{3}{\tau} \lesssim 1,
\end{align}
for $\varsigma\geq 2$ and $\tau\in[0,1]$.
Note that \eqref{Jacobi} implies
\begin{equation}\label{e:Dj}
	D_\tau J^{iq}(\tau,\textbf{u}^{0\nu},u_{k})=\texttt{J}^{iq}(\tau,\textbf{u}^{0\nu},u_{k}),
\end{equation}
in other words, $J^{ij}\in C^{1}([0,1],C^{\infty}( \mathbb{V}))$.
Assisted with above preparations  \eqref{e:Duu0}--\eqref{e:Dj}, we can directly expand $D_\tau \hat{N}^0$ and then it is evident to conclude $B^0\in C^1([0,1], C^\infty(\mathbb{V}))$. The other part of this condition is easy to be examined and we omit the details.

Then, let us turn to \textbf{Condition \eqref{c:5}}. By \eqref{e:cir0}, to verify \eqref{e:Bineq}, as before, we only concern
\begin{align*}
Q=\frac{1}{\hat{\kappa}}\ring{\textbf{B}}-\ring{B}^0=\cir{\Bigl[\frac{1}{\hat{\kappa}} \textbf{B} - B^0\Bigr]}=\frac{1}{\hat{\kappa}} \textbf{B}(\tau,0,0) - B^0(\tau,0,0),
\end{align*}
and to verify $Q$ is positive definite. We focus on one element of $Q$ that is $
	 \frac{1}{\hat{\kappa}} \hat{\textbf{N}}(\tau,0,0) - \hat{N}^0(\tau,0,0)  $.
There are two cases to proceed:
\begin{enumerate}
	\item If $\hat{\textbf{N}}$ and $\hat{N}^0$ are defined by \eqref{e:case1}, then, with the help of \eqref{e:Mexp}, \eqref{e:expL},  \eqref{velocity} and  \eqref{Jacobi} (Note that $u_q=0$ and $u^{0i}=0$ imply $v^i=0$ and further $\text{v}^i=0$ by \eqref{e:NewVar} and \eqref{velocity}),
	\begin{align*}
	 \frac{1}{\hat{\kappa}} \hat{\textbf{N}}(\tau,0,0) - \hat{N}^0(\tau,0,0)
	= & \p{\bigl(\frac{1}{\hat{\kappa}}-1\bigr)  & 0  \\
	0  & \frac{1}{\bar{\lambda}^2} \bigl(\frac{1}{\hat{\kappa}}(1-3\bar{s}^2)-1\bigr)\delta^{iq}   } . 
	\end{align*}
		\item If $\hat{\textbf{N}}$ and $\hat{N}^0$ are defined by \eqref{e:case1b}, then,
	\begin{align*}
	 \frac{1}{\hat{\kappa}} \hat{\textbf{N}}(\tau,0,0) -  \hat{N}^0(\tau,0,0)
	= & \p{\bigl(\frac{1}{\hat{\kappa}}-1\bigr)   & 0  \\
		0  & \frac{1}{\bar{\lambda}^2}\bigl(\frac{1}{\hat{\kappa}} -1\bigr)\delta^{iq} } . 
	\end{align*}
\end{enumerate}
By taking
\begin{equation*}
	\hat{\kappa} :=  \Bigl(1+\sqrt{\frac{3}{\Lambda}}\Bigr)^{-1} \hat{\delta},
\end{equation*}
above $	 \frac{1}{\hat{\kappa}} \hat{\textbf{N}}(\tau,0,0) - \hat{N}^0(\tau,0,0)$, by noting $\hat{\kappa}<3/4$ in \eqref{e:hdel}, with the help of \eqref{e:Maksym} and Assumption \ref{a:postvty}.$(2)$, is positive definite in both cases. This, with the help of \eqref{e:A}, \eqref{e:A-I},  \eqref{e:AAI} and \eqref{e:gAI} and the derivations in the previous section \S \ref{s:Mak1}, implies $Q$ is positive definite, which, in turn, implies \textbf{Condition \eqref{c:5}}.

For \textbf{Conditions \eqref{c:6} and \eqref{c:7}}, 
if $\hat{\textbf{N}}$ and $\hat{\textbf{P}}^\dagger$ are defined by \eqref{e:case1b}, Then Conditions \eqref{c:6} and \eqref{c:7} hold evidently due to $\hat{\textbf{P}}^\dagger\equiv 0$.
We only need to examine the case when $\hat{\textbf{N}}$ and $\hat{\textbf{P}}^\dagger$ are defined by \eqref{e:case1}.
In this case,
\begin{equation}\label{e:PUV}
	\textbf{P}  (\textbf{U}, \hat{\textbf{V}})^T=\Bigl(\frac{1}{2}\textbf{u}^{0\mu}_0+\frac{1}{2}\textbf{u}^{0\mu}, \textbf{u}^{0\mu}_l, \frac{1}{2}\textbf{u}^{0\mu}_0+\frac{1}{2}\textbf{u}^{0\mu}, \textbf{u}^{lm}_0, 0,0,\textbf{u}_0, 0,0, 0, u_i\Bigr)^T.
\end{equation}
Verifying Condition \eqref{c:6} is equivalent to verifying the following conditions,
\begin{align*}
	(\textbf{P}^\star)^\perp A^0 (\tau,\textbf{P}^\perp (\textbf{U}, \hat{\textbf{V}})^T) \textbf{P}^\star=& \textbf{P}^\star  A^0(\tau,\textbf{P}^\perp (\textbf{U}, \hat{\textbf{V}})^T) (\textbf{P}^\star)^\perp=0,  
	\\
	\Pi^\perp A^0(\tau,\textbf{P}^\perp (\textbf{U}, \hat{\textbf{V}})^T)\Pi=&\Pi A^0(\tau,\textbf{P}^\perp (\textbf{U}, \hat{\textbf{V}})^T) \Pi^\perp=0,  
	\\
	(\hat{\textbf{P}}^\dagger)^\perp N^0 (\tau,\textbf{P}^\perp (\textbf{U}, \hat{\textbf{V}})^T) \hat{\textbf{P}}^\dagger=& \hat{\textbf{P}}^\dagger  N^0(\tau,\textbf{P}^\perp (\textbf{U}, \hat{\textbf{V}})^T) (\hat{\textbf{P}}^\dagger)^\perp=0 . 
\end{align*}
To obtain $A^0 (\tau,\textbf{P}^\perp (\textbf{U}, \hat{\textbf{V}})^T)$ and $N^0 (\tau,\textbf{P}^\perp (\textbf{U}, \hat{\textbf{V}})^T)$, we set $\textbf{P}  (\textbf{U}, \hat{\textbf{V}})^T=0$ in the variables $(\textbf{U}, \hat{\textbf{V}})^T$ of $A^0$ and $N^0$. Direct calculations yield
\begin{align}
	& (\textbf{P}^\star)^\perp A^0 (\tau,\textbf{P}^\perp (\textbf{U}, \hat{\textbf{V}})^T) \textbf{P}^\star = \textbf{P}^\star  A^0(\tau,\textbf{P}^\perp (\textbf{U}, \hat{\textbf{V}})^T) (\textbf{P}^\star)^\perp \notag  \\
	= & \p{\frac{1}{2} & 0 & -\frac{1}{2}\\
    0 & 0 & 0 \\
    -\frac{1}{2} & 0 & \frac{1}{2} } \left.\left(
	\begin{array}{ccc}
	-g^{00}  & 0& 0\\
	0& g^{ij}& 0\\
	0&0& -g^{00}
	\end{array}\right)\right|_{\textbf{P}  (\textbf{U}, \hat{\textbf{V}})^T=0} \p{\frac{1}{2} & 0 & \frac{1}{2}\\
		0 & \delta^j_l & 0 \\
		\frac{1}{2} & 0 & \frac{1}{2} }=0,  \\
	& \Pi^\perp A^0(\tau,\textbf{P}^\perp (\textbf{U}, \hat{\textbf{V}})^T)\Pi=\Pi A^0(\tau,\textbf{P}^\perp (\textbf{U}, \hat{\textbf{V}})^T) \Pi^\perp  \notag  \\
	= & \p{0 & 0 & 0\\
	0 & 1 & 0 \\
	0 & 0 & 1 } \left.\left(
\begin{array}{ccc}
-g^{00}  & 0& 0\\
0& g^{ij}& 0\\
0&0& -g^{00}
\end{array}\right)\right|_{\textbf{P}  (\textbf{U}, \hat{\textbf{V}})^T=0} \p{1 & 0 & 0\\
	0 & 0 & 0 \\
	0 & 0 & 0 }=0
\intertext{and}
	& (\hat{\textbf{P}}^\dagger)^\perp \hat{N}^0 (\tau,\textbf{P}^\perp (\textbf{U}, \hat{\textbf{V}})^T) \hat{\textbf{P}}^\dagger= \hat{\textbf{P}}^\dagger  \hat{N}^0(\tau,\textbf{P}^\perp (\textbf{U}, \hat{\textbf{V}})^T) (\hat{\textbf{P}}^\dagger)^\perp \notag \\
	= & \p{1 & 0 \\0 & 0} \left.\left(\begin{array}{cc}
	1 & 0 \\
	0 & \frac{1}{\lambda^2} M_{ki}J^{kj}J^{iq}
	\end{array}
	\right)\right|_{\textbf{P}  (\textbf{U}, \hat{\textbf{V}})^T=0} \left(\begin{array}{cc}
	0&0\\
	0&\delta^{q}_{i}
	\end{array}
	\right)=0 . \label{e:PDP}
\end{align}
Note that we have set $u_i=0$ (from \eqref{e:PUV} and $\textbf{P}  (\textbf{U}, \hat{\textbf{V}})^T=0$) in \eqref{e:PDP}, which leads to $L^0_i=0$ in $\hat{N}^0$. This completes the verification of \textbf{Condition \eqref{c:6}}.

At the end, the examination of \textbf{Condition \eqref{c:7}} is the same as the previous work in \cite[\S $3$]{Oliynyk2016a} and \cite[\S $7$]{Liu2017}. To avoid repeating the examination, we only note a crucial identity in the proof, that is
\begin{align} \label{e:PDDP}
	(\hat{\textbf{P}}^{\dagger})^\perp[D_{(\textbf{U}, \hat{\textbf{V}}) }\hat{N}^0(\tau,\textbf{U}, \hat{\textbf{V}})\textbf{W}](\hat{\textbf{P}}^{\dagger})^\perp\equiv 0 ,
\end{align}
for any $\textbf{W} \in \mathbb{V}$. The same derivations, with the help of the identity \eqref{e:PDDP}, conclude \textbf{Condition \eqref{c:7}}.

Having verified that all of the hypotheses of Theorem \ref{pro:3.16} are satisfied, we can conclude the  main Theorem \ref{theorem:1.4} via the similar arguments to the previous section \S\ref{s:Mak1}, but only noting that, in this case, the invertible transformation between $u_q$ and $u^i$ is given by \eqref{e:u-q} and \eqref{velocity}, and then they can be controlled by each other with the gravitational variables,
\begin{align*}
	\|u_q\|_{H^k}\lesssim \|(\mathbf{U},u^i)\|_{H^k} \AND \|u^i\|_{H^k} \lesssim \|(\mathbf{U},u_q)\|_{H^k}.
\end{align*}
We omit the reduplicate details and complete the proof of the main Theorem \ref{theorem:1.4} for the Fluids $(II)$.

\appendix


\section{A class of symmetric hyperbolic systems}\label{section:3.5}

In this Appendix, we introduce the main tool for this article which is a variation of the theorem originally established in \cite[Appendix B]{Oliynyk2016a}. The proof of it has been omit, but readers can see the details\footnote{A minor revision and improvement about the condition \eqref{c:7} has been included in \href{https://arxiv.org/abs/1505.00857}{arXiv:1505.00857v4} and only \eqref{e:A3} is necessary for our case. An alternative expression of this condition is given in \cite{Liu2017,Liu2018}. } in \cite{Oliynyk2016a}, and its generalizations in \cite{Liu2017,Liu2018}.

Consider the following symmetric hyperbolic system.
\begin{align}
B^{\mu}\partial_{\mu}u =&\frac{1}{t}\textbf{BP}u+H\quad\text{in}\;[T_{0},T_{1}]\times\mathbb{T}^{n},  \label{e:model1}\\
u =&u_{0}\quad\quad\quad\qquad\text{in}\;{T_{0}}\times\mathbb T^{n},\label{e:model2}
\end{align}
where we require the following \textbf{Conditions}:
\begin{enumerate}[(I)]
\item \label{c:1} $T_{0}<T_{1}\leq0$.

\item \label{c:2} $\textbf{P}$ is a constant, symmetric projection operator, i.e., $\textbf{P}^{2}=\textbf{P}$, $\textbf{P}^{T}=\textbf{P}$ and $\partial_\mu \textbf{P}=0$.

\item \label{c:3} $u=u(t,x)$ and $H(t,u)$ are $\mathbb R^{N}$-valued maps, $H\in C^{0}([T_{0},0],C^{\infty}(\mathbb R^{N}))$ and satisfies $H(t,0)=0$.

\item \label{c:4} $B^{\mu}=B^{\mu}(t,u)$ and $\textbf{B}=\textbf{B}(t,u)$ are $\mathbb M_{N\times N}$-valued maps, and $B^{\mu},\,\textbf{B}\in C^{0}([T_{0},0],C^{\infty}( \mathbb R^{N}))$, $B^{0}\in C^{1}([T_{0},0],C^{\infty}( \mathbb R^{N}))$ and they satisfy
\begin{align}\label{e:comBP}
	(B^{\mu})^{T}=B^{\mu},\quad [\textbf{P}, \textbf{B}]=\textbf{PB}-\textbf{BP}=0.
\end{align}

\item \label{c:5} Suppose\footnote{This variation of the original condition (v) and (B.3) in \cite{Oliynyk2016a} facilitates the examinations of the conditions of this theorem, and the proof is easy to be recovered by minor corrections. Note that the $\tau$-singular terms caused by $\tilde{B}^0 (t,u)$ and $\tilde{\textbf{B}} (t,u)$ can be absorbed into the principle singular term with the good sign. We give a Remark \ref{R:A2} on the key revisions in the proof. See the proof in \cite{Oliynyk2016a} or more details in \cite[\S $5$]{Liu2017} involving such variations. }
\begin{align}
	B^0=&\ring{B}^0(t)+\tilde{B}^0 (t,u)  \label{e:B0exp}
	\intertext{and}
	\textbf{B}=&\ring{\textbf{B}}(t) +\tilde{\textbf{B}} (t,u) \label{e:bfBexp}
\end{align}
where $\tilde{B}^0 (t,0)=0$ and $\tilde{\textbf{B}} (t,0)=0$.
There exists constants $\kappa,\,\gamma_{1},\,\gamma_{2}$ such that
\begin{align}
	\frac{1}{\gamma_{1}}\mathbb I\leq \ring{B}^{0}\leq \frac{1}{\kappa}\ring{\textbf{B}}\leq\gamma_{2}\mathbb I \label{e:Bineq}
\end{align}
for all $ t \in[T_{0},0] $.

\item \label{c:6} For all $(t,u)\in[T_{0},0]\times\mathbb R^{N}$, we have
$$
\textbf{P}^{\bot}B^{0}(t,\textbf{P}^\perp u)\textbf{P}=\textbf{P}B^{0}(t,\textbf{P}^\perp u)\textbf{P}^{\bot}=0,
$$
where
$
\textbf{P}^{\bot}=\mathbb I-\textbf{P}
$
is the complementary projection operator.

\item \label{c:7} There exists constants $\varsigma,\,\beta_{1}$ and $\varpi>0$ such that
\begin{align} \label{e:A3}
|\textbf{P}^{\bot}[D_{u}B^{0}(t,u)(B^{0})^{-1}\textbf{BP}u]\textbf{P}^{\bot}|_{op} \leq& |t|\varsigma+\frac{2\beta_{1}}{\varpi+|P^{\bot}u|^2}|\textbf{P}u|^{2}. 
\end{align}

\end{enumerate}

\begin{theorem}\label{pro:3.16}
	Suppose that $k\geq \frac{n}{2}+1$, $u_{0}\in H^{k}(\mathbb T^{n})$ and conditions \eqref{c:1}--\eqref{c:7} are fulfilled. Then there exists a $T_{\ast}\in (T_{0},0)$, and a unique classical solution $u\in C^{1}([T_{0},T_{\ast}]\times\mathbb T^{n})$ that satisfies $u\in C^{0}([T_{0},T_{\ast}],H^{k})\cap C^{1}([T_{0},T_{\ast}],H^{k-1})$ and the energy estimate
	$$
	\|u(t)\|_{H^{k}}^{2}-\int_{T_{0}}^{t}\frac{1}{\tau}\|\textbf{P}u\|_{H^{k}}^{2}d\tau\leq
	Ce^{C(t-T_{0})}(\|u(T_{0})\|_{H^{k}}^{2})
	$$
	for all $T_{0}\leq t<T_{\ast}$, where
	$
	C=C(\|u\|_{L^{\infty}([T_{0},T_{\ast}),H^{k})},\gamma_{1},\gamma_{2},\kappa),
	$
	and can be uniquely continued to a larger time interval $[T_{0},T^{\ast})$ for all $T^{\ast}\in(T_{\ast},0]$ provided $\|u\|_{L^{\infty}([T_{0},T_{\ast}),W^{1,\infty})}<\infty$.
\end{theorem}
Let us end this Appendix with a remark of the proof although we omit the detailed proof.
	\begin{remark}\label{R:A2}
		We give the key revision in the proof due to the changing of Condition \eqref{c:5}, (i.e. \eqref{e:B0exp}--\eqref{e:Bineq}). As in \cite[Page $2203$]{Liu2017}, we rewriting $\ring{B}^0 $ as $\ring{B}^0 =(\ring{B} ^0 )^{\frac{1}{2}}(\ring{B}^0 )^{\frac{1}{2}}$, which can be done since
		$\ring{B}^0 $ is a real symmetric and positive-definite. We see from \eqref{e:Bineq} that
		\begin{align} \label{Afrbound}
		(\ring{B} ^0)^{-\frac{1}{2}}\ring{\textbf{B}}(\ring{B}^0)^{-\frac{1}{2}}\geq \kappa\mathds{1}.
		\end{align}
		Since, by \eqref{e:comBP}--\eqref{e:bfBexp}
		\begin{align*}
		\frac{2}{t} \langle D^\alpha u, \textbf{B}  D^\alpha \mathbf{P}  u\rangle=&\frac{2}{t} \langle D^\alpha \mathbf{P} u, (\ring{B}^0)^{\frac{1}{2}}[(\ring{B}^0)^{-\frac{1}{2}}\ring{\textbf{B}} (\ring{B}^0)^{-\frac{1}{2}}](\ring{B}^0)^{\frac{1}{2}} D^\alpha \mathbf{P}  u\rangle +\frac{2}{t} \langle D^\alpha\mathbf{P} u, \tilde{\textbf{B}}  D^\alpha \mathbf{P}  u\rangle  \notag  \\
		\leq & \frac{2\kappa}{t} \langle D^\alpha \mathbf{P} u,  \ring{B}^0  D^\alpha \mathbf{P}  u\rangle +\frac{2}{t} \langle D^\alpha\mathbf{P} u, \tilde{\textbf{B}}  D^\alpha \mathbf{P}  u\rangle  \notag  \\
		=& \frac{2\kappa}{t} \langle D^\alpha \mathbf{P} u,  B^0  D^\alpha \mathbf{P}  u\rangle-\frac{2\kappa}{t} \langle D^\alpha \mathbf{P} u,  \tilde{B}^0  D^\alpha \mathbf{P}  u\rangle +\frac{2}{t} \langle D^\alpha\mathbf{P} u, \tilde{\textbf{B}}  D^\alpha \mathbf{P}  u\rangle.
		\end{align*}
		It follows immediately from \eqref{Afrbound} that
		\begin{align} \label{E:KAPPACONTR}
		\frac{2}{t}\sum_{0\leq |\alpha|\leq s-1}\langle D^\alpha u,\textbf{B}  D^\alpha \mathbf{P}   u\rangle \leq  \underbrace{\frac{2\kappa}{t}\vertiii{\mathbf{P}  u}^2_{H^{s-1}}}_{\text{I}}-\underbrace{\frac{1}{t} C \|u\|_{H^{s-1}} \|\textbf{P}u\|^2_{H^{s-1}}}_{\text{II}},
		\end{align}
		where
		\begin{align*}
			\vertiii{u}^2_{H^k}:=\sum_{0\leq |\alpha|\leq k}\langle D^\alpha u, B^0\bigl(t,u(t,\cdot)\bigr) D^\alpha u\rangle.
		\end{align*}
	Then, the following proof are the same as \cite[Appendix B]{Oliynyk2016a} or \cite[\S $5$]{Liu2017} just noting that the term II in \eqref{E:KAPPACONTR} can be absorbed by I  in the rest of estimates provided the data is small enough.
	\end{remark}


\section*{Acknowledgement}
We would like to thank Prof. Todd A. Oliynyk for helpful discussions. The first author thanks Prof. Uwe Brauer for his useful comments and is partially supported by the China Postdoctoral Science Foundation Grant under the grant No. 2018M641054, the Fundamental Research Funds for the Central Universities, HUST: 5003011036 and the National Natural Science Foundation of China (NSFC) under the grant No. 11971503. The second author is grateful to the financial support from Monash University during his stay, and is partially supported by the National Natural Science Foundation of China (NSFC) under the grant Nos. 12071435,11701517,11871212, the Natural Science Foundation of Zhejiang Province under the granr No. LY20A010026 and the Fundamental Research Funds of Zhejiang Sci-Tech University under the grant No. 2020Q037. We also thank the referee for their comments and criticisms, which have served to improve the content and exposition of this article.


\bibliographystyle{amsplain}
\bibliography{FLRWgegas}
\end{document}